\documentclass[aps,twocolumn,superscriptaddress,showkeywords]{revtex4-2}

\usepackage[T1]{fontenc}
\usepackage{lmodern}
\usepackage{microtype}
\usepackage{amsmath,amssymb,amsfonts}
\usepackage{amsthm}
\usepackage{bm}
\usepackage{graphicx}
\usepackage{color}
\usepackage{hyperref}
\usepackage{cleveref}
\usepackage{physics}
\usepackage{mathtools}
\hypersetup{
 colorlinks=true,
 linkcolor=blue,
 citecolor=blue,
 urlcolor=blue
}
\theoremstyle{plain}
\newtheorem{theorem}{Theorem}
\newtheorem{lemma}{Lemma}
\newtheorem{corollary}{Corollary}
\newtheorem{proposition}{Proposition}

\theoremstyle{definition}
\newtheorem{definition}{Definition}
\theoremstyle{remark}
\newtheorem{remark}{Remark}

\newcommand{\Cov}{\operatorname{Cov}}

\DeclareMathOperator{\supp}{supp}
\DeclareMathOperator{\wt}{wt}

\begin{document}

\title{Covariance Spectra of Classical Shadows: Finite-Sample Recovery for Selected Observables}

\author{Masahito Hayashi}
\email{hmasahito@cuhk.edu.cn}
\affiliation{School of Data Science, The Chinese University of Hong Kong, Shenzhen, Longgang District, Shenzhen, 518172, China}
\affiliation{International Quantum Academy, Futian District, Shenzhen 518048, China}
\affiliation{Graduate School of Mathematics, Nagoya University, Nagoya, 464-8602, Japan}
\date{\today}

\begin{abstract}
Classical shadows use a single measurement snapshot to estimate many observables. These estimates are jointly distributed, and their covariance is not determined by their expectation values or individual variances. It also differs from physical correlations between observables because it depends on the measurement and reconstruction protocol.

We study the covariance matrix for a fixed set of shadow estimates and its spectrum. For general shadow protocols, we give a finite-sample operator-norm bound for the empirical covariance computed using the sample mean. This bound also controls the eigenvalues and isolated spectral subspaces of the covariance matrix.

For local product shadow protocols with fixed local dimension, observables supported on a bounded number of sites lead to bounds that do not depend on the total system size, provided that the number of selected observables and the relevant local reconstruction factors remain bounded. For biased local Pauli shadows, we obtain explicit finite-sample bounds and an exact covariance formula determined by Pauli compatibility and local basis probabilities. A comparison with global Clifford shadows shows that dimension-independent covariance estimation depends on the measurement protocol and does not follow from the general theory alone.
\end{abstract}


\maketitle

\section{Introduction}
\label{sec:introduction}
Classical shadows use each measurement snapshot to estimate many observables. Because these estimates come from the same snapshot, they form a jointly distributed random vector. Most classical-shadow results describe expectation values, individual variances, or the simultaneous estimation of many expectation values. These results do not, however, describe the covariances between estimates obtained from the same snapshot. This paper studies that joint fluctuation structure through the covariance matrix of the shadow-output vector.

Unbiased reconstruction fixes the mean of each shadow estimate, but it does not fix the covariance between different estimates. Two shadow protocols can therefore give the same expectation values for a state and a set of observables while producing different covariance matrices. This covariance depends not only on the state and observables, but also on the measurement and reconstruction protocol. It is different from a physical correlation between observables, which is determined by the state and observables alone. The shadow-output covariance instead describes how estimates produced by the protocol fluctuate together.

We focus on the covariance matrix of a fixed set of shadow estimates, where the set is chosen before the data are observed. This matrix describes the variance of every linear combination of the selected estimates. Its eigenvalues give the fluctuation scales along its principal directions, and its eigenvectors identify those directions. We study how accurately this covariance matrix and its spectrum can be recovered from independent shadow snapshots.

For general shadow protocols, we derive a finite-sample bound for the empirical covariance of the selected estimates. The empirical covariance is centered using the sample mean and can therefore be computed directly from the observed data. Our bound controls its error in operator norm and depends on the number of selected estimates and on a bound for the size of their centered one-shot output. It also gives error bounds for the eigenvalues and for isolated spectral subspaces. When the number of selected estimates and the centered one-shot output size remain bounded as the system grows, the required sample size is independent of the ambient system size.

We next apply the general result to local product shadow protocols with fixed local dimension. A product observable with bounded support depends only on the reconstructed local snapshots within that support. Therefore, if the number of selected observables, their support sizes, and the relevant local reconstruction factors remain bounded, the selected one-shot output size does not grow with the total number of sites. The general theorem then yields a bound independent of the total system size. This conclusion is not limited to Pauli measurements.

\begin{table*}[t]
\centering
\small
\caption{
Positioning of the present work relative to representative areas of the
classical-shadow literature, organized by the main object of analysis.
}
\label{tab:related_work_positioning}
\begin{tabular}{p{3.2cm} p{4.5cm} p{6.4cm}}
\hline
\textbf{Area} & \textbf{Main object} & \textbf{Relation to this work} \\
\hline
Foundational and observable-wise shadows
\cite{HuangKuengPreskill2020,ElbenHuangKuengPreskill2023}
& Expectation values, shadow norms, marginal variances, and guarantees for many observables.
& We study covariances among estimates reconstructed from the same snapshot and the spectra of fixed selected covariance matrices. \\
\hline
General shadow frameworks
\cite{BuKohGarciaJaffe2024,InnocentiLorenzoPalmisanoAlbarelliFerraroPaternostroPalma2023}
& Measurement channels, reconstruction maps, and general measurement frames.
& We use these frameworks to define and estimate the covariance matrix of a shadow-output vector. \\
\hline
Local and biased Pauli shadows
\cite{HuangKuengPreskill2020,HadfieldBravyiRaymondMezzacapo2022}
& Observable-wise moments, Hamiltonian estimation, and variance reduction under local Pauli measurements.
& We give selected finite-sample covariance bounds and extend the local-Pauli second-moment mechanism to an exact covariance formula for biased basis choices. \\
\hline
Other shadow settings
\cite{IppolitiLiRakovszkyKhemani2023,ZhouZhang2024MBLShadows,ChenYuZengFlammia2021,AliakbarpourBravermanChiaLinLiuOufkirShen2025,WestSauvageSenForestanoWierichsKilloranGrinkoCerezoLarocca2026,GandhariAlbertGerritsTaylorGullans2024}
& Shallow or dynamics-informed measurements, robustness, adversarial settings, broader group representations, and continuous-variable systems.
& These works broaden the protocols or statistical models; our focus is the covariance and spectrum of a fixed set of reconstructed outputs. \\
\hline
\textbf{Present work}
& Fixed selected shadow-output covariance matrices and their spectra.
& Finite-sample operator-norm and spectral bounds, system-size-independent local-product regimes, and explicit biased local Pauli covariance formulas. \\
\hline
\end{tabular}
\end{table*}

We study biased local Pauli shadows as an explicit example. For a fixed set of Pauli observables, the finite-sample bound can be written directly in terms of their supports and the local probabilities of choosing each Pauli basis. If the number and weights of the selected Pauli observables are bounded and all basis probabilities are bounded away from zero, the resulting bound is independent of the total number of qubits. We also derive an exact formula for the covariance between every pair of reconstructed Pauli estimates. The formula separates the role of the state from the effects of Pauli compatibility and the local basis probabilities on individual variances and covariances between different estimates. In contrast, the state-uniform selected-output bound used in our analysis of global Clifford shadows grows with the Hilbert-space dimension. Thus, the favorable local scaling does not follow from the general theorem alone.

This work complements the standard theory of classical shadows, which focuses mainly on expectation-value estimation, shadow norms, marginal variances, and guarantees for estimating many observables \cite{HuangKuengPreskill2020,ElbenHuangKuengPreskill2023}. General shadow formalisms provide broad descriptions of measurement and reconstruction protocols \cite{BuKohGarciaJaffe2024,InnocentiLorenzoPalmisanoAlbarelliFerraroPaternostroPalma2023}. We use these frameworks to study a different statistical object: the covariance matrix of estimates reconstructed from the same snapshot. Biased local Pauli measurements were previously developed to improve Hamiltonian expectation estimation \cite{HadfieldBravyiRaymondMezzacapo2022}. For uniform local Pauli measurements, the second-moment mechanism underlying our covariance formula appears in the observable-wise analysis of Huang, Kueng, and Preskill \cite[Lemma~4]{HuangKuengPreskill2020}. We formulate this mechanism as a covariance-matrix relation and extend it to biased basis choices. Table~\ref{tab:related_work_positioning} summarizes these distinctions by the main object of analysis.

Covariances estimated from classical-shadow data also appear in the CoVaR method for variational quantum algorithms \cite{BoydKoczor2022CoVaR}. CoVaR uses covariances between a Hamiltonian and an operator pool as conditions for finding eigenstates. Our object and goal are different: we study the covariance matrix of the reconstructed shadow estimates themselves and give finite-sample guarantees for a fixed part of its spectrum. The table does not include CoVaR as a shadow-protocol category because the distinction concerns the role of covariance rather than the measurement framework.

The paper is organized as follows. Section~\ref{sec:general_covariance} defines the shadow-output covariance, distinguishes it from physical observable correlation, and introduces fixed selected covariance spectra. Section~\ref{sec:general_selected_covariance_spectral_recovery} proves the finite-sample operator-norm and spectral bounds. Section~\ref{sec:local_pauli_exact_covariance} treats biased local Pauli shadows, derives the exact covariance formula, and compares the selected-output bound with global Clifford shadows. Section~\ref{sec:general_local_product_shadows} gives the general local product extension, and Section~\ref{sec:discussion} discusses the scope and limitations of the results.

\section{Beyond expectation values: the missing joint statistical layer of classical shadows}
\label{sec:general_covariance}
\subsection{From a common measurement outcome to jointly reconstructed estimators}
\label{subsec:general_finite_dimensional_shadow_channel}
A defining feature of classical shadows is that one measurement snapshot is  
reused to estimate many observables. Yet established theory has primarily  
described these outputs through expectation values, shadow norms, marginal  
variances, and guarantees for many means. This leaves their protocol-induced  
statistical dependence outside the usual first-moment description. We isolate  
that missing second-order layer as the shadow-output covariance, using a  
finite randomized-POVM formulation that includes the protocols studied later.  
Let
$  \mathcal H
  =
  \bigotimes_{j=1}^k \mathcal H_j $
be a finite-dimensional multipartite Hilbert space.  We describe a single
round of a general randomized local measurement protocol.  For each site
\(j=1,\ldots,k\), let \(\mathcal U_j\) be a finite set of local measurement
settings.  For each \(u_j\in\mathcal U_j\), let
$  M^{(j)}_{u_j}
  =
  \{M^{(j)}_{b_j|u_j}\}_{b_j\in\mathcal B_{j,u_j}} 
$ be a POVM on \(\mathcal H_j\), so that
\begin{equation}
  M^{(j)}_{b_j|u_j}\succeq 0,
  \qquad
  \sum_{b_j\in\mathcal B_{j,u_j}}
  M^{(j)}_{b_j|u_j}
  =
  I_j .
  \label{eq:general_local_povm_conditions}
\end{equation}

The notation introduced so far is local to a single tensor factor
\(\mathcal H_j\). For reference, Table~\ref{tab:notation-local-site}
summarizes the local measurement notation before we pass to the global
product measurement.
\begin{table}[t]
\centering
\small
\caption{Notation on the local Hilbert space \(\mathcal H_j\).}
\label{tab:notation-local-site}
\begin{tabular}{p{1.8cm} p{5.2cm}}
\hline
\textbf{Notation} & \textbf{Meaning} \\
\hline
\(\mathcal H_j\)
&
Local Hilbert space at site \(j\).
\\[1mm]

\(j\)
&
Site index, ranging over \(1,\ldots,k\).
\\[1mm]
\(\mathcal U_j\)
&
Finite set of local measurement settings available at site \(j\).
\\[1mm]

\(u_j\)
&
A local measurement setting chosen from \(\mathcal U_j\).
\\[1mm]

\(\mathcal B_{j,u_j}\)
&
Outcome set for the local POVM at site \(j\) under setting \(u_j\).
\\[1mm]

\(b_j\)
&
A local measurement outcome in \(\mathcal B_{j,u_j}\).
\\[1mm]

\(M^{(j)}_{b_j|u_j}\)
&
Local POVM element on \(\mathcal H_j\) corresponding to outcome \(b_j\)
under setting \(u_j\).
\\[1mm]

\(I_j\)
&
Identity operator on \(\mathcal H_j\).
\\
\hline
\end{tabular}
\end{table}

Set
\begin{equation}
  \mathcal U
  :=
  \mathcal U_1\times\cdots\times\mathcal U_k .
  \label{eq:general_setting_space}
\end{equation}
In one round of the protocol, a random measurement setting
\begin{equation}
  \hat u
  =
  (\hat u_1,\ldots,\hat u_k)
  \in\mathcal U .
  \label{eq:general_random_measurement_setting}
\end{equation}
is drawn according to a prescribed probability distribution \(q\) on
\(\mathcal U\).  We do not assume here that \(q\) is a product distribution.
Conditioned on \(\hat u=u=(u_1,\ldots,u_k)\), the product POVM
\begin{equation}
  M_u
  =
  \{M_{b|u}\}_{b\in\mathcal B_u}.
  \label{eq:general_product_povm_family}
\end{equation}
is measured, where
\begin{align}
  \mathcal B_u
  &:=
  \mathcal B_{1,u_1}\times\cdots\times\mathcal B_{k,u_k},
  \label{eq:general_product_outcome_space}
  \\
  M_{b|u}
  &:=
  \bigotimes_{j=1}^k
  M^{(j)}_{b_j|u_j},
  \qquad
  b=(b_1,\ldots,b_k)\in\mathcal B_u .
  \label{eq:general_product_povm_element}
\end{align}
for \(b=(b_1,\ldots,b_k)\in\mathcal B_u\).  Thus, conditioned on
\(\hat u=u\), the measurement outcome
\(
  \hat b=(\hat b_1,\ldots,\hat b_k)\in\mathcal B_u
\)
is obtained with probability
\begin{equation}
  \mathbb P_\rho(\hat b=b\mid \hat u=u)
  =
  \operatorname{Tr}(\rho M_{b|u}).
  \label{eq:general_conditional_outcome_probability}
\end{equation}

For notational convenience, we regard the pair \((\hat u,\hat b)\) as the
single-shot classical outcome of the randomized measurement protocol.  The
raw operator associated with the outcome \((u,b)\) is the corresponding POVM
element
\begin{equation}
  \sigma_{u,b}
  :=
  M_{b|u}.
  \label{eq:general_raw_operator_def}
\end{equation}
For the realized outcome, we write
\begin{equation}
  \hat\sigma
  :=
  \sigma_{\hat u,\hat b}
  =
  M_{\hat b|\hat u}.
  \label{eq:general_realized_raw_operator}
\end{equation}

Let
\(\mathsf L(\mathcal H)\) denote the real vector space of Hermitian operators
on \(\mathcal H\).  
Following the standard classical-shadow measurement-channel construction
\cite{HuangKuengPreskill2020}, define the associated shadow channel 
as the linear map
\(\mathcal M:\mathsf L(\mathcal H)\to\mathsf L(\mathcal H)\) by
\begin{equation}
 \mathcal M[\rho]
 :=
 \mathbb E_{\hat u}
 \sum_{b\in\mathcal B_{\hat u}}
 M_{b|\hat u}\,
 \operatorname{Tr}(\rho M_{b|\hat u}) .
 \label{eq:general_shadow_channel_expectation_form}
\end{equation}

Equivalently,
\begin{equation}
  \mathcal M[\rho]
  =
  \sum_{u\in\mathcal U} q(u)
  \sum_{b\in\mathcal B_u}
  M_{b|u}\,
  \operatorname{Tr}(\rho M_{b|u}).
  \label{eq:general_shadow_channel_povm}
\end{equation}

Throughout this general formulation, we assume that the randomized
measurement protocol is tomographically complete, in the sense that the
shadow channel
\(
  \mathcal M:\mathsf L(\mathcal H)\to\mathsf L(\mathcal H)
\)
is invertible on the real vector space \(\mathsf L(\mathcal H)\) of Hermitian
operators.  Thus \(\mathcal M^{-1}\) denotes the inverse linear map on
\(\mathsf L(\mathcal H)\).

The classical-shadow reconstruction applies the inverse channel to the raw
snapshot~\cite{HuangKuengPreskill2020}. Thus, for an observed outcome
\((\hat u,\hat b)\), we set
\begin{equation}
  \hat\rho_{\hat u,\hat b}
  :=
  \mathcal M^{-1}(M_{\hat b|\hat u}) .
  \label{eq:general_reconstructed_snapshot}
\end{equation}
Since the mean raw snapshot is \(\mathcal M[\rho]\), the reconstructed
snapshot is unbiased:
\begin{equation}
  \mathbb E_\rho
  \bigl[
    \hat\rho_{\hat u,\hat b}
  \bigr]
  =
  \mathcal M^{-1}\mathcal M[\rho]
  =
  \rho .
  \label{eq:general_snapshot_unbiasedness}
\end{equation}
Consequently, for every observable \(O_\alpha\) considered below,
\begin{equation}
  \mathbb E_\rho
  \operatorname{Tr}(\hat\rho_{\hat u,\hat b}O_\alpha)
  =
  \operatorname{Tr}(\rho O_\alpha).
  \label{eq:general_snapshot_unbiasedness_observable_form}
\end{equation}
The notation up to this point describes one global measurement shot and its
linear reconstruction on \(\mathcal H\). Table~\ref{tab:notation-global-measurement}
summarizes the global measurement and reconstruction notation before we
introduce the output-coordinate vector.

\begin{table}[t]
\centering
\footnotesize
\caption{Global measurement and reconstruction notation on \(\mathcal H\).}
\label{tab:notation-global-measurement}
\begin{tabular}{@{}p{0.30\columnwidth}p{0.64\columnwidth}@{}}
\hline
\textbf{Notation} & \textbf{Meaning} \\
\hline

\(\mathcal H\)
&
Global Hilbert space \(\bigotimes_{j=1}^k\mathcal H_j\).
\\[1mm]
\(k\)
&
Number of tensor factors, or sites, in the multipartite system.
\\[1mm]

\(\mathsf L(\mathcal H)\)
&
Real vector space of Hermitian operators on \(\mathcal H\).
\\[1mm]

\(\mathcal U\)
&
Global measurement-setting space $\mathcal U_1\times\cdots\times\mathcal U_k $.
\\[1mm]

\(u\)
&
A global measurement setting
$  (u_1,\ldots,u_k)$.
\\[1mm]

\(\hat u\)
&
Random global measurement setting in one shot
$  (\hat u_1,\ldots,\hat u_k)$.
\\[1mm]

\(q\)
&
Probability distribution of \(\hat u\) on \(\mathcal U\).
\\[1mm]

\(\mathcal B_u\)
&
Outcome space associated with the setting \(u\).
\\[1mm]

\(b\)
&
A global outcome in \(\mathcal B_u\).
\\[1mm]

\(M_{b|u}\)
&
Global POVM element for outcome \(b\) under setting \(u\).
\\[1mm]

\(\sigma_{u,b}\)
&
Raw operator associated with the outcome pair \((u,b)\).
\\[1mm]

\(\hat\sigma\)
&
Realized raw operator in one shot.
\\[1mm]

\(\mathcal M\)
&
Shadow channel associated with the randomized measurement protocol.
\\[1mm]

\(\mathcal M^{-1}\)
&
Inverse reconstruction map of the shadow channel.
\\[1mm]

\(\hat\rho_{\hat u,\hat b}\)
&
Reconstructed shadow snapshot from the realized outcome
\((\hat u,\hat b)\).
\\
\hline
\end{tabular}
\end{table}

The standard random-unitary classical-shadow protocol is recovered by taking
the setting \(u\) to be a unitary \(U\) and the POVM elements to be
\begin{equation}
  M_{b|U}
  =
  U^\dagger |b\rangle\!\langle b|U .
  \label{eq:general_projective_shadow_povm_element}
\end{equation}
the above definition reduces to the standard shadow channel
\begin{equation}
  \mathcal M[\rho]
  =
  \mathbb E_{\hat U}
  \sum_b
  \hat U^\dagger |b\rangle\!\langle b|\hat U\,
  \operatorname{Tr}\!\left(
    \rho\,\hat U^\dagger |b\rangle\!\langle b|\hat U
  \right).
  \label{eq:general_projective_shadow_channel}
\end{equation}
The biased local Pauli protocol will be recovered later by taking
\(\mathcal U_j=\{X,Y,Z\}\) and
\(M^{(j)}_{s_j|r_j}=\frac12(I+s_j r_j)\).

\subsection{Same expectation values, different joint statistics: estimator covariance versus physical correlation}
\label{subsec:shadow_output_vectors_covariance_matrices}

Unbiased reconstruction can make different shadow protocols equivalent at the
level of the mean vector for a fixed state and observable family. It does not
make them statistically equivalent as multivariate estimators: their joint
laws, and in particular their off-diagonal covariances, retain dependence on
the measurement and reconstruction protocol.
A single reconstructed snapshot can be evaluated against a fixed finite
family
$  \mathcal O
  =
  \{O_1,\ldots,O_p\}$
of Hermitian observables on \(\mathcal H\), thereby producing many estimators
from the same measurement outcome. The identity observable may be included in
this collection. These observables need not be measured directly; they label
the state functionals evaluated on the reconstructed snapshot. The resulting
coordinates are jointly distributed outputs of one shadow protocol, rather
than independent estimators attached to separate experiments.
For a single randomized measurement outcome \((\hat u,\hat b)\), define
\begin{equation}
  \hat x_\alpha
  :=
  \operatorname{Tr}\!\left(
    \hat\rho_{\hat u,\hat b} O_\alpha
  \right)
  =
  \operatorname{Tr}\!\left(
    \mathcal M^{-1}(M_{\hat b|\hat u}) O_\alpha
  \right),
  \label{eq:general_shadow_output_coordinate}
\end{equation}
for $   \alpha=1,\ldots,p$.
We write
\begin{equation}
  \hat x
  :=
  (\hat x_1,\ldots,\hat x_p)^\top
  \in\mathbb R^p .
  \label{eq:general_shadow_output_vector}
\end{equation}
for the resulting single-shot shadow-output vector.  Its mean is
denoted by
\begin{equation}
  m_\alpha
  :=
  \operatorname{Tr}(\rho O_\alpha),
  \qquad
  m:=(m_1,\ldots,m_p)^\top\in\mathbb R^p .
  \label{eq:general_shadow_output_mean_vector}
\end{equation}
By the unbiasedness of the reconstructed snapshot,
\(
  \mathbb E_\rho[\hat\rho_{\hat u,\hat b}]=\rho,
\)
we have
\begin{equation}
  \mathbb E_\rho[\hat x_\alpha]
  =
  \operatorname{Tr}(\rho O_\alpha)
  =
  m_\alpha,
  \qquad
  \alpha=1,\ldots,p.
  \label{eq:general_shadow_output_coordinate_unbiasedness}
\end{equation}
or equivalently
\begin{equation}
  \mathbb E_\rho[\hat x]=m.
  \label{eq:general_shadow_output_unbiasedness}
\end{equation}

These identities give the usual first-moment description, but do not  
characterize the outputs as one jointly generated random vector. They leave  
undetermined the statistical dependence created when the same measurement  
record is reused for several estimates. The following matrices contain this  
missing second-order information.  
The raw second-moment matrix of the shadow-output vector is
\begin{equation}
  M^{(2)}(\rho)
  :=
  \mathbb E_\rho[\hat x\hat x^\top]
  \in\mathbb R^{p\times p}.
  \label{eq:general_raw_second_moment}
\end{equation}
The associated covariance matrix is
\begin{equation}
  \Sigma(\rho)
  :=
  \operatorname{Cov}_\rho(\hat x)
  =
  M^{(2)}(\rho)-mm^\top
  \in\mathbb R^{p\times p}.
  \label{eq:general_shadow_output_covariance}
\end{equation}
Equivalently, its entries are
\begin{equation}
  \Sigma_{\alpha\beta}(\rho)
  =
  \operatorname{Cov}_\rho(\hat x_\alpha,\hat x_\beta)
  =
  \mathbb E_\rho[\hat x_\alpha\hat x_\beta]
  -
  m_\alpha m_\beta .
  \label{eq:general_covariance_entries}
\end{equation}

The distinction is one of both dependence and interpretation. For fixed
observables, the mean vector is determined by the state and is shared by every
unbiased shadow protocol, whereas \(\Sigma(\rho)\) depends additionally on
the measurement and reconstruction protocol. Hence protocols that agree on
all expectation values can nevertheless disagree on the joint statistical
behavior of their estimators.

A physical observable correlation and a shadow-output covariance also answer
different questions. The former is a property of the quantum state and the
observable family and describes correlations in the quantum system. The
latter is a property of the random reconstruction procedure and describes how
its estimators fluctuate together. Physical correlations therefore cannot
replace \(\Sigma(\rho)\). The next proposition records the elementary
properties needed below.
\begin{proposition}[Covariance object]
\label{prop:covariance_object_general}
For the shadow-output vector \(\hat x\) defined above, the matrices
\(M^{(2)}(\rho)\) and \(\Sigma(\rho)\) are positive semidefinite:
\begin{equation}
  M^{(2)}(\rho)\succeq 0,
  \qquad
  \Sigma(\rho)\succeq 0.
  \label{eq:general_matrices_psd}
\end{equation}
Moreover,
\begin{equation}
  \mathbb E_\rho[\hat x_\alpha]
  =
  m_\alpha
  =
  \operatorname{Tr}(\rho O_\alpha),
  \qquad
  \alpha=1,\ldots,p.
  \label{eq:general_mean_identity}
\end{equation}
\end{proposition}

\begin{proof}
The mean identity follows from the unbiasedness of the reconstructed snapshot:
\[
  \mathbb E_\rho[\hat x_\alpha]
  =
  \operatorname{Tr}\!\left(
    \mathbb E_\rho[\hat\rho_{\hat u,\hat b}]\,O_\alpha
  \right)
  =
  \operatorname{Tr}(\rho O_\alpha).
\]
Any \(v\in\mathbb R^p\) satisfies 
\begin{equation}
  v^\top M^{(2)}(\rho)v
  =
  \mathbb E_\rho[(v^\top\hat x)^2]
  \ge 0,
  \label{eq:general_raw_second_moment_psd_proof}
\end{equation}
so \(M^{(2)}(\rho)\succeq0\).  Similarly,
\begin{equation}
  v^\top\Sigma(\rho)v
  =
  \operatorname{Var}_\rho(v^\top\hat x)
  \ge 0,
  \label{eq:general_covariance_psd_proof}
\end{equation}
which proves \(\Sigma(\rho)\succeq0\).
\end{proof}

\begin{remark}[Identity coordinate]
\label{rem:identity_coordinate_general}
The identity observable may be included among the \(O_\alpha\).  In that case
the corresponding shadow-output coordinate is
$  \hat x_I
  =
  \operatorname{Tr}(\hat\rho_{\hat u,\hat b}I)$.
Its expectation is always
\(
  \mathbb E_\rho[\hat x_I]
  =
  \operatorname{Tr}(\rho I)
  =
  1
\)
for normalized states.  In many standard shadow protocols this coordinate is
deterministic, but this determinism is not needed for the finite-sample
covariance theory below.  The coordinate can simply be included as one of the
components of the shadow-output vector whenever it is useful.
\end{remark}

\begin{remark}[Pauli coordinates as a special case]
\label{rem:pauli_coordinates_special_case_general}
If \(\mathcal H=(\mathbb C^2)^{\otimes k}\) and
\(\mathcal O\) is chosen to be a collection of Pauli strings, then
\(\hat x_\alpha\) is the reconstructed Pauli coefficient associated with the
corresponding Pauli observable.  This Pauli-coordinate viewpoint is closely
related to the Pauli-invariant shadow formalism of
Bu--Koh--Garcia--Jaffe~\cite{BuKohGarciaJaffe2024}, where Pauli-invariant
unitary ensembles lead to Pauli-diagonal reconstruction maps.  The biased local Pauli protocol studied later is a special case of the present
finite randomized POVM formulation.  In that case, the selected one-shot
radius appearing in the finite-sample theory can be computed explicitly from
the local basis-selection probabilities.
\end{remark}

Table~\ref{tab:notation-shadow-output} summarizes the single-shot output
and covariance notation used below.
\begin{table}[t]
\centering
\footnotesize
\caption{Single-shot shadow-output and covariance notation.}
\label{tab:notation-shadow-output}
\begin{tabular}{@{}p{0.30\columnwidth}p{0.64\columnwidth}@{}}
\hline
\textbf{Notation} & \textbf{Meaning} \\
\hline

\(\mathcal O\)
&
Fixed finite family of Hermitian observables.
\\[1mm]

\(p\)
&
Number of observables in \(\mathcal O=\{O_1,\ldots,O_p\}\).
\\[1mm]

\(\alpha\)
&
Coordinate index, ranging over \(1,\ldots,p\).
\\[1mm]

\(O_\alpha\)
&
The \(\alpha\)-th observable in \(\mathcal O\).
\\[1mm]

\(\hat x_\alpha\)
&
Single-shot reconstructed output coordinate associated with \(O_\alpha\).
\\[1mm]

\(\hat x\)
&
Single-shot shadow-output vector in \(\mathbb R^p\).
\\[1mm]

\(m_\alpha\)
&
Mean of the single-shot coordinate \(\hat x_\alpha\).
\\[1mm]

\(m\)
&
Mean vector of the single-shot shadow-output vector \(\hat x\).
\\[1mm]

\(M^{(2)}(\rho)\)
&
Raw second-moment matrix of \(\hat x\).
\\[1mm]

\(\Sigma(\rho)\)
&
Covariance matrix of \(\hat x\).
\\[1mm]

\(\Sigma_{\alpha\beta}(\rho)\)
&
Covariance between \(\hat x_\alpha\) and \(\hat x_\beta\).
\\
\hline
\end{tabular}
\end{table}

\subsection{Recovering protocol-induced covariance from repeated shadow snapshots}
\label{subsec:independent_samples_empirical_covariances}

We now pass from the single-shot shadow-output distribution to a finite
sample.  Fix a sample size \(N\).  For \(i=1,\ldots,N\), let
\(
  (\hat u_i,\hat b_i)
\)
be independent outcomes generated by applying the same randomized shadow
measurement protocol to independent copies of the state \(\rho\).  Thus
\(\hat u_i\) is the randomized measurement setting in the \(i\)-th shot, and
\(\hat b_i\) is the corresponding measurement outcome.

The reconstructed shadow snapshot in the \(i\)-th shot is
\begin{equation}
  \hat\rho_i
  :=
  \hat\rho_{\hat u_i,\hat b_i}
  =
  \mathcal M^{-1}(M_{\hat b_i|\hat u_i}).
  \label{eq:general_sample_reconstructed_snapshot}
\end{equation}
For the fixed observable family
\( \mathcal O=\{O_1,\ldots,O_p\}\), we 
define the associated shadow-output vector by
\begin{equation}
  \hat x_{i,\alpha}
  :=
  \operatorname{Tr}(\hat\rho_i O_\alpha),
  \qquad
  \alpha=1,\ldots,p,
  \label{eq:general_sample_shadow_output_coordinate}
\end{equation}
and write
\begin{equation}
  \hat x_i
  :=
  (\hat x_{i,1},\ldots,\hat x_{i,p})^\top
  \in\mathbb R^p .
  \label{eq:general_sample_shadow_output_vector}
\end{equation}
Then
\(
  \hat x_1,\ldots,\hat x_N
\)
are independent copies of the single-shot shadow-output vector \(\hat x\)
defined in the previous subsection.  In particular,
\begin{equation}
  \mathbb E_\rho[\hat x_i]=m,
  \qquad
  \operatorname{Cov}_\rho(\hat x_i)=\Sigma(\rho),
  \qquad
  i=1,\ldots,N.
  \label{eq:general_sample_mean_covariance}
\end{equation}

The empirical shadow-output mean is
\begin{equation}
  \bar{\hat x}_N
  :=
  \frac1N\sum_{i=1}^N \hat x_i .
  \label{eq:general_empirical_shadow_output_mean}
\end{equation}
We distinguish the true-centered empirical covariance matrix
\begin{equation}
  \widehat\Sigma_N^{\mathrm{tc}}
  :=
  \frac1N
  \sum_{i=1}^N
  (\hat x_i-m)(\hat x_i-m)^\top
  \in\mathbb R^{p\times p}
  \label{eq:true_centered_empirical_covariance_general}
\end{equation}
from the sample-centered empirical covariance matrix
\begin{equation}
  \widehat\Sigma_N^{\mathrm{sc}}
  :=
  \frac1N
  \sum_{i=1}^N
  (\hat x_i-\bar{\hat x}_N)
  (\hat x_i-\bar{\hat x}_N)^\top
  \in\mathbb R^{p\times p}.
  \label{eq:sample_centered_empirical_covariance_general}
\end{equation}
We use the $1/N$-normalized empirical covariance throughout; no unbiased $1/(N-1)$ correction is used.
The true-centered covariance is useful for concentration estimates because
its summands are centered with respect to the mean \(m\).  The
sample-centered covariance is the empirical covariance matrix computed
directly from data, since \(m\) is unknown in a tomography problem.

The two empirical covariance matrices are related by the exact identity
\begin{align}
  \widehat\Sigma_N^{\mathrm{sc}}
  -
  \widehat\Sigma_N^{\mathrm{tc}}
  =
  -
  (\bar{\hat x}_N-m)(\bar{\hat x}_N-m)^\top ,
  \label{eq:exact_rank_one_relation_true_sample_centered_general}
\end{align}
which is shown in Appendix \ref{AA1}.

Thus the sample-centered covariance is a rank-one negative perturbation of
the true-centered covariance.
This rank-one relation will be used later to transfer the finite-sample
operator-norm approximation first proved for the true-centered selected
covariance to the sample-centered selected covariance computed from the
measurement data.

\subsection{Selected estimator covariance and its collective fluctuation spectrum}
\label{subsec:selected_covariance_spectral_meaning}

This is not merely a generic covariance compression. Classical shadows  
target a prescribed observable family inside a much larger operator space.  
For that family, selected covariance is the second-order counterpart of the  
usual selected expectation-value problem. We therefore fix the coordinates  
before observing the data and study the covariance and spectrum of their  
jointly reconstructed outputs.  
Let  
\begin{equation}
  J_l
  =
  \{\alpha_1,\ldots,\alpha_l\}
  \subset\{1,\ldots,p\}.
  \label{eq:general_selected_coordinate_set}
\end{equation}
be a deterministic selected coordinate set, fixed independently of the
measurement data.  We assume that the indices in \(J_l\) are distinct.  Let
\begin{equation}
  R_l\in\{0,1\}^{l\times p}.
  \label{eq:general_selection_matrix_def}
\end{equation}
be the corresponding coordinate-selection matrix.  Thus the \(a\)-th row of
\(R_l\) has a single nonzero entry, equal to \(1\), in the column
\(\alpha_a\).  

Equivalently, the rows of \(R_l\) are orthonormal coordinate vectors, and hence
\begin{equation}
  R_lR_l^\top=I_l.
  \label{eq:general_selection_matrix_property}
\end{equation}
Thus \(R_l\) acts by retaining precisely the coordinates indexed by
\(J_l=\{\alpha_1,\ldots,\alpha_l\}\) and discarding all remaining coordinates.
More explicitly, for any vector \(x\in\mathbb R^p\),
\begin{equation}
  R_lx
  =
  \begin{pmatrix}
    x_{\alpha_1}\\
    \vdots\\
    x_{\alpha_l}
  \end{pmatrix}
  \in\mathbb R^l .
  \label{eq:general_selection_matrix_action}
\end{equation}
Applying this coordinate projection to the shadow-output vector gives the
selected shadow-output vector
\begin{equation}
  R_l\hat x
  =
  \begin{pmatrix}
    \hat x_{\alpha_1}\\
    \vdots\\
    \hat x_{\alpha_l}
  \end{pmatrix}.
  \label{eq:general_selected_shadow_output_vector}
\end{equation}
Its mean is obtained by applying the same selection matrix to the
full mean vector \(m\):
\begin{equation}
  R_lm
  =
  \begin{pmatrix}
    m_{\alpha_1}\\
    \vdots\\
    m_{\alpha_l}
  \end{pmatrix}.
  \label{eq:general_selected_mean_vector}
\end{equation}

We define the selected covariance matrix by compressing the full
shadow-output covariance matrix with the fixed selection matrix \(R_l\):
\begin{equation}
  \Sigma_l(\rho)
  :=
  R_l\Sigma(\rho)R_l^\top
  \in\mathbb R^{l\times l}.
  \label{eq:selected_covariance_general}
\end{equation}
Since \(R_l\) is deterministic, this compressed matrix is exactly the
covariance matrix of the selected shadow-output vector \(R_l\hat x\). Using
\(\mathbb E_\rho[\hat x]=m\), we have
\begin{equation}
  \Sigma_l(\rho)
  =
  \operatorname{Cov}_\rho(R_l\hat x)
  =
  \mathbb E_\rho
  \left[
    (R_l\hat x-R_lm)(R_l\hat x-R_lm)^\top
  \right].
  \label{eq:selected_covariance_cov_form_general}
\end{equation}
Equivalently, if \(J_l=\{\alpha_1,\ldots,\alpha_l\}\), then the
\((a,b)\)-entry of \(\Sigma_l(\rho)\) is
\begin{equation}
  \bigl(\Sigma_l(\rho)\bigr)_{ab}
  =
  \operatorname{Cov}_\rho
  \bigl(
    \hat x_{\alpha_a},
    \hat x_{\alpha_b}
  \bigr),
  \qquad
  1\le a,b\le l .
  \label{eq:selected_covariance_entry_general}
\end{equation}
Thus \(\Sigma_l(\rho)\) records the joint statistical behavior of the
prescribed shadow estimators, while retaining the standard classical-shadow
focus on a fixed observable family rather than the full ambient operator
space.
The spectral meaning of \(\Sigma_l(\rho)\) is immediate from its variational
characterization.  For any vector \(v\in\mathbb R^l\), the scalar random
variable
\(
  \langle v,R_l\hat x\rangle
\)
is the shadow-output estimate of the linear combination of selected
coordinates specified by \(v\).  
Its variance is
\begin{equation}
  \operatorname{Var}_\rho(\langle v,R_l\hat x\rangle)
  =
  v^\top \Sigma_l(\rho)v .
  \label{eq:selected_covariance_variance_form_general}
\end{equation}
Consequently, for unit vectors \(v\in\mathbb R^l\),
we have
\begin{equation}
  \max_{\|v\|_2=1}
  \operatorname{Var}_\rho(\langle v,R_l\hat x\rangle)
  =
  \max_{\|v\|_2=1}
  v^\top\Sigma_l(\rho)v
  =
  \lambda_{\max}(\Sigma_l(\rho)).
  \label{eq:selected_covariance_max_variance_general}
\end{equation}
Similarly,
\begin{equation}
  \min_{\|v\|_2=1}
  \operatorname{Var}_\rho(\langle v,R_l\hat x\rangle)
  =
  \lambda_{\min}(\Sigma_l(\rho)).
  \label{eq:selected_covariance_min_variance_general}
\end{equation}
Thus the eigenvalues of \(\Sigma_l(\rho)\) give the extremal fluctuation
scales of normalized linear combinations of the selected shadow outputs, and
the corresponding eigenvectors identify the extremal directions. This is why
we estimate the selected covariance in operator norm. If an empirical selected
covariance matrix \(\widehat\Sigma_{l,N}\) satisfies
\begin{equation}
  \|\widehat\Sigma_{l,N}-\Sigma_l(\rho)\|_{\mathrm{op}}
  \le
  \varepsilon ,
  \label{eq:selected_operator_norm_approx_assumption_general}
\end{equation}
then all selected variance scales are uniformly approximated:
\begin{equation}
  \left|
    v^\top\widehat\Sigma_{l,N}v
    -
    v^\top\Sigma_l(\rho)v
  \right|
  \le
  \varepsilon
  \qquad
  \text{for all } \|v\|_2=1.
  \label{eq:selected_operator_norm_variance_approx_general}
\end{equation}
In particular, the selected eigenvalues are stable under such an
operator-norm perturbation, and isolated selected spectral subspaces can be
controlled by standard spectral perturbation theory.  These consequences
will be made precise in the finite-sample results below.

Finally, we emphasize that the selection matrix \(R_l\) is fixed before the
data are observed.  The results below are therefore fixed-selection
statements.  They do not address data-dependent or adaptive choices of
selected coordinates, which would require additional post-selection or
uniform concentration arguments.

Thus many-output shadow reconstruction contains a statistical layer absent  
from a purely expectation-value description. The next section develops its  
finite-sample operator-norm and spectral recovery.  
\section{Finite-Sample Spectral Approximation of Selected Shadow Covariances}
\label{sec:general_selected_covariance_spectral_recovery}

We now prove the finite-sample selected covariance-estimation theorem in a
form that does not rely on the local Pauli structure.  The goal is to control
the selected covariance matrix
\(\Sigma_l(\rho)=R_l\Sigma(\rho)R_l^\top\) from independent shadow samples,
and then to convert this operator-norm control into selected eigenvalue and
spectral-projector guarantees.

\subsection{Overview of the selected finite-sample argument}
\label{subsec:overview_selected_finite_sample_argument}

The input is the general shadow-output vector
\(\hat x=(\hat x_1,\ldots,\hat x_p)^\top\in\mathbb R^p\) introduced in
Section~\ref{sec:general_covariance}, together with a fixed
coordinate-selection matrix \(R_l\).  The covariance matrix is
\(\Sigma(\rho)=\operatorname{Cov}_\rho(\hat x)\), and the selected covariance
matrix is \(\Sigma_l(\rho)=R_l\Sigma(\rho)R_l^\top\).

The point of this section is deliberately selected rather than
full-dimensional.  We do not attempt to estimate the full \(p\times p\)
covariance matrix in operator norm.  Instead, all concentration takes place
after the fixed compression \(R_l\).  Thus the matrix dimension entering the
probability bounds is the selected dimension \(l\), not the ambient
dimension \(p\).

The probabilistic input is standard.  Once the selected centered vectors
\(R_l(\hat x_i-m)\) are uniformly bounded in Euclidean norm, the
true-centered empirical covariance error becomes an average of independent
centered self-adjoint matrices.  We control this average by the
self-adjoint matrix Bernstein inequality in the form of
Tropp~\cite[Theorem~6.1]{Tropp2012UserFriendly}.  This gives an
operator-norm perturbation bound for the selected covariance.  Standard
spectral perturbation results, namely Weyl's inequality and the
Davis--Kahan theorem, then convert this operator-norm control into selected
eigenvalue and spectral-projector bounds
\cite{HornJohnson2012,DavisKahan1970,StewartSun1990,Vershynin2018}.

There is one statistical complication: the mean \(m\) is unknown
in data analysis.  The covariance matrix computed from data is therefore
sample-centered, not true-centered.  We first prove concentration for the
true-centered selected covariance, and then use the exact rank-one centering
identity to pass to the sample-centered covariance.  The additional
empirical-mean term is controlled by a coordinatewise Hoeffding inequality
and a union bound over the selected coordinates
\cite{Hoeffding1963,Vershynin2018}.

The section is organized as follows.  We first introduce the selected
one-shot radius and state the sample-centered main theorem.  We then prove
the true-centered matrix concentration bound, derive its spectral
consequences, and finally transfer the result to the sample-centered
covariance using the rank-one centering identity.

\subsection{Selected one-shot radius and main theorem}
\label{subsec:selected_one_shot_radius_main_theorem}
This subsection states the main finite-sample result of the paper in the
form used in the concrete protocol sections.  The statement is intentionally
a constant-error theorem under a bounded selected radius.  It says that, if
the selected dimension and the selected centered one-shot radius remain
bounded independently of the ambient system size, then the selected
sample-centered covariance matrix can be estimated in operator norm with a
sample size independent of the ambient system size.

All protocol dependence enters through the selected one-shot radius.  The
later local Pauli and general local product sections are devoted to computing
or bounding this radius in concrete measurement models.  The theorem is not
tied to Pauli compatibility, locality, or an explicit covariance formula;
those structures are used only later to verify the bounded-radius hypothesis.

The theorem is stated for the sample-centered covariance because this is the
matrix computed directly from observed shadow data.  Its proof is given in
the rest of this section.  We first establish a concentration bound for the
true-centered selected empirical covariance, and then transfer it to the
sample-centered covariance by using the exact rank-one centering identity.
Along the way, we also obtain more detailed eigenvalue and spectral-projector
perturbation bounds.

Let \(R_l\in\{0,1\}^{l\times p}\) be the fixed coordinate-selection matrix
introduced in Section~\ref{subsec:selected_covariance_spectral_meaning}, so
that
\begin{equation}
  R_lR_l^\top=I_l.
  \label{eq:fs_general_selection_matrix_orthonormal_rows}
\end{equation}
We define the selected sample-centered and true-centered empirical covariance
matrices by
\begin{align}
  \widehat\Sigma_{l,N}^{\mathrm{sc}}
  &:=
  R_l\widehat\Sigma_N^{\mathrm{sc}}R_l^\top,
  \label{eq:fs_general_selected_sample_centered_def}
  \\
  \widehat\Sigma_{l,N}^{\mathrm{tc}}
  &:=
  R_l\widehat\Sigma_N^{\mathrm{tc}}R_l^\top .
  \label{eq:fs_general_selected_true_centered_def}
\end{align}
Equivalently, we have
\begin{align}
  \widehat\Sigma_{l,N}^{\mathrm{sc}}
  =&
  \frac1N\sum_{i=1}^N
  (R_l\hat x_i-R_l\bar{\hat x}_N)
  (R_l\hat x_i-R_l\bar{\hat x}_N)^\top ,
  \label{eq:fs_general_selected_sc_expanded} \\
  \widehat\Sigma_{l,N}^{\mathrm{tc}}
  =&
  \frac1N\sum_{i=1}^N
  (R_l\hat x_i-R_lm)
  (R_l\hat x_i-R_lm)^\top 
  \label{eq:fs_general_selected_tc_expanded}\\
  =&
  \frac1N\sum_{i=1}^N \hat z_i \hat z_i^\top, 
  \label{eq:selected_true_centered_covariance_z_form_general}
\end{align}
where  the selected centered vector $\hat z_i$ is defined as
\begin{equation}
\hat z_i:=R_l(\hat x_i-m)\in\mathbb R^l,
\qquad
i=1,\ldots,N.
\label{eq:selected_centered_vector_general}
\end{equation}

The next definition isolates the deterministic boundedness assumption needed
for the matrix Bernstein argument.  It should be read as a selected analogue
of the usual bounded-sample assumption in nonasymptotic covariance
estimation.  The full vector \(\hat x\) may live in a very large ambient
space, but the theorem uses only the Euclidean radius of the selected vector
\(R_l\hat x\).

\begin{table}[t]
\centering
\footnotesize
\caption{One-shot selected-output notation, covariance quantities, and radius parameters.}
\label{tab:notation-one-shot-selected}
\begin{tabular}{@{}p{0.32\columnwidth}p{0.62\columnwidth}@{}}
\hline
\textbf{Notation} & \textbf{Meaning} \\
\hline

\(\hat\rho_{\hat u,\hat b}\)
&
Single-shot reconstructed shadow snapshot.
\\[1mm]

\(\mathcal O\)
&
Fixed finite family of Hermitian observables defining the output coordinates.
\\[1mm]

\(O_\alpha\)
&
The \(\alpha\)-th observable in \(\mathcal O\).
\\[1mm]

\(\hat x_\alpha\)
&
Single-shot reconstructed output coordinate associated with \(O_\alpha\).
\\[1mm]

\(\hat x\)
&
Single-shot shadow-output vector.
\\[1mm]

\(m_\alpha\)
&
Mean of \(\hat x_\alpha\).
\\[1mm]

\(m\)
&
Mean vector of \(\hat x\).
\\[1mm]

\(M^{(2)}(\rho)\)
&
Raw second-moment matrix of the shadow-output vector.
\\[1mm]

\(\Sigma(\rho)\)
&
Covariance matrix of the full shadow-output vector.
\\[1mm]

\(\Sigma_{\alpha\beta}(\rho)\)
&
Covariance between the output coordinates \(\hat x_\alpha\) and
\(\hat x_\beta\).
\\[1mm]

\(J_l\)
&
Fixed selected coordinate set.
\\[1mm]

\(l\)
&
Number of selected coordinates.
\\[1mm]

\(R_l\)
&
Coordinate-selection matrix associated with \(J_l\).
\\[1mm]

\(R_l\hat x\)
&
Selected one-shot shadow-output vector.
\\[1mm]

\(R_lm\)
&
Selected mean vector.
\\[1mm]

\(\Sigma_l(\rho)\)
&
Selected covariance matrix \(R_l\Sigma(\rho)R_l^\top\).
\\[1mm]

\(B_l\)
&
Selected one-shot radius of \(R_l\hat x\).
\\[1mm]

\(A_l(\rho)\)
&
Centered selected one-shot radius of \(R_l(\hat x-m)\).
\\[1mm]

\(\widetilde A_l(\rho)\)
&
Deterministic upper bound \(B_l+\|R_lm\|_2\) on \(A_l(\rho)\).
\\
\hline
\end{tabular}
\end{table}

\begin{definition}[Selected one-shot radii]
\label{def:selected_one_shot_radius_general}
Fix the selection matrix \(R_l\).  The selected one-shot output is the
\(l\)-dimensional random vector \(R_l\hat x\).  We define its one-shot radius
to be the smallest deterministic radius of an origin-centered Euclidean ball
that contains this selected output almost surely:
\begin{align}
  B_l
  :=&
  \operatorname*{ess\,sup}
  \|R_l\hat x\|_2 .
  \label{eq:selected_one_shot_radius_general}
\end{align}
We also define the centered selected one-shot radius by
\begin{align}  A_l(\rho)
  :=&
  \operatorname*{ess\,sup}
  \|R_l(\hat x-m)\|_2 .
  \label{eq:selected_centered_radius_general}
\end{align}
Equivalently, \(B_l\) is the least number, up to null events, such that
\(\|R_l\hat x\|_2\le B_l\) almost surely.  The essential supremum is taken over
the one-shot randomness of the shadow protocol.  In the finite-outcome
setting, this is simply the maximum over all outcomes that can occur under
the randomized measurement protocol.  
Throughout this section we assume that
\begin{equation}
  B_l<\infty,
  \qquad
  A_l(\rho)<\infty .
  \label{eq:selected_one_shot_radii_finite_general}
\end{equation}
For later use, we also record the elementary deterministic upper bound
\begin{equation}
  \widetilde A_l(\rho)
  :=
  B_l+\|R_lm\|_2
  \ge
  A_l(\rho).
  \label{eq:selected_centered_radius_upper_bound_general}
\end{equation}
Indeed, this follows from the triangle inequality:
\(
  \|R_l(\hat x-m)\|_2
  \le
  \|R_l\hat x\|_2+\|R_lm\|_2 .
\)
Thus \(\widetilde A_l(\rho)\) is a convenient upper bound on the centered
selected one-shot radius.  In applications, one may use either the exact
radius \(A_l(\rho)\) or any deterministic upper bound such as
\(\widetilde A_l(\rho)\).
\end{definition}

The finite-sample estimates below are built from the one-shot selected output,
its covariance, and the corresponding selected-radius parameters.
Table~\ref{tab:notation-one-shot-selected} collects this notation before we
introduce the \(N\)-sample quantities and error levels.

The quantity \(B_l\) bounds the raw selected one-shot output, while
\(A_l(\rho)\) bounds the centered selected output.
For \(\delta\in(0,1)\), define the true-centered error level by
\begin{equation}
  \varepsilon_{l,N}^{\mathrm{tc}}(\delta)
  :=
  A_l(\rho)^2
  \left(
    \sqrt{\frac{8\log(2l/\delta)}{N}}
    +
    \frac{4\log(2l/\delta)}{3N}
  \right).
  \label{eq:epsilon_tc_general}
\end{equation}
Define also the selected empirical-mean correction by
\begin{equation}
  \mu_{l,N}(\delta)
  :=
  \frac{2lA_l(\rho)^2\log(2l/\delta)}{N}.
  \label{eq:mu_general}
\end{equation}
The sample-centered error level is
\begin{equation}
  \varepsilon_{l,N}^{\mathrm{sc}}(\delta)
  :=
  \varepsilon_{l,N}^{\mathrm{tc}}(\delta/2)
  +
  \mu_{l,N}(\delta/2).
  \label{eq:epsilon_sc_def_general}
\end{equation}
Equivalently, substituting
\eqref{eq:epsilon_tc_general} and \eqref{eq:mu_general} into
\eqref{eq:epsilon_sc_def_general} gives
\begin{equation}
  \varepsilon_{l,N}^{\mathrm{sc}}(\delta)
  =
  A_l(\rho)^2
  \left(
    \sqrt{\frac{8\log(4l/\delta)}{N}}
    +
    \left(\frac43+2l\right)
    \frac{\log(4l/\delta)}{N}
  \right).
  \label{eq:epsilon_sc_general}
\end{equation}

We next collect the notation that depends on the \(N\) independent samples
and the finite-sample error parameters used in the main theorem.
Table~\ref{tab:notation-n-sample-selected} separates these empirical objects
from the one-shot quantities listed in
Table~\ref{tab:notation-one-shot-selected}.

\begin{table}[t]
\centering
\footnotesize
\caption{Notation for \(N\) independent samples, empirical selected covariances,
and finite-sample error levels.}
\label{tab:notation-n-sample-selected}
\begin{tabular}{@{}p{0.36\columnwidth}p{0.58\columnwidth}@{}}
\hline
\textbf{Notation} & \textbf{Meaning} \\
\hline

\(N\)
&
Number of independent shadow samples.
\\[1mm]

\((\hat u_i,\hat b_i)\)
&
Measurement setting and outcome in the \(i\)-th shot.
\\[1mm]

\(\hat\rho_i\)
&
Reconstructed shadow snapshot in the \(i\)-th shot.
\\[1mm]

\(\hat x_{i,\alpha}\)
&
\(\alpha\)-th shadow-output coordinate in the \(i\)-th shot.
\\[1mm]

\(\hat x_i\)
&
Shadow-output vector obtained from the \(i\)-th shot.
\\[1mm]

\(\bar{\hat x}_N\)
&
Empirical mean of the shadow-output vectors.
\\[1mm]

\(\widehat\Sigma_N^{\mathrm{tc}}\)
&
Full true-centered empirical covariance matrix.
\\[1mm]

\(\widehat\Sigma_N^{\mathrm{sc}}\)
&
Full sample-centered empirical covariance matrix.
\\[1mm]

\(\widehat\Sigma_{l,N}^{\mathrm{tc}}\)
&
Selected true-centered empirical covariance matrix.
\\[1mm]

\(\widehat\Sigma_{l,N}^{\mathrm{sc}}\)
&
Selected sample-centered empirical covariance matrix.
\\[1mm]

\(\hat z_i\)
&
Selected centered sample vector \(R_l(\hat x_i-m)\).
\\[1mm]

\(\delta\)
&
Failure probability in the high-probability bounds.
\\[1mm]

\(\varepsilon_{l,N}^{\mathrm{tc}}(\delta)\)
&
True-centered selected covariance error level.
\\[1mm]

\(\mu_{l,N}(\delta)\)
&
Selected empirical-mean correction term.
\\[1mm]

\(\varepsilon_{l,N}^{\mathrm{sc}}(\delta)\)
&
Sample-centered selected covariance error level.
\\[1mm]

\(l_0\)
&
Uniform upper bound on the selected dimension \(l\).
\\[1mm]

\(A_0\)
&
Uniform upper bound on the centered selected radius \(A_l(\rho)\).
\\[1mm]

\(\epsilon\)
&
Target operator-norm accuracy in the constant-error theorem.
\\
\hline
\end{tabular}
\end{table}

We now state the main finite-sample theorem for the sample-centered selected
covariance.  This is the covariance matrix computed directly from data,
because the mean \(m\) is unknown.  The theorem gives a
high-probability operator-norm bound for its distance from the selected
covariance
\(
  \Sigma_l(\rho)=R_l\Sigma(\rho)R_l^\top .
\)
The more detailed selected eigenvalue and spectral-projector consequences
are stated below as refinements of the same operator-norm approximation.

The error bound depends on four quantities: the selected centered radius
\(A_l(\rho)\), the selected dimension \(l\), the sample size \(N\), and the
failure probability \(\delta\).  The ambient dimension \(p\) does not appear
in the logarithmic factor, because the covariance matrix is compressed by
the fixed selection matrix \(R_l\) before concentration is applied.

\begin{theorem}[Constant selected error under a bounded selected radius]
\label{thm:constant_selected_error_bounded_radius}
Let \(\hat x_1,\ldots,\hat x_N\) be independent copies of the single-shot
shadow-output vector \(\hat x\), and let \(R_l\) be a deterministic
coordinate-selection matrix fixed independently of the data.  Assume that
\begin{equation}
  l\le l_0,
  \qquad
  A_l(\rho)\le A_0,
  \label{eq:bounded_selected_radius_assumption_general}
\end{equation}
where \(l_0\) and \(A_0\) are constants independent of the ambient system
size.  Then
\begin{equation}
  \varepsilon_{l,N}^{\mathrm{sc}}(\delta)
  \le
  A_0^2
  \left(
    \sqrt{\frac{8\log(4l_0/\delta)}{N}}
    +
    \left(\frac43+2l_0\right)
    \frac{\log(4l_0/\delta)}{N}
  \right).
  \label{eq:epsilon_sc_constant_radius_bound_general}
\end{equation}
In particular, for any target accuracy \(\epsilon>0\), it is enough to take
\begin{equation}
  N
  \ge
  \max\left\{
    \frac{32A_0^4\log(4l_0/\delta)}{\epsilon^2},
    \frac{2A_0^2(\frac43+2l_0)\log(4l_0/\delta)}{\epsilon}
  \right\}
  \label{eq:N_sufficient_constant_radius_general}
\end{equation}
to ensure that, with probability at least \(1-\delta\),
\begin{equation}
  \left\|
    \widehat\Sigma_{l,N}^{\mathrm{sc}}-\Sigma_l(\rho)
  \right\|_{\mathrm{op}}
  \le
  \epsilon .
  \label{eq:operator_constant_radius_general}
\end{equation}
Thus, whenever \(l\) and \(A_l(\rho)\) are bounded independently of the
ambient system size, constant operator-norm accuracy for the selected
sample-centered covariance requires a number of samples independent of the
ambient system size.
\end{theorem}

The assumptions \(B_l<\infty\) and \(A_l(\rho)\le A_0\) are stated
abstractly in this general section.  They should be understood as
protocol-dependent boundedness conditions on the selected one-shot output and
on its centered version.  
In applications, one may verify the second condition
either by computing the centered radius \(A_l(\rho)\) directly or by using the
deterministic upper bound
\(
  A_l(\rho)
  \le
  \widetilde A_l(\rho)
  =
  B_l+\|R_lm\|_2 .
\)
For the biased local Pauli protocol, the raw selected radius \(B_l\)
can be computed explicitly from the local basis-selection probabilities and
the supports of the selected Pauli coordinates.  Together with the elementary
bound \(\|R_lm\|_2\le \sqrt l\) for Pauli expectation values, this gives a
state-uniform bound on \(A_l(\rho)\) whenever the selected set size, the Pauli
weights, and the inverse local basis probabilities are uniformly bounded.

More generally, for local product shadow protocols with finite-weight
selected product observables, the raw selected radius \(B_l\), and hence the
centered radius \(A_l(\rho)\), can be bounded in terms of local reconstruction
coefficients and support sizes rather than the total number of tensor
factors.  These protocol-specific estimates are carried out in the later
sections and verify the bounded-radius hypothesis used in
Theorem~\ref{thm:constant_selected_error_bounded_radius}.

\subsection{True-centered selected covariance concentration}
\label{subsec:true_centered_selected_covariance_concentration}
We now begin the proof of
Theorem~\ref{thm:constant_selected_error_bounded_radius}.  The first step is
to analyze the true-centered selected covariance.  This is the random-matrix
core of the section.
After true centering, the covariance error
can be written as
\[
  \widehat\Sigma_{l,N}^{\mathrm{tc}}-\Sigma_l(\rho)
  =
  \frac1N\sum_{i=1}^N
  \left(\hat z_i\hat z_i^\top-\Sigma_l(\rho)\right),
\]
where the summands are independent, centered, self-adjoint \(l\times l\)
matrices.  This is exactly the setting of the self-adjoint matrix Bernstein
inequality.  We use Tropp's noncommutative Bernstein bound
\cite[Theorem~6.1]{Tropp2012UserFriendly}; the proof below verifies its two
inputs: an almost-sure operator-norm bound on each summand and a variance
proxy bound for the sum of squared summands.

The role of the selected radius is transparent here.  The bound
\(\|\hat z_i\|_2\le A_l(\rho)\) implies
\(
  \|\hat z_i\hat z_i^\top\|_{\mathrm{op}}\le A_l(\rho)^2,
\)
which in turn gives both the summand norm bound and the variance proxy.  Thus
the concentration rate is governed by \(A_l(\rho)^2\), and the matrix
dimension factor is \(l\).

Since \eqref{eq:selected_centered_vector_general} implies
\begin{equation}
  \mathbb E_\rho[\hat z_i]=0,
  \label{eq:selected_centered_vector_mean_zero_general}
\end{equation}
the definition of the selected covariance implies
\begin{equation}
  \Sigma_l(\rho)
  =
  \mathbb E_\rho[\hat z_i\hat z_i^\top].
  \label{eq:selected_covariance_z_second_moment_general}
\end{equation}
The following elementary bound is the only place where the centered selected
radius \(A_l(\rho)\) enters the true-centered concentration argument.

\begin{lemma}[Selected centered radius bound]
\label{lem:selected_centered_radius_general}
For each shot \(i=1,\ldots,N\), we have
\begin{equation}
  \|\hat z_i\|_2
  =
  \|R_l(\hat x_i-m)\|_2
  \le
  A_l(\rho)
  \label{eq:selected_centered_radius_bound_general}
\end{equation}
almost surely.
\end{lemma}

\begin{proof}
This is immediate from the definition
\[
  A_l(\rho)
  =
  \operatorname*{ess\,sup}
  \|R_l(\hat x-m)\|_2 .
\]
Since \(\hat x_i\) has the same one-shot distribution as \(\hat x\), the same
almost-sure bound holds for every \(i=1,\ldots,N\).
\end{proof}

The next proposition is the random-matrix concentration input for the
true-centered covariance.  We use the self-adjoint matrix Bernstein inequality
of Tropp~\cite[Theorem~6.1]{Tropp2012UserFriendly}, which controls the
operator norm of a sum of independent centered self-adjoint random matrices
through a uniform summand bound and a matrix variance proxy.  In the present
setting, the summands are
\(\hat X_i=\hat z_i\hat z_i^\top-\Sigma_l(\rho)\), where
\(\hat z_i=R_l(\hat x_i-m)\).  The selected centered-radius bound
\(\|\hat z_i\|_2\le A_l(\rho)\) supplies both Bernstein inputs: it gives the
almost-sure summand bound and the variance-proxy bound used in the proof
below.  No coordinatewise independence of \(R_l\hat x_i\) is assumed; only
independence across shadow shots is used.

\begin{proposition}[True-centered concentration for selected shadow covariance]
\label{prop:true_centered_selected_shadow_covariance_concentration}
For every \(t>0\),
\begin{align}
&  \mathbb P\!\left(
  \left\|
    \widehat\Sigma_{l,N}^{\mathrm{tc}}-\Sigma_l(\rho)
  \right\|_{\mathrm{op}}
  \ge t
  \right) \notag\\
  \le &
  2l
  \exp\!\left(
  -
  \frac{Nt^2}{
    8A_l(\rho)^4+\frac43A_l(\rho)^2t
  }
  \right).
  \label{eq:true_centered_selected_tail_general}
\end{align}
Consequently, for every \(\delta\in(0,1)\), with probability at least
\(1-\delta\),
\begin{equation}
  \left\|
    \widehat\Sigma_{l,N}^{\mathrm{tc}}-\Sigma_l(\rho)
  \right\|_{\mathrm{op}}
  \le
  \varepsilon_{l,N}^{\mathrm{tc}}(\delta),
  \label{eq:true_centered_selected_high_probability_general}
\end{equation}
where \(\varepsilon_{l,N}^{\mathrm{tc}}(\delta)\) is defined in
\eqref{eq:epsilon_tc_general}.
\end{proposition}

\begin{proof}
\noindent\textbf{Step 1: Centered matrix summands.}
Define
\begin{equation}
  \hat X_i
  :=
  \hat z_i\hat z_i^\top-\Sigma_l(\rho),
  \qquad
  i=1,\ldots,N.
  \label{eq:tc_general_Xi_def}
\end{equation}
Using \eqref{eq:selected_covariance_z_second_moment_general}, we have
\begin{equation}
  \mathbb E_\rho[\hat X_i]=0.
  \label{eq:tc_general_Xi_centered}
\end{equation}
The matrices \(\hat X_1,\ldots,\hat X_N\) are independent self-adjoint
\(l\times l\) random matrices.  Combining
\eqref{eq:tc_general_Xi_def} with
\eqref{eq:selected_true_centered_covariance_z_form_general}, we obtain
\begin{equation}
  \widehat\Sigma_{l,N}^{\mathrm{tc}}
  -
  \Sigma_l(\rho)
  =
  \frac1N
  \sum_{i=1}^N
  \hat X_i .
  \label{eq:tc_general_error_average_Xi}
\end{equation}

\noindent\textbf{Step 2: Almost-sure operator-norm bound.}
The relation \eqref{eq:selected_centered_radius_bound_general} implies 
\begin{equation}
  \|\hat z_i\hat z_i^\top\|_{\mathrm{op}}
  =
  \|\hat z_i\|_2^2
  \le
  A_l(\rho)^2 .
  \label{eq:tc_general_rank_one_norm_bound}
\end{equation}
Using \eqref{eq:selected_covariance_z_second_moment_general} and
\eqref{eq:tc_general_rank_one_norm_bound}, we also have
\begin{align}
  \|\Sigma_l(\rho)\|_{\mathrm{op}}
  &=
  \left\|
    \mathbb E_\rho[\hat z_i\hat z_i^\top]
  \right\|_{\mathrm{op}}
  \notag\\
  &\le
  \mathbb E_\rho
  \left[
    \|\hat z_i\hat z_i^\top\|_{\mathrm{op}}
  \right]
  \le
  A_l(\rho)^2 .
  \label{eq:tc_general_Sigma_l_op_bound}
\end{align}
Therefore, the combination of \eqref{eq:tc_general_Xi_def},
\eqref{eq:tc_general_rank_one_norm_bound}, and
\eqref{eq:tc_general_Sigma_l_op_bound} yields
\begin{equation}
  \|\hat X_i\|_{\mathrm{op}}
  \le
  2A_l(\rho)^2
  \label{eq:tc_general_Xi_norm_bound}
\end{equation}
almost surely.

\noindent\textbf{Step 3: Matrix variance proxy.}
We use the elementary self-adjoint inequality
\begin{equation}
  (A-B)^2
  \preceq
  2A^2+2B^2,
  \label{eq:tc_general_square_ineq}
\end{equation}
valid for self-adjoint matrices \(A\) and \(B\).  Applying
\eqref{eq:tc_general_square_ineq} with
$  A=\hat z_i\hat z_i^\top$ and $  B=\Sigma_l(\rho)$,
we obtain
\begin{equation}
  \hat X_i^2
  \preceq
  2(\hat z_i\hat z_i^\top)^2
  +
  2\Sigma_l(\rho)^2 .
  \label{eq:tc_general_Xi_square_bound}
\end{equation}
The combination of \eqref{eq:selected_centered_radius_bound_general} and
the relation $  (\hat z_i\hat z_i^\top)^2 =  \|\hat z_i\|_2^2\hat z_i\hat z_i^\top $
gives
\begin{equation}
  (\hat z_i\hat z_i^\top)^2
  \preceq
  A_l(\rho)^2
  \hat z_i\hat z_i^\top .
  \label{eq:tc_general_rank_one_square_bound}
\end{equation}
Taking expectation in \eqref{eq:tc_general_Xi_square_bound} and using
\eqref{eq:tc_general_rank_one_square_bound} together with
\eqref{eq:selected_covariance_z_second_moment_general}, we get
\begin{equation}
  \mathbb E_\rho[\hat X_i^2]
  \preceq
  2A_l(\rho)^2\Sigma_l(\rho)
  +
  2\Sigma_l(\rho)^2 .
  \label{eq:tc_general_EXi_square_first_bound}
\end{equation}

Since \(\Sigma_l(\rho)\) is positive semidefinite and
\eqref{eq:selected_centered_radius_bound_general} holds almost surely, we have
\begin{equation}
  0
  \preceq
  \Sigma_l(\rho)
  =
  \mathbb E_\rho[\hat z_i\hat z_i^\top]
  \preceq
  A_l(\rho)^2 I_l,
  \label{eq:tc_general_Sigma_l_loewner_radius_bound}
\end{equation}
which implies
\begin{equation}
  \Sigma_l(\rho)^2
  \preceq
  A_l(\rho)^2\Sigma_l(\rho)
  \preceq
  A_l(\rho)^4 I_l.
  \label{eq:tc_general_Sigma_l_square_bound}
\end{equation}
Combining \eqref{eq:tc_general_EXi_square_first_bound}, 
\eqref{eq:tc_general_Sigma_l_loewner_radius_bound}, and
\eqref{eq:tc_general_Sigma_l_square_bound}, 
we obtain
\begin{equation}
  \mathbb E_\rho[\hat X_i^2]
  \preceq
  4A_l(\rho)^4 I_l .
  \label{eq:tc_general_EXi_square_final_bound}
\end{equation}
Summing \eqref{eq:tc_general_EXi_square_final_bound} over
\(i=1,\ldots,N\), we have
\begin{equation}
  \sum_{i=1}^N
  \mathbb E_\rho[\hat X_i^2]
  \preceq
  4N A_l(\rho)^4 I_l,
  \label{eq:tc_general_variance_loewner_bound}
\end{equation}
and hence
\begin{equation}
  \left\|
    \sum_{i=1}^N
    \mathbb E_\rho[\hat X_i^2]
  \right\|_{\mathrm{op}}
  \le
  4N A_l(\rho)^4.
  \label{eq:tc_general_variance_proxy_bound}
\end{equation}

\noindent\textbf{Step 4: Matrix Bernstein inequality.}
We now recall the self-adjoint matrix Bernstein inequality
\cite[(ii) of Theorem~6.1]{Tropp2012UserFriendly}.  Let
\(\hat Y_1,\ldots,\hat Y_N\) be independent self-adjoint random matrices
of dimension \(d\) satisfying
\begin{equation}
  \mathbb E[\hat Y_i]=0,
  \qquad
  \|\hat Y_i\|_{\mathrm{op}}\le L
  \quad\text{almost surely}.
  \label{eq:tc_general_matrix_bernstein_assumptions}
\end{equation}
Define the variance parameter
\begin{equation}
  \sigma^2
  :=
  \left\|
    \sum_{i=1}^N
    \mathbb E[\hat Y_i^2]
  \right\|_{\mathrm{op}}.
  \label{eq:tc_general_matrix_bernstein_variance_parameter}
\end{equation}
Then, for every \(s\ge0\),
\begin{equation}
  \mathbb P\!\left(
    \lambda_{\max}\!\left(
      \sum_{i=1}^N \hat Y_i
    \right)
    \ge s
  \right)
  \le
  d\,
  \exp\!\left(
    -
    \frac{s^2}{
      2\sigma^2+\frac{2Ls}{3}
    }
  \right).
  \label{eq:tc_general_matrix_bernstein_statement}
\end{equation}

We apply \eqref{eq:tc_general_matrix_bernstein_statement} to the centered
self-adjoint matrices
\(
  \hat X_1,\ldots,\hat X_N .
\)
By \eqref{eq:tc_general_Xi_centered}, these matrices are centered.
Since the matrices have dimension \(l\), substituting
\eqref{eq:tc_general_Xi_norm_bound} and
\eqref{eq:tc_general_variance_proxy_bound} into
\eqref{eq:tc_general_matrix_bernstein_statement} gives, for every \(s>0\),
\begin{align}
&  \mathbb P\!\left(
    \lambda_{\max}\!\left(
      \sum_{i=1}^N \hat X_i
    \right)
    \ge s
  \right) \notag\\
  \le &
  l
  \exp\!\left(
    -
    \frac{s^2}{
      2\cdot 4N A_l(\rho)^4
      +
      \frac{2}{3}\cdot 2A_l(\rho)^2 s
    }
  \right)
  \notag\\
  =&
  l
  \exp\!\left(
    -
    \frac{s^2}{
      8N A_l(\rho)^4
      +
      \frac{4}{3}A_l(\rho)^2s
    }
  \right).
  \label{eq:tc_general_Bernstein_upper_tail}
\end{align}

Applying the same bound to the centered self-adjoint matrices
\(-\hat X_1,\ldots,-\hat X_N\), we obtain
\begin{align}
&  \mathbb P\!\left(
    \lambda_{\max}\!\left(
      -\sum_{i=1}^N \hat X_i
    \right)
    \ge s
  \right)\notag\\
  \le&
  l
  \exp\!\left(
    -
    \frac{s^2}{
      8N A_l(\rho)^4
      +
      \frac{4}{3}A_l(\rho)^2s
    }
  \right).
  \label{eq:tc_general_Bernstein_lower_tail}
\end{align}
Therefore, by the union bound,
\begin{align}
&  \mathbb P\!\left(
    \left\|
      \sum_{i=1}^N \hat X_i
    \right\|_{\mathrm{op}}
    \ge s
  \right) \notag\\
  \le &
  \mathbb P\!\left(
    \lambda_{\max}\!\left(
      \sum_{i=1}^N \hat X_i
    \right)
    \ge s
  \right)
  +
  \mathbb P\!\left(
    \lambda_{\max}\!\left(
      -\sum_{i=1}^N \hat X_i
    \right)
    \ge s
  \right)
  \notag\\
  \le &
  2l
  \exp\!\left(
    -
    \frac{s^2}{
      8N A_l(\rho)^4
      +
      \frac{4}{3}A_l(\rho)^2s
    }
  \right).
  \label{eq:tc_general_Bernstein_two_sided_sum}
\end{align}

\noindent\textbf{Step 5: Tail bound for the empirical covariance.}
By the averaging identity \eqref{eq:tc_general_error_average_Xi}, the event
\(
  \left\|
    \widehat\Sigma_{l,N}^{\mathrm{tc}}
    -
    \Sigma_l(\rho)
  \right\|_{\mathrm{op}}
  \ge t
\)
is equivalent to
\(
  \left\|
    \sum_{i=1}^N \hat X_i
  \right\|_{\mathrm{op}}
  \ge Nt.
\)
Thus, substituting
$s=Nt$
into \eqref{eq:tc_general_Bernstein_two_sided_sum}, we obtain
\begin{align}
&  \mathbb P\!\left(
  \left\|
    \widehat\Sigma_{l,N}^{\mathrm{tc}}
    -
    \Sigma_l(\rho)
  \right\|_{\mathrm{op}}
  \ge t
  \right) \notag\\
  \le &
  2l
  \exp\!\left(
    -
    \frac{Nt^2}{
      8A_l(\rho)^4+\frac43A_l(\rho)^2t
    }
  \right).
  \label{eq:tc_general_tail_bound_final}
\end{align}
This proves the tail estimate
\eqref{eq:true_centered_selected_tail_general}.

\noindent{\bf Step 6:} Proof of \eqref{eq:true_centered_selected_high_probability_general}

Set
\begin{align*}
u &:= \log\!\left(\frac{2l}{\delta}\right), \\
t &:=
 \varepsilon_{l,N}^{\mathrm{tc}}(\delta)
=
 A_l(\rho)^2
 \left(
 \sqrt{\frac{8\log(2l/\delta)}{N}}
 +
 \frac{4\log(2l/\delta)}{3N}
 \right)\\
& =
A_l(\rho)^2
\left(
\sqrt{\frac{8u}{N}}
+
\frac{4u}{3N}
\right).
\end{align*}
We first record the following elementary estimate, whose verification is given below:
\begin{align}
\frac{Nt^2}{8A_l(\rho)^4+\frac43A_l(\rho)^2 t}
\ge u. \label{SAA1}
\end{align}

Substituting \eqref{SAA1} into \eqref{eq:tc_general_tail_bound_final} yields
\begin{equation}
\mathbb P\!\left(
\left\|
\widehat\Sigma_{l,N}^{\mathrm{tc}}
-
\Sigma_l(\rho)
\right\|_{\mathrm{op}}
\ge t
\right)
\le
2l e^{-u}
=
\delta.
\end{equation}
Hence, with probability at least \(1-\delta\), we have
\begin{equation}
\left\|
\widehat\Sigma_{l,N}^{\mathrm{tc}}
-
\Sigma_l(\rho)
\right\|_{\mathrm{op}}
\le
\varepsilon_{l,N}^{\mathrm{tc}}(\delta),
\label{eq:tc_bound_start}
\end{equation}
which proves \eqref{eq:true_centered_selected_high_probability_general}.

\noindent{\bf Step 7:} Proof of \eqref{SAA1}.

We put $a:=A_l(\rho)^2,
x:=\sqrt{\frac{8u}{N}},
y:=\frac{4u}{3N}$.
Then \(t=a(x+y)\). Hence
\[
\frac{Nt^2}{8A_l(\rho)^4+\frac43A_l(\rho)^2 t}
=
\frac{N a^2(x+y)^2}{8a^2+\frac43a^2(x+y)}
=
\frac{N(x+y)^2}{8+\frac43(x+y)}.
\]
Therefore, to prove \eqref{SAA1}, it suffices to
show
\begin{align}
N(x+y)^2
\ge
u\left(8+\frac43(x+y)\right).\label{BAS1}
\end{align}
Now, by the definitions of \(x\) and \(y\), we have
$Nx^2=8u, Ny=\frac{4u}{3}$.
Thus
\[
N(x+y)^2
=
Nx^2+2Nxy+Ny^2
=
8u+2Nxy+Ny^2,
\]
while
\[
u\left(8+\frac43(x+y)\right)
=
8u+\frac{4u}{3}x+\frac{4u}{3}y
=
8u+Nxy+Ny^2.
\]
Consequently,
\begin{align*}
N(x+y)^2
=&
8u+2Nxy+Ny^2
\ge
8u+Nxy+Ny^2\notag\\
=&
u\left(8+\frac43(x+y)\right).
\end{align*}
This proves \eqref{BAS1}, which implies
\eqref{SAA1}.
\end{proof}

\subsection{True-centered selected spectral approximation}

The concentration proposition gives an operator-norm perturbation bound for
the selected covariance matrix.  The next theorem packages the deterministic
spectral consequences of this perturbation.  
The two-sided matrix inequality and the operator-norm bound are immediate
once the deviation matrix is controlled in operator norm.
Eigenvalue stability follows
from Weyl's inequality for Hermitian matrices
\cite[Chapters~4 and~6]{HornJohnson2012}.  Stability of isolated spectral subspaces
is obtained from the Davis--Kahan theorem, in a standard gap-dependent form
\cite{DavisKahan1970,StewartSun1990,Vershynin2018,YuWangSamworth2015}.

This theorem is stated separately because it clarifies which part of the
argument is probabilistic and which part is deterministic.  Probability enters
only through the event on which
\(
  \|\widehat\Sigma_{l,N}^{\mathrm{tc}}-\Sigma_l(\rho)\|_{\mathrm{op}}
\)
is small.  Once that event holds, the eigenvalue and projector estimates are
ordinary perturbation theory.

\begin{theorem}[True-centered spectral approximation for selected shadow covariance]
\label{thm:selected_shadow_true_centered_spectral_recovery}
With probability at least \(1-\delta\), the following hold.

\begin{enumerate}
\item[(i)] \emph{True-centered Loewner approximation:}
\begin{equation}
  -\varepsilon_{l,N}^{\mathrm{tc}}(\delta)I_l
  \preceq
  \widehat\Sigma_{l,N}^{\mathrm{tc}}-\Sigma_l(\rho)
  \preceq
  \varepsilon_{l,N}^{\mathrm{tc}}(\delta)I_l .
  \label{eq:true_centered_loewner_general}
\end{equation}

\item[(ii)] \emph{True-centered operator-norm approximation:}
\begin{equation}
  \left\|
    \widehat\Sigma_{l,N}^{\mathrm{tc}}-\Sigma_l(\rho)
  \right\|_{\mathrm{op}}
  \le
  \varepsilon_{l,N}^{\mathrm{tc}}(\delta).
  \label{eq:true_centered_operator_general}
\end{equation}

\item[(iii)] \emph{True-centered selected eigenvalue approximation:}
For every \(j=1,\ldots,l\),
\begin{equation}
  \left|
    \lambda_j(\widehat\Sigma_{l,N}^{\mathrm{tc}})
    -
    \lambda_j(\Sigma_l(\rho))
  \right|
  \le
  \varepsilon_{l,N}^{\mathrm{tc}}(\delta).
  \label{eq:true_centered_eigenvalue_general}
\end{equation}

\item[(iv)] \emph{True-centered selected spectral-projector approximation:}
Let \(S\) be an isolated spectral cluster of \(\Sigma_l(\rho)\), and define
\begin{equation}
  \gamma_{S,l}
  :=
  \operatorname{dist}
  \bigl(
    S,
    \operatorname{spec}(\Sigma_l(\rho))\setminus S
  \bigr)>0 .
  \label{eq:true_centered_gap_definition_general}
\end{equation}
Let \(\Pi_{S,l}\) be the spectral projector of \(\Sigma_l(\rho)\) associated
with \(S\), and let \(\widehat\Pi_{S,l}^{\mathrm{tc}}\) be the spectral
projector of \(\widehat\Sigma_{l,N}^{\mathrm{tc}}\) associated with the
empirical eigenvalues corresponding to \(S\).  If
\begin{equation}
  \gamma_{S,l}
  >
  2\varepsilon_{l,N}^{\mathrm{tc}}(\delta),
  \label{eq:true_centered_gap_condition_general}
\end{equation}
then
\begin{equation}
  \left\|
    \widehat\Pi_{S,l}^{\mathrm{tc}}-\Pi_{S,l}
  \right\|_{\mathrm{op}}
  \le
  \frac{
    2\varepsilon_{l,N}^{\mathrm{tc}}(\delta)
  }{
    \gamma_{S,l}
  } .
  \label{eq:true_centered_projector_general}
\end{equation}
\end{enumerate}
\end{theorem}

\begin{proof}
\par

\noindent\textbf{Step 1: Operator-norm and two-sided matrix bounds.}
By Proposition~\ref{prop:true_centered_selected_shadow_covariance_concentration},
and in particular by the high-probability estimate
\eqref{eq:true_centered_selected_high_probability_general}, there exists an
event \(\mathcal E_{\mathrm{tc}}\) with
\(
  \mathbb P(\mathcal E_{\mathrm{tc}})\ge 1-\delta
\)
such that, on \(\mathcal E_{\mathrm{tc}}\),
\begin{align}
  \left\|
    \widehat\Sigma_{l,N}^{\mathrm{tc}}-\Sigma_l(\rho)
  \right\|_{\mathrm{op}}
  \le
  \varepsilon_{l,N}^{\mathrm{tc}}(\delta).\label{eq:true_centered_operator_bound_from_concentration_general}
\end{align}
This proves the operator-norm statement \eqref{eq:true_centered_operator_general}.  
Since
\(\widehat\Sigma_{l,N}^{\mathrm{tc}}-\Sigma_l(\rho)\) is self-adjoint, the same
operator-norm bound \eqref{eq:true_centered_operator_bound_from_concentration_general}
implies the two-sided matrix inequality
\begin{equation}
  -\varepsilon_{l,N}^{\mathrm{tc}}(\delta)I_l
  \preceq
  \widehat\Sigma_{l,N}^{\mathrm{tc}}-\Sigma_l(\rho)
  \preceq
  \varepsilon_{l,N}^{\mathrm{tc}}(\delta)I_l.
  \label{eq:true_centered_loewner_bound_from_operator_general}
\end{equation}
This proves \eqref{eq:true_centered_loewner_general}.

\noindent\textbf{Step 2: Eigenvalue perturbation.}
We recall Weyl's eigenvalue perturbation inequality for self-adjoint
matrices.  If \(A\) and \(B\) are \(l\times l\) self-adjoint matrices and
\(\lambda_1(\cdot)\le\cdots\le\lambda_l(\cdot)\) denote their ordered
eigenvalues, then
\begin{equation}
  \left|
    \lambda_j(A)-\lambda_j(B)
  \right|
  \le
  \|A-B\|_{\mathrm{op}},
  \qquad
  j=1,\ldots,l.
  \label{eq:weyl_inequality_recalled_general}
\end{equation}
See, for example, \cite[Chapters~4 and~6]{HornJohnson2012}.  Applying
\eqref{eq:weyl_inequality_recalled_general} with
$A=\widehat\Sigma_{l,N}^{\mathrm{tc}}$ and 
$B=\Sigma_l(\rho)$,
and using \eqref{eq:true_centered_operator_bound_from_concentration_general},
we obtain
\begin{equation}
  \left|
    \lambda_j(\widehat\Sigma_{l,N}^{\mathrm{tc}})
    -
    \lambda_j(\Sigma_l(\rho))
  \right|
  \le
  \varepsilon_{l,N}^{\mathrm{tc}}(\delta),
  \qquad
  j=1,\ldots,l.
  \label{eq:true_centered_eigenvalue_bound_from_weyl_general}
\end{equation}
This proves \eqref{eq:true_centered_eigenvalue_general}.

\noindent\textbf{Step 3: Spectral-projector perturbation.}
We use a standard finite-dimensional Davis--Kahan perturbation bound; see
Vershynin~\cite[Theorem~4.5.5]{Vershynin2018} or
Yu--Wang--Samworth~\cite[Theorem~1]{YuWangSamworth2015}.
The original source is Davis--Kahan~\cite{DavisKahan1970}.
Let \(A\) and
\(\widehat A\) be self-adjoint matrices.  Let \(S\) be an isolated spectral
cluster of \(A\), and let
\begin{equation}
  \gamma  :=
  \operatorname{dist}
  \bigl(
    S,
    \operatorname{spec}(A)\setminus S
  \bigr)>0.
  \label{eq:davis_kahan_gap_recalled_general}
\end{equation}
Let \(\Pi_S\) and \(\widehat\Pi_S\) be the spectral projectors of \(A\) and
\(\widehat A\) associated with the corresponding spectral clusters.  If
\begin{equation}
  \|\widehat A-A\|_{\mathrm{op}}<\frac{\gamma}{2},
  \label{eq:davis_kahan_gap_condition_recalled_general}
\end{equation}
then
\begin{equation}
  \|\widehat\Pi_S-\Pi_S\|_{\mathrm{op}}
  \le
  \frac{2\|\widehat A-A\|_{\mathrm{op}}}{\gamma}.
  \label{eq:davis_kahan_projector_bound_recalled_general}
\end{equation}

We apply \eqref{eq:davis_kahan_projector_bound_recalled_general} with
\begin{equation}
  A=\Sigma_l(\rho),
  \qquad
  \widehat A=\widehat\Sigma_{l,N}^{\mathrm{tc}},
  \qquad
  \gamma=\gamma_{S,l}.
  \label{eq:davis_kahan_application_true_centered_general}
\end{equation}
The gap condition
\eqref{eq:true_centered_gap_condition_general}, together with
\eqref{eq:true_centered_operator_general}, implies
\(
  \|\widehat A-A\|_{\mathrm{op}}
  \le
  \varepsilon_{l,N}^{\mathrm{tc}}(\delta)
  <
  \frac{\gamma_{S,l}}{2}.
\)
Therefore, by \eqref{eq:davis_kahan_projector_bound_recalled_general},
\(
  \left\|
    \widehat\Pi_{S,l}^{\mathrm{tc}}-\Pi_{S,l}
  \right\|_{\mathrm{op}}
  \le
  \frac{
    2\varepsilon_{l,N}^{\mathrm{tc}}(\delta)
  }{
    \gamma_{S,l}
  },
\)
which is exactly \eqref{eq:true_centered_projector_general}.
\end{proof}

\begin{remark}[Scope of the true-centered selected theorem]
\label{rem:true_centered_selected_scope_general}
The theorem controls only the selected covariance matrix
\begin{equation}
  \Sigma_l(\rho)
  =
  R_l\Sigma(\rho)R_l^\top
  \label{eq:true_centered_scope_selected_covariance_general}
\end{equation}
and its true-centered empirical approximation.  It does not assert
operator-norm recovery of the full \(p\times p\) covariance matrix
\(\Sigma(\rho)\).  The dimension factor in the concentration estimate is the
selected dimension \(l\), and the deterministic radius is the selected
centered radius \(A_l(\rho)\).
\end{remark}

The true-centered covariance is an intermediate object.  The next subsection
transfers the preceding bounds to the sample-centered selected covariance by
using the exact rank-one centering relation and a concentration bound for the
selected empirical mean.

\subsection{Centering correction for selected empirical covariance}
\label{subsec:centering_correction_selected_covariance}

We first record the exact algebraic relation between the true-centered and
sample-centered selected empirical covariance matrices.  This step is
independent of the shadow protocol.  It only uses the definitions of empirical
centering and the fixed selection matrix \(R_l\).

Recall that
$  \hat z_i
  =
  R_l(\hat x_i-m)$ for 
  $i=1,\ldots,N$,
as defined in \eqref{eq:selected_centered_vector_general}.  The selected
empirical mean of the centered outputs is
\begin{equation}
  \bar{\hat z}_N
  :=
  R_l(\bar{\hat x}_N-m)
  =
  \frac1N\sum_{i=1}^N \hat z_i .
  \label{eq:selected_empirical_mean_z_general}
\end{equation}
The point of this subsection is that replacing the unknown mean
\(m\) by the empirical mean \(\bar{\hat x}_N\) introduces a correction that is
not arbitrary: after selection, it is exactly the negative rank-one matrix
\(-\bar{\hat z}_N\bar{\hat z}_N^\top\).

\begin{lemma}[Compressed centering identity]
\label{lem:compressed_centering_identity_general}
The selected true-centered and sample-centered empirical covariance matrices
satisfy
\begin{equation}
  \widehat\Sigma_{l,N}^{\mathrm{sc}}
  -
  \widehat\Sigma_{l,N}^{\mathrm{tc}}
  =
  -
  \bar{\hat z}_N\bar{\hat z}_N^\top .
  \label{eq:compressed_centering_identity_general}
\end{equation}
Consequently,
\begin{equation}
  \widehat\Sigma_{l,N}^{\mathrm{sc}}
  -
  \Sigma_l(\rho)
  =
  \left(
  \widehat\Sigma_{l,N}^{\mathrm{tc}}
  -
  \Sigma_l(\rho)
  \right)
  -
  \bar{\hat z}_N\bar{\hat z}_N^\top .
  \label{eq:compressed_sample_centered_error_decomposition_general}
\end{equation}
\end{lemma}

\begin{proof}
The full empirical covariance matrices satisfy the exact rank-one identity
\begin{equation}
  \widehat\Sigma_N^{\mathrm{sc}}
  -
  \widehat\Sigma_N^{\mathrm{tc}}
  =
  -
  (\bar{\hat x}_N-m)(\bar{\hat x}_N-m)^\top ,
  \label{eq:compressed_centering_full_identity_recalled_general}
\end{equation}
which is the identity stated in
\eqref{eq:exact_rank_one_relation_true_sample_centered_general}.  Multiplying
\eqref{eq:compressed_centering_full_identity_recalled_general} by \(R_l\) on
the left and by \(R_l^\top\) on the right, and using the definitions
\eqref{eq:fs_general_selected_sample_centered_def} and
\eqref{eq:fs_general_selected_true_centered_def}, gives
\begin{align}
&  \widehat\Sigma_{l,N}^{\mathrm{sc}}
  -
  \widehat\Sigma_{l,N}^{\mathrm{tc}}
  =
  R_l
  \left(
  \widehat\Sigma_N^{\mathrm{sc}}
  -
  \widehat\Sigma_N^{\mathrm{tc}}
  \right)
  R_l^\top
  \notag\\
  =&
  -
  R_l(\bar{\hat x}_N-m)(\bar{\hat x}_N-m)^\top R_l^\top
  \notag\\
  =&
  -
  \bigl(R_l(\bar{\hat x}_N-m)\bigr)
  \bigl(R_l(\bar{\hat x}_N-m)\bigr)^\top
  =
  -
  \bar{\hat z}_N\bar{\hat z}_N^\top ,
  \label{eq:compressed_centering_identity_proof_general}
\end{align}
where the last equality uses \eqref{eq:selected_empirical_mean_z_general}.
This proves \eqref{eq:compressed_centering_identity_general}.  Subtracting
\(\Sigma_l(\rho)\) from both sides of
\eqref{eq:compressed_centering_identity_general} gives
\eqref{eq:compressed_sample_centered_error_decomposition_general}.
\end{proof}

It remains to control the rank-one correction in
\eqref{eq:compressed_sample_centered_error_decomposition_general}.  
This evaluation reduces to bounding the selected empirical mean
because 
\(
  \left\|
  \bar{\hat z}_N\bar{\hat z}_N^\top
  \right\|_{\mathrm{op}}
  =
  \|\bar{\hat z}_N\|_2^2\).
This remaining issue is done in the next subsection.

\subsection{Empirical mean control and sample-centering transfer}
\label{subsec:empirical_mean_transfer_sample_centering}

The centering identity in
Lemma~\ref{lem:compressed_centering_identity_general} reduces the passage
from true-centered to sample-centered covariance to a bound on
\(\|\bar{\hat z}_N\|_2^2\).  We control this quantity using the selected
centered-radius bound
\eqref{eq:selected_centered_radius_bound_general}, together with coordinatewise Hoeffding
inequalities and a union bound over the \(l\) selected coordinates.

We now bound the selected empirical mean.  The rank-one correction is
controlled by its only nonzero eigenvalue,
\(
  \|\bar{\hat z}_N\bar{\hat z}_N^\top\|_{\mathrm{op}}
  =
  \|\bar{\hat z}_N\|_2^2 .
\)
The coordinates of \(\hat z_i\) need not be independent within a single shot.
We therefore use only independence across shots: each selected coordinate
average is a bounded scalar average.  Hoeffding's inequality is applied
coordinatewise, and a union bound over the \(l\) selected coordinates gives
the Euclidean-norm bound.

\begin{lemma}[Selected empirical-mean bound]
\label{lem:selected_empirical_mean_bound_general}
For every \(\delta\in(0,1)\), with probability at least \(1-\delta\),
\begin{equation}
  \|\bar{\hat z}_N\|_2^2
  \le
  \mu_{l,N}(\delta),
  \label{eq:selected_empirical_mean_bound_general}
\end{equation}
where \(\mu_{l,N}(\delta)\) is defined in \eqref{eq:mu_general}.  Equivalently,
the same event satisfies the relation
\begin{equation}
  \bar{\hat z}_N\bar{\hat z}_N^\top
  \preceq
  \mu_{l,N}(\delta) I_l .
  \label{eq:selected_empirical_mean_loewner_bound_general}
\end{equation}
\end{lemma}

\begin{proof}
Write
\begin{equation}
  \hat z_i
  =
  (\hat z_{i,1},\ldots,\hat z_{i,l})^\top,
  \qquad
  \bar{\hat z}_N
  =
  (\bar{\hat z}_{N,1},\ldots,\bar{\hat z}_{N,l})^\top .
  \label{eq:selected_empirical_mean_coordinates_general}
\end{equation}
The relation \eqref{eq:selected_centered_vector_mean_zero_general}
implies
\begin{equation}
  \mathbb E_\rho[\hat z_{i,a}]=0,
  \qquad
  a=1,\ldots,l.
  \label{eq:selected_empirical_mean_coordinate_centered_general}
\end{equation}
Moreover, the relation \eqref{eq:selected_centered_radius_bound_general} in Lemma~\ref{lem:selected_centered_radius_general} implies
\begin{equation}
  |\hat z_{i,a}|
  \le
  \|\hat z_i\|_2
  =
  \|R_l(\hat x_i-m)\|_2
  \le
  A_l(\rho),
  \qquad
  a=1,\ldots,l.
  \label{eq:selected_empirical_mean_coordinate_bound_general}
\end{equation}

We recall Hoeffding's inequality
\cite[Theorem 2]{Hoeffding1963},\cite{Vershynin2018}.  If
\(\hat Y_1,\ldots,\hat Y_N\) are independent real random variables satisfying
\(a_i\le \hat Y_i\le b_i\) almost surely, then for every \(t>0\),
\begin{equation}
  \mathbb P\!\left(
  \left|
  \sum_{i=1}^N
  \bigl(\hat Y_i-\mathbb E[\hat Y_i]\bigr)
  \right|
  \ge t
  \right)
  \le
  2\exp\!\left(
  -
  \frac{2t^2}{
  \sum_{i=1}^N (b_i-a_i)^2
  }
  \right).
  \label{eq:hoeffding_inequality_recalled_general}
\end{equation}
Apply \eqref{eq:hoeffding_inequality_recalled_general} to
\(\hat Y_i=\hat z_{i,a}\).  Due to
\eqref{eq:selected_empirical_mean_coordinate_centered_general} and
\eqref{eq:selected_empirical_mean_coordinate_bound_general}, 
every \(s>0\) satisfies
\begin{align}
  \mathbb P\!\left(
  |\bar{\hat z}_{N,a}|
  \ge s
  \right)
  &=
  \mathbb P\!\left(
  \left|
  \frac1N\sum_{i=1}^N \hat z_{i,a}
  \right|
  \ge s
  \right)
  \notag\\
  &\le
  2\exp\!\left(
  -
  \frac{Ns^2}{2A_l(\rho)^2}
  \right).
  \label{eq:selected_empirical_mean_coordinate_tail_general}
\end{align}
Taking a union bound over \(a=1,\ldots,l\), we obtain
\begin{equation}
  \mathbb P\!\left(
  \max_{1\le a\le l}
  |\bar{\hat z}_{N,a}|
  \ge s
  \right)
  \le
  2l
  \exp\!\left(
  -
  \frac{Ns^2}{2A_l(\rho)^2}
  \right).
  \label{eq:selected_empirical_mean_union_bound_general}
\end{equation}
Choose
$  s
  :=
  A_l(\rho)
  \sqrt{
  \frac{2\log(2l/\delta)}{N}
  }$.
Then the right-hand side of
\eqref{eq:selected_empirical_mean_union_bound_general} is equal to
\(\delta\).  Hence, with probability at least \(1-\delta\), 
the relation
$  |\bar{\hat z}_{N,a}|
  \le
  A_l(\rho)
  \sqrt{
  \frac{2\log(2l/\delta)}{N}
  }$ holds for 
  $a=1,\ldots,l$.
This event satisfies
\begin{align}
  \|\bar{\hat z}_N\|_2^2
  &=
  \sum_{a=1}^l
  \bar{\hat z}_{N,a}^2
  \le
  l\,
  A_l(\rho)^2
  \frac{2\log(2l/\delta)}{N}
  \notag\\
  &=
  \frac{
  2lA_l(\rho)^2\log(2l/\delta)
  }{N}
  =
  \mu_{l,N}(\delta),
  \label{eq:selected_empirical_mean_bound_proof_general}
\end{align}
where the last equality is exactly the definition
\eqref{eq:mu_general}, which proves
\eqref{eq:selected_empirical_mean_bound_general}.

Finally, any \(v\in\mathbb R^l\) satisfies
\begin{align*}
  v^\top
  \bar{\hat z}_N\bar{\hat z}_N^\top
  v
  =
  \langle v,\bar{\hat z}_N\rangle^2
  \le
  \|v\|_2^2\|\bar{\hat z}_N\|_2^2
  \le
  \mu_{l,N}(\delta)\|v\|_2^2.
\end{align*}
Thus, we have
$  \bar{\hat z}_N\bar{\hat z}_N^\top
  \preceq
  \mu_{l,N}(\delta)I_l$,
which proves \eqref{eq:selected_empirical_mean_loewner_bound_general}.
\end{proof}

The next lemma records how the true-centered covariance bound is converted
into a sample-centered covariance bound.  The conversion uses only the
rank-one centering identity from
Lemma~\ref{lem:compressed_centering_identity_general}.  Since sample centering
subtracts the positive semidefinite matrix
\(\bar{\hat z}_N\bar{\hat z}_N^\top\), the upper matrix inequality is
unchanged, while the lower matrix inequality loses the additional
empirical-mean term \(\mu\).  This asymmetry is later absorbed into the
symmetric operator-norm bound.

\begin{lemma}[Transfer from true-centered to sample-centered covariance]
\label{lem:loewner_transfer_true_to_sample_general}
Assume that there exist real numbers \(\varepsilon,\mu\ge0\) such that
\begin{align}
  -\varepsilon I_l
  \preceq &
  \widehat\Sigma_{l,N}^{\mathrm{tc}}-\Sigma_l(\rho)
  \preceq
  \varepsilon I_l
  \label{eq:loewner_transfer_tc_assumption_general} \\
  \bar{\hat z}_N\bar{\hat z}_N^\top
  \preceq &
  \mu I_l .
  \label{eq:loewner_transfer_mean_assumption_general}
\end{align}
Then, the relation
\begin{equation}
  -(\varepsilon+\mu)I_l
  \preceq
  \widehat\Sigma_{l,N}^{\mathrm{sc}}-\Sigma_l(\rho)
  \preceq
  \varepsilon I_l 
  \label{eq:loewner_transfer_sc_bound_general}
\end{equation}
holds.
Consequently, the inequality
\begin{equation}
  \left\|
  \widehat\Sigma_{l,N}^{\mathrm{sc}}-\Sigma_l(\rho)
  \right\|_{\mathrm{op}}
  \le
  \varepsilon+\mu 
  \label{eq:loewner_transfer_operator_bound_general}
\end{equation}
holds.
\end{lemma}

\begin{proof}
The sample-centered error decomposition
\eqref{eq:compressed_sample_centered_error_decomposition_general} yields
\begin{equation}
  \widehat\Sigma_{l,N}^{\mathrm{sc}}
  -
  \Sigma_l(\rho)
  =
  \left(
  \widehat\Sigma_{l,N}^{\mathrm{tc}}
  -
  \Sigma_l(\rho)
  \right)
  -
  \bar{\hat z}_N\bar{\hat z}_N^\top .
  \label{eq:loewner_transfer_decomposition_recalled_general}
\end{equation}
Since
\(\bar{\hat z}_N\bar{\hat z}_N^\top\succeq0\), the upper bound in
\eqref{eq:loewner_transfer_tc_assumption_general} gives
\begin{align}
  \widehat\Sigma_{l,N}^{\mathrm{sc}}
  -
  \Sigma_l(\rho)
  &=
  \left(
  \widehat\Sigma_{l,N}^{\mathrm{tc}}
  -
  \Sigma_l(\rho)
  \right)
  -
  \bar{\hat z}_N\bar{\hat z}_N^\top
  \notag\\
  &\preceq
  \widehat\Sigma_{l,N}^{\mathrm{tc}}
  -
  \Sigma_l(\rho)
  \preceq
  \varepsilon I_l,
  \label{eq:loewner_transfer_upper_bound_general}
\end{align}
which proves the upper side of
\eqref{eq:loewner_transfer_sc_bound_general}.

For the lower bound, the lower side of
\eqref{eq:loewner_transfer_tc_assumption_general} gives
\(
  \widehat\Sigma_{l,N}^{\mathrm{tc}}
  -
  \Sigma_l(\rho)
  \succeq
  -\varepsilon I_l,
\)
while \eqref{eq:loewner_transfer_mean_assumption_general} implies
\(
  -\bar{\hat z}_N\bar{\hat z}_N^\top
  \succeq
  -\mu I_l.
\)
Using \eqref{eq:loewner_transfer_decomposition_recalled_general}, we therefore
obtain
\begin{equation}
  \widehat\Sigma_{l,N}^{\mathrm{sc}}
  -
  \Sigma_l(\rho)
  \succeq
  -(\varepsilon+\mu)I_l.
  \label{eq:loewner_transfer_lower_bound_general}
\end{equation}
Combining
\eqref{eq:loewner_transfer_upper_bound_general} and
\eqref{eq:loewner_transfer_lower_bound_general} proves
\eqref{eq:loewner_transfer_sc_bound_general}.

Finally, \eqref{eq:loewner_transfer_sc_bound_general} implies that every
eigenvalue of
\(\widehat\Sigma_{l,N}^{\mathrm{sc}}-\Sigma_l(\rho)\) lies in
\([-(\varepsilon+\mu),\varepsilon]\).  Since
\(\varepsilon\le \varepsilon+\mu\), we get the relation
\(
  \left\|
  \widehat\Sigma_{l,N}^{\mathrm{sc}}-\Sigma_l(\rho)
  \right\|_{\mathrm{op}}
  \le
  \varepsilon+\mu,
\)
which proves \eqref{eq:loewner_transfer_operator_bound_general}.
\end{proof}

\subsection{Sample-centered selected spectral approximation}
\label{subsec:sample_centered_selected_spectral_approximation}

We now combine the true-centered spectral approximation with the centering
correction estimates.  The resulting theorem gives the detailed
sample-centered operator-norm, eigenvalue, and spectral-projector bounds.
This result is stronger than what is needed for the constant-radius main
theorem, but it is useful because it makes explicit the spectral consequences
of the selected covariance approximation.

\begin{theorem}[Sample-centered spectral approximation for selected shadow covariance]
\label{thm:selected_shadow_sample_centered_spectral_recovery}
Let \(\hat x_1,\ldots,\hat x_N\) be independent copies of the single-shot
shadow-output vector \(\hat x\in\mathbb R^p\).  Let
\(R_l\in\{0,1\}^{l\times p}\) be a deterministic coordinate-selection matrix
with \(R_lR_l^\top=I_l\), fixed independently of the measurement data.
Assume that the selected one-shot radius \(B_l\) in
\eqref{eq:selected_one_shot_radius_general} is finite, and define
\(A_l(\rho)\) by \eqref{eq:selected_centered_radius_general}.  Then, with
probability at least \(1-\delta\), the following hold.
\begin{enumerate}
\item[(i)] \emph{Two-sided semidefinite-order approximation:}
The relation
\begin{equation}
  -\varepsilon_{l,N}^{\mathrm{sc}}(\delta)I_l
  \preceq
  \widehat\Sigma_{l,N}^{\mathrm{sc}}-\Sigma_l(\rho)
  \preceq
  \varepsilon_{l,N}^{\mathrm{tc}}(\delta/2)I_l 
  \label{eq:selected_sc_loewner_general}
\end{equation}
holds.
\item[(ii)] \emph{Sample-centered operator-norm approximation:}
The inequality
\begin{equation}
  \left\|
    \widehat\Sigma_{l,N}^{\mathrm{sc}}-\Sigma_l(\rho)
  \right\|_{\mathrm{op}}
  \le
  \varepsilon_{l,N}^{\mathrm{sc}}(\delta)
  \label{eq:selected_sc_operator_general}
\end{equation}
holds.
\item[(iii)] \emph{Sample-centered selected eigenvalue approximation:}
Let
\begin{equation}
  \lambda_1(M)\le\cdots\le\lambda_l(M)
  \label{eq:ordered_eigenvalues_convention_general}
\end{equation}
denote the ordered eigenvalues of an \(l\times l\) self-adjoint matrix
\(M\).  Then, for every \(j=1,\ldots,l\),
\begin{equation}
  \left|
    \lambda_j(\widehat\Sigma_{l,N}^{\mathrm{sc}})
    -
    \lambda_j(\Sigma_l(\rho))
  \right|
  \le
  \varepsilon_{l,N}^{\mathrm{sc}}(\delta).
  \label{eq:selected_sc_eigenvalue_general}
\end{equation}

\item[(iv)] \emph{Sample-centered selected spectral-projector approximation:}
Let \(S\) be an isolated spectral cluster of \(\Sigma_l(\rho)\), and let
\begin{equation}
  \gamma_{S,l}
  :=
  \operatorname{dist}
  \bigl(
    S,
    \operatorname{spec}(\Sigma_l(\rho))\setminus S
  \bigr)>0 .
  \label{eq:selected_sc_gap_definition_general}
\end{equation}
Let \(\Pi_{S,l}\) be the spectral projector of \(\Sigma_l(\rho)\) associated
with \(S\), and let \(\widehat\Pi_{S,l}^{\mathrm{sc}}\) be the spectral
projector of \(\widehat\Sigma_{l,N}^{\mathrm{sc}}\) associated with the
empirical eigenvalues corresponding to \(S\).  If
\begin{equation}
  \gamma_{S,l}
  >
  2\varepsilon_{l,N}^{\mathrm{sc}}(\delta),
  \label{eq:selected_sc_gap_condition_general}
\end{equation}
then
\begin{equation}
  \left\|
    \widehat\Pi_{S,l}^{\mathrm{sc}}-\Pi_{S,l}
  \right\|_{\mathrm{op}}
  \le
  \frac{
    2\varepsilon_{l,N}^{\mathrm{sc}}(\delta)
  }{
    \gamma_{S,l}
  } .
  \label{eq:selected_sc_projector_general}
\end{equation}
\end{enumerate}
\end{theorem}

The proof combines the three estimates established above: the true-centered
spectral approximation, the Hoeffding bound for the selected empirical mean,
and the exact rank-one identity relating the true-centered and
sample-centered covariances.  A union bound allocates half of the failure
probability to the true-centered event and half to the empirical-mean event.
Once the operator-norm bound for the sample-centered covariance is obtained,
the eigenvalue and spectral-projector statements again follow from Weyl's
inequality~\cite[Chapters~4 and~6]{HornJohnson2012} and the Davis--Kahan
perturbation theorem~\cite{DavisKahan1970,StewartSun1990,Vershynin2018,YuWangSamworth2015}.

\begin{proof}[Proof of Theorem~\ref{thm:selected_shadow_sample_centered_spectral_recovery}]
\noindent\textbf{Step 1: True-centered event.}
Apply Theorem~\ref{thm:selected_shadow_true_centered_spectral_recovery} with
failure probability \(\delta/2\).  With probability at least \(1-\delta/2\), we have
\begin{equation}
  -\varepsilon_{l,N}^{\mathrm{tc}}(\delta/2)I_l
  \preceq
  \widehat\Sigma_{l,N}^{\mathrm{tc}}-\Sigma_l(\rho)
  \preceq
  \varepsilon_{l,N}^{\mathrm{tc}}(\delta/2)I_l.
  \label{eq:sample_centered_true_event_general}
\end{equation}
This is \eqref{eq:true_centered_loewner_general} with \(\delta\) replaced by
\(\delta/2\).

\noindent\textbf{Step 2: Selected empirical-mean event.}
Apply Lemma~\ref{lem:selected_empirical_mean_bound_general} with failure
probability \(\delta/2\).  With probability at least \(1-\delta/2\), we have
\begin{equation}
  \bar{\hat z}_N\bar{\hat z}_N^\top
  \preceq
  \mu_{l,N}(\delta/2)I_l.
  \label{eq:sample_centered_mean_event_general}
\end{equation}

\noindent\textbf{Step 3: Union bound.}
By the union bound, the events
\eqref{eq:sample_centered_true_event_general} and
\eqref{eq:sample_centered_mean_event_general} hold simultaneously with
probability at least
\begin{equation}
  1-\frac{\delta}{2}-\frac{\delta}{2}
  =
  1-\delta.
  \label{eq:sample_centered_union_probability_general}
\end{equation}
We work on this simultaneous event.

\noindent\textbf{Step 4: Sample-centered Loewner approximation.}
Apply Lemma~\ref{lem:loewner_transfer_true_to_sample_general} with
\begin{equation}
  \varepsilon
  =
  \varepsilon_{l,N}^{\mathrm{tc}}(\delta/2),
  \qquad
  \mu
  =
  \mu_{l,N}(\delta/2).
  \label{eq:sample_centered_transfer_parameters_general}
\end{equation}
Using \eqref{eq:epsilon_sc_def_general}, we have
\begin{equation}
  \varepsilon+\mu
  =
  \varepsilon_{l,N}^{\mathrm{tc}}(\delta/2)
  +
  \mu_{l,N}(\delta/2)
  =
  \varepsilon_{l,N}^{\mathrm{sc}}(\delta).
  \label{eq:sample_centered_error_identity_general}
\end{equation}
Therefore, \eqref{eq:loewner_transfer_sc_bound_general} gives
\begin{equation}
  -\varepsilon_{l,N}^{\mathrm{sc}}(\delta)I_l
  \preceq
  \widehat\Sigma_{l,N}^{\mathrm{sc}}-\Sigma_l(\rho)
  \preceq
  \varepsilon_{l,N}^{\mathrm{tc}}(\delta/2)I_l.
  \label{eq:sample_centered_loewner_bound_proof_general}
\end{equation}
This proves the sample-centered Loewner approximation
\eqref{eq:selected_sc_loewner_general}.  Notice that the upper side remains
the true-centered error level because the centering correction in
\eqref{eq:compressed_sample_centered_error_decomposition_general} is negative
semidefinite.

\noindent\textbf{Step 5: Sample-centered operator-norm approximation.}
Since the relation \eqref{eq:epsilon_sc_def_general} implies
\[
  \varepsilon_{l,N}^{\mathrm{tc}}(\delta/2)
  \le
  \varepsilon_{l,N}^{\mathrm{sc}}(\delta),
\]
the two-sided Loewner bound
\eqref{eq:sample_centered_loewner_bound_proof_general} implies
\begin{equation}
  \left\|
    \widehat\Sigma_{l,N}^{\mathrm{sc}}-\Sigma_l(\rho)
  \right\|_{\mathrm{op}}
  \le
  \varepsilon_{l,N}^{\mathrm{sc}}(\delta).
  \label{eq:sample_centered_operator_bound_proof_general}
\end{equation}
This proves \eqref{eq:selected_sc_operator_general}.

\noindent\textbf{Step 6: Sample-centered eigenvalue approximation.}
We recall Weyl's eigenvalue perturbation inequality for self-adjoint matrices:
if \(A\) and \(B\) are \(l\times l\) self-adjoint matrices and
\(\lambda_1(\cdot)\le\cdots\le\lambda_l(\cdot)\) denote their ordered
eigenvalues, then
\begin{equation}
  \left|
    \lambda_j(A)-\lambda_j(B)
  \right|
  \le
  \|A-B\|_{\mathrm{op}},
  \qquad
  j=1,\ldots,l.
  \label{eq:sample_centered_weyl_inequality_recalled_general}
\end{equation}
See, for example, \cite[Chapters~4 and~6]{HornJohnson2012}.  Applying
\eqref{eq:sample_centered_weyl_inequality_recalled_general} with
\begin{equation}
  A=\widehat\Sigma_{l,N}^{\mathrm{sc}},
  \qquad
  B=\Sigma_l(\rho),
  \label{eq:sample_centered_weyl_application_general}
\end{equation}
and using \eqref{eq:sample_centered_operator_bound_proof_general}, we obtain
\begin{equation}
  \left|
    \lambda_j(\widehat\Sigma_{l,N}^{\mathrm{sc}})
    -
    \lambda_j(\Sigma_l(\rho))
  \right|
  \le
  \varepsilon_{l,N}^{\mathrm{sc}}(\delta),
  \qquad
  j=1,\ldots,l.
  \label{eq:sample_centered_eigenvalue_bound_proof_general}
\end{equation}
This proves \eqref{eq:selected_sc_eigenvalue_general}.

\noindent\textbf{Step 7: Sample-centered spectral-projector approximation.}
We recall the Davis--Kahan spectral perturbation theorem
\cite{DavisKahan1970,StewartSun1990,Vershynin2018,YuWangSamworth2015}
stated as \eqref{eq:davis_kahan_projector_bound_recalled_general}.
We apply
\eqref{eq:davis_kahan_projector_bound_recalled_general}
with
\begin{equation}
  A=\Sigma_l(\rho),
  \qquad
  \widehat A=\widehat\Sigma_{l,N}^{\mathrm{sc}},
  \qquad
  \gamma=\gamma_{S,l}.
  \label{eq:sample_centered_davis_kahan_application_general}
\end{equation}
The gap condition \eqref{eq:selected_sc_gap_condition_general}, together with
\eqref{eq:sample_centered_operator_bound_proof_general}, implies
\begin{equation}
  \|\widehat A-A\|_{\mathrm{op}}
  \le
  \varepsilon_{l,N}^{\mathrm{sc}}(\delta)
  <
  \frac{\gamma_{S,l}}{2}.
  \label{eq:sample_centered_gap_verified_general}
\end{equation}
Therefore, by
\eqref{eq:davis_kahan_projector_bound_recalled_general},
\begin{equation}
  \left\|
    \widehat\Pi_{S,l}^{\mathrm{sc}}-\Pi_{S,l}
  \right\|_{\mathrm{op}}
  \le
  \frac{
    2\varepsilon_{l,N}^{\mathrm{sc}}(\delta)
  }{
    \gamma_{S,l}
  }.
  \label{eq:sample_centered_projector_bound_proof_general}
\end{equation}
This proves \eqref{eq:selected_sc_projector_general}.
\end{proof}

\begin{remark}[What the abstract theorem does and does not use]
\label{rem:abstract_selected_theorem_inputs}
The finite-sample theorem in this section uses only three inputs: independent
shadow shots, fixed coordinate selection, and a deterministic bound on the
selected one-shot radius.  It does not use Pauli compatibility, locality, or
an explicit formula for the covariance entries.  These additional
structures enter only when the abstract radius \(B_l\) is evaluated in a
concrete protocol.  In the local Pauli specialization below, bounded Pauli
weight and nondegenerate local basis probabilities give a dimension-independent
bound on \(B_l\), which is what turns the abstract theorem into a
constant-selected-error statement.
\end{remark}

\begin{proof}[Proof of Theorem~\ref{thm:constant_selected_error_bounded_radius}]

The bound \eqref{eq:epsilon_sc_constant_radius_bound_general} follows from
the explicit sample-centered error formula
\eqref{eq:epsilon_sc_general} by using the two assumptions in
\eqref{eq:bounded_selected_radius_assumption_general}.  Indeed,
\(l\le l_0\) implies
\begin{equation}
  \log(4l/\delta)\le \log(4l_0/\delta),
  \qquad
  \frac43+2l\le \frac43+2l_0,
  \label{eq:constant_radius_l_monotonicity_general}
\end{equation}
and \(A_l(\rho)\le A_0\) gives
\begin{equation}
  A_l(\rho)^2\le A_0^2.
  \label{eq:constant_radius_A_monotonicity_general}
\end{equation}
Substituting \eqref{eq:constant_radius_l_monotonicity_general} and
\eqref{eq:constant_radius_A_monotonicity_general} into
\eqref{eq:epsilon_sc_general} proves
\eqref{eq:epsilon_sc_constant_radius_bound_general}.

Now assume \eqref{eq:N_sufficient_constant_radius_general}.  
To show \eqref{eq:operator_constant_radius_general}, we employ item (ii) of
Theorem~\ref{thm:selected_shadow_sample_centered_spectral_recovery}.
The first lower
bound on \(N\) in \eqref{eq:N_sufficient_constant_radius_general} gives
\begin{equation}
  A_0^2
  \sqrt{\frac{8\log(4l_0/\delta)}{N}}
  \le
  \frac{\epsilon}{2}.
  \label{eq:constant_radius_first_term_epsilon_half_general}
\end{equation}
The second lower bound on \(N\) in
\eqref{eq:N_sufficient_constant_radius_general} gives
\begin{equation}
  A_0^2
  \left(\frac43+2l_0\right)
  \frac{\log(4l_0/\delta)}{N}
  \le
  \frac{\epsilon}{2}.
  \label{eq:constant_radius_second_term_epsilon_half_general}
\end{equation}
Combining \eqref{eq:epsilon_sc_constant_radius_bound_general},
\eqref{eq:constant_radius_first_term_epsilon_half_general}, and
\eqref{eq:constant_radius_second_term_epsilon_half_general}, we obtain
\begin{equation}
  \varepsilon_{l,N}^{\mathrm{sc}}(\delta)\le \epsilon.
  \label{eq:constant_radius_epsilon_sc_le_epsilon_general}
\end{equation}
The operator-norm estimate \eqref{eq:operator_constant_radius_general} then
follows from
\eqref{eq:selected_sc_operator_general}
of item (ii) of 
Theorem~\ref{thm:selected_shadow_sample_centered_spectral_recovery} and
\eqref{eq:constant_radius_epsilon_sc_le_epsilon_general}.
\end{proof}

\section{Local Pauli shadows}
\label{sec:local_pauli_exact_covariance}
In this section we specialize the general selected-covariance framework to
the \(k\)-qubit biased local Pauli shadow protocol.  The main purpose
of this section is first to evaluate the selected one-shot radius \(B_l\) for
a fixed selected Pauli coordinate set and thereby to instantiate the general
finite-sample selected covariance theorem in the local Pauli setting.

After this finite-sample specialization, we turn to an analytic calculation
of the covariance entries for the same protocol.  This latter
calculation uses structural features that are specific to qubit Pauli
measurements: each non-identity local Pauli observable is one of \(X,Y,Z\),
local measurement outcomes take values in \(\{\pm1\}\), and Pauli strings
compatible with the same local basis pattern have a simple cancellation
structure.  Thus the closed covariance formulas below should be understood
as additional local-Pauli structure, not as assumptions needed for the
abstract finite-sample theorem.

Biased local Pauli measurements have been studied previously as a
variance-reduction tool for Hamiltonian expectation estimation
\cite{HadfieldBravyiRaymondMezzacapo2022}. Here we use the same biased
basis-selection mechanism to analyze the full covariance structure of
selected reconstructed Pauli coordinates.

\subsection{Formulation of biased qubit local Pauli shadows}
\label{subsec:biased_qubit_local_pauli_formulation}

We now specialize the abstract shadow-output framework of
Section~\ref{sec:general_covariance} to the \(k\)-qubit local Pauli setting.
There are two distinct choices to be fixed.

First, we fix the physical quantities whose expectation values are represented
by the coordinates of the shadow-output vector.  In the general framework,
these quantities were specified by an observable family
\(
  \mathcal O=\{O_1,\ldots,O_p\}.
\)
In the present section, we take this observable family to be the set of all
non-identity \(k\)-qubit Pauli strings.  Let
$  \mathcal P_k:=
  \{I,X,Y,Z\}^{\otimes k}
$ be the \(k\)-qubit Pauli string set, and let
$  \mathcal P_k^\times
  :=
  \mathcal P_k\setminus\{I^{\otimes k}\}
$ be the set of non-identity Pauli strings.  Thus, in the notation of
Section~\ref{subsec:shadow_output_vectors_covariance_matrices}, we specialize
\begin{equation}
  \mathcal O
  =
  \{O_1,\ldots,O_p\}
  =
  \mathcal P_k^\times,
  \qquad
  p=4^k-1.
  \label{eq:local_pauli_observable_family_specialization}
\end{equation}
Each coordinate of the general shadow-output vector
\(\hat x\in\mathbb R^p\) in
\eqref{eq:general_shadow_output_vector} is therefore indexed by a
non-identity Pauli string \(P\in\mathcal P_k^\times\).  We write this
coordinate as
\begin{equation}
  \hat x_P
  :=
  \operatorname{Tr}(\hat\rho\,P),
  \qquad
  P\in\mathcal P_k^\times,
  \label{eq:local_pauli_shadow_output_coordinate_specialization}
\end{equation}
where \(\hat\rho\) is the single-shot reconstructed shadow snapshot.  The
corresponding coordinate is the Pauli expectation value
\begin{equation}
  m_P
  :=
  \operatorname{Tr}(\rho P),
  \qquad
  P\in\mathcal P_k^\times.
  \label{eq:local_pauli_mean_coordinate_specialization}
\end{equation}
At this stage, we have only fixed the target observables whose reconstructed
coefficients form the vector \(\hat x\); no physical covariance between
observables is being assumed or directly measured.

Second, we fix the measurement mechanism that generates these reconstructed
coefficients.  In the general framework, this was an arbitrary randomized
measurement protocol with a corresponding reconstruction map.  Here, we use
the biased local Pauli shadow protocol.  At each site
\(j=1,\ldots,k\), a local Pauli basis
\(
  \hat r_j\in\{X,Y,Z\}
\)
is chosen with probabilities
\begin{align}
&  \mathbb P(\hat r_j=r)
  =
  p_{j,r},
  \qquad
  r\in\{X,Y,Z\},\notag\\
 & p_{j,X},p_{j,Y},p_{j,Z}>0,
  \quad
  p_{j,X}+p_{j,Y}+p_{j,Z}=1.
  \label{eq:local_pauli_basis_probabilities}
\end{align}
We write
\begin{equation}
  \hat r
  =
  (\hat r_1,\ldots,\hat r_k)
  \in
  \{X,Y,Z\}^k
  \label{eq:local_pauli_basis_pattern}
\end{equation}
for the resulting local basis pattern.  Conditional on \(\hat r\), each site
is measured in the corresponding Pauli basis, and the outcome at site \(j\)
is
$  \hat s_j\in\{\pm1\}$.
Thus one shot of the measurement protocol produces
\begin{equation}
  (\hat r,\hat s)
  =
  \bigl((\hat r_1,\ldots,\hat r_k),(\hat s_1,\ldots,\hat s_k)\bigr).
  \label{eq:local_pauli_single_shot_data}
\end{equation}

The full $k$-qubit snapshot $\hat\rho$ is given as
\begin{align}
\hat\rho=\bigotimes_{j=1}^k 
\frac12\Bigl(I+\frac{\hat s_j}{p_{j,\hat r_j}}\,\hat r_j\Bigr),
\label{KJ13}
\end{align}
which is shown in Appendix \ref{KJ13P}.
For \(P\in\mathcal P_k^\times\), we write
\(\hat x_P=\Tr(\hat\rho P)\) and \(m_P=\Tr(\rho P)\),
identifying each element of \(\mathcal P_k^\times\) with its coordinate index.
Once these two components have been fixed, the covariance matrix considered
in this section is the covariance of the reconstructed Pauli coefficients
\((\hat x_P)_{P\in\mathcal P_k^\times}\).  Thus the entries of
\(\Sigma(\rho)\) are
\begin{equation}
  \Sigma_{PQ}(\rho)
  =
  \operatorname{Cov}_\rho(\hat x_P,\hat x_Q)
  =
  \mathbb E_\rho[\hat x_P\hat x_Q]-m_Pm_Q
  \label{eq:local_pauli_covariance_entry_specialization}
\end{equation}
for $ P,Q\in\mathcal P_k^\times$.
This is a covariance of the random reconstructed coefficients produced by
the shadow protocol.  It should not be confused with directly measuring a
physical product observable such as \(PQ\).

For a Pauli string \(P\in\mathcal P_k\), we define its support 
and its weight as
\begin{equation}
  \supp(P)
  :=
  \{j:\,P_j\neq I\}, \quad
  \wt(P):=|  \supp(P)|.
  \label{eq:local_pauli_support_def}
\end{equation}
We say that a local basis pattern
\(r=(r_1,\ldots,r_k)\in\{X,Y,Z\}^k\) is compatible with \(P\), and write
$  r\succeq P$,
if
$  r_j=P_j$ 
for every $ j\in\supp(P)$.
This compatibility condition connects the target observable \(P\) with the
measurement mechanism: the reconstructed coefficient \(\hat x_P\) can be
nonzero only when the chosen basis pattern is compatible with \(P\).

For $P,Q\in\mathcal P_k$, we say that $P$ and $Q$ are compatible, and write $P\sim Q$, if on every site where both are non-identity, they carry the same Pauli.
If $P\sim Q$, define $P\ominus Q\in\mathcal P_k$ by canceling the common non-identity factors and retaining the remaining ones.
For later use, define the overlap bias factor
\begin{align}
 \beta(P,Q):=\prod_{j\in\supp(P)\cap\supp(Q)} p_{j,P_j}^{-1},
 \qquad (P\sim Q),
\label{BSHJ3}
\end{align}
which is well defined because compatibility implies $P_j=Q_j$ on the overlap.
The coefficient rule below is the basic local Pauli input for both parts of
this section.  It is used first to compute the selected one-shot radius, and
later to derive the exact covariance formula.

\begin{lemma}[Local Pauli shadow-output coefficient rule]
\label{lem:biased_local_pauli_coefficient_rule}
For every \(P\in\mathcal P_k^\times\), the biased local Pauli
snapshot satisfies
\begin{equation}
  \hat x_P
  :=
  \operatorname{Tr}(\hat\rho P)
  =
  \left(
    \prod_{j\in\supp(P)}
    p_{j,P_j}^{-1}
  \right)
  \mathbf 1_{\{\hat r\succeq P\}}
  \prod_{j\in\supp(P)}
  \hat s_j .
  \label{eq:local_pauli_shadow_output_coefficient_rule}
\end{equation}
\end{lemma}

\begin{proof}
Using the product form of the biased local Pauli snapshot,
\(
  \hat\rho
  =
  \bigotimes_{j=1}^k
  \frac12
  \left(
    I+\frac{\hat s_j}{p_{j,\hat r_j}}\hat r_j
  \right),
\)
we compute, for \(P=P_1\otimes\cdots\otimes P_k\),
\begin{align}
  \hat x_P
  =
  \operatorname{Tr}(\hat\rho P)
  =
  \prod_{j=1}^k
  \operatorname{Tr}\!\left[
    \frac12
    \left(
      I+\frac{\hat s_j}{p_{j,\hat r_j}}\hat r_j
    \right)
    P_j
  \right].
  \label{eq:local_pauli_output_factorization_proof}
\end{align}
We evaluate the local factor.  If \(P_j=I\), then
\[
  \operatorname{Tr}\!\left[
    \frac12
    \left(
      I+\frac{\hat s_j}{p_{j,\hat r_j}}\hat r_j
    \right)
    I
  \right]
  =
  1,
\]
because \(\operatorname{Tr}(I)=2\) and
\(\operatorname{Tr}(\hat r_j)=0\).  If \(P_j\neq I\), then
\(P_j\in\{X,Y,Z\}\), and
\begin{align}
  \operatorname{Tr}\!\left[
    \frac12
    \left(
      I+\frac{\hat s_j}{p_{j,\hat r_j}}\hat r_j
    \right)
    P_j
  \right]
  =
  \frac{\hat s_j}{2p_{j,\hat r_j}}
  \operatorname{Tr}(\hat r_jP_j).
  \label{eq:local_pauli_nonidentity_local_factor}
\end{align}
By orthogonality of the single-qubit Pauli matrices,
\[
  \operatorname{Tr}(\hat r_jP_j)
  =
  \begin{cases}
  2, & \hbox{when }\hat r_j=P_j,\\
  0, & \hbox{when }\hat r_j\neq P_j.
  \end{cases}
\]
Therefore, for \(P_j\neq I\),
\[
  \operatorname{Tr}\!\left[
    \frac12
    \left(
      I+\frac{\hat s_j}{p_{j,\hat r_j}}\hat r_j
    \right)
    P_j
  \right]
  =
  p_{j,P_j}^{-1}
  \mathbf 1_{\{\hat r_j=P_j\}}
  \hat s_j .
\]
Multiplying the local factors over all sites, the sites with \(P_j=I\)
contribute \(1\), while the sites in \(\supp(P)\) contribute the factors
above.  Hence, we have
\(
  \hat x_P
  =
  \prod_{j\in\supp(P)}
  \left(
    p_{j,P_j}^{-1}
    {\mathbf 1}_{\{\hat r_j=P_j\}}
    \hat s_j
  \right).
\)
Finally, it follows that
\[
  \prod_{j\in\supp(P)}
  \mathbf 1_{\{\hat r_j=P_j\}}
  =
  \mathbf 1_{\{\hat r\succeq P\}}
\]
from the definition of compatibility of the basis pattern \(\hat r\) with
\(P\).  This proves
\eqref{eq:local_pauli_shadow_output_coefficient_rule}.
\end{proof}

This explicit coefficient rule is the point at which the present subsection
becomes specific to qubit local Pauli shadows.  The exact covariance formula
below is obtained by analyzing products \(\hat x_P\hat x_Q\) using
Lemma~\ref{lem:biased_local_pauli_coefficient_rule}.

\subsection{Selected radius and finite-sample consequences}
\label{subsec:local_pauli_selected_radius_finite_sample_consequences}

We now connect the preceding local Pauli formulation with the general
finite-sample selected covariance theorem of
Section~\ref{sec:general_selected_covariance_spectral_recovery}.  That
theorem applies to arbitrary shadow protocols once the protocol-dependent
selected one-shot radius is identified.  Thus, in the biased local
Pauli setting, the remaining task is to compute the selected radius \(B_l\)
appearing in Definition~\ref{def:selected_one_shot_radius_general} and then
substitute it into the sample-centered selected covariance theorem.

Let
\(
  J_l=\{P_1,\ldots,P_l\}
  \subset \mathcal P_k^\times
\)
be a fixed selected set of distinct non-identity Pauli strings, and let
\(R_l\) be the corresponding coordinate-selection matrix.  Define
\begin{equation}
  B_{l,\mathrm{LP}}^2
  :=
  \max_{r\in\{X,Y,Z\}^k}
  \sum_{a=1}^l
  \left(
    \prod_{j\in\supp(P_a)}
    p_{j,P_{a,j}}^{-2}
  \right)
  \mathbf 1_{\{r\succeq P_a\}}.
  \label{eq:local_pauli_Bl_LP_def}
\end{equation}
This is the deterministic selected one-shot radius predicted by the biased local Pauli coefficient rule.

\begin{lemma}[Selected one-shot radius for biased local Pauli shadows]
\label{lem:local_pauli_Bl_equals_BlLP}
For the biased local Pauli shadow protocol,
\begin{equation}
  \|R_l\hat x\|_2^2
  =
  \sum_{a=1}^l
  \left(
    \prod_{j\in\supp(P_a)}
    p_{j,P_{a,j}}^{-2}
  \right)
  \mathbf 1_{\{\hat r\succeq P_a\}}.
  \label{eq:local_pauli_selected_norm_identity}
\end{equation}
Consequently, the selected one-shot radius \(B_l\) in
Definition~\ref{def:selected_one_shot_radius_general} is
\begin{equation}
  B_l=B_{l,\mathrm{LP}}.
  \label{eq:local_pauli_Bl_equals_BlLP}
\end{equation}
\end{lemma}

\begin{proof}
For each selected Pauli string \(P_a\), the biased local Pauli
coefficient rule given in Lemma \ref{lem:biased_local_pauli_coefficient_rule}
gives
\begin{equation}
  \hat x_{P_a}
  =
  \left(
    \prod_{j\in\supp(P_a)}
    p_{j,P_{a,j}}^{-1}
  \right)
  \mathbf 1_{\{\hat r\succeq P_a\}}
  \prod_{j\in\supp(P_a)}\hat s_j .
  \label{eq:local_pauli_radius_coefficient_rule}
\end{equation}
Since \((\hat s_j)^2=1\) and
\(\mathbf 1_{\{\hat r\succeq P_a\}}^2
=\mathbf 1_{\{\hat r\succeq P_a\}}\), we obtain
\begin{equation}
  \hat x_{P_a}^2
  =
  \left(
    \prod_{j\in\supp(P_a)}
    p_{j,P_{a,j}}^{-2}
  \right)
  \mathbf 1_{\{\hat r\succeq P_a\}}.
  \label{eq:local_pauli_radius_coefficient_square}
\end{equation}
Summing \eqref{eq:local_pauli_radius_coefficient_square} over
\(a=1,\ldots,l\) gives
\eqref{eq:local_pauli_selected_norm_identity}.  The right-hand side depends
only on the basis pattern \(\hat r\), not on the signs \(\hat s_j\).  Since all
basis probabilities are strictly positive, the essential supremum of
\(\|R_l\hat x\|_2\) over one-shot outcomes is the maximum over
\(r\in\{X,Y,Z\}^k\).  By the definition
\eqref{eq:local_pauli_Bl_LP_def}, this proves
\eqref{eq:local_pauli_Bl_equals_BlLP}.
\end{proof}

Define the following quantity:
\begin{equation}
\widetilde A_{l,\mathrm{LP}}(\rho)
  :=
  B_{l,\mathrm{LP}}+\|R_lm\|_2 .
  \label{eq:local_pauli_Al_LP_def}
\end{equation}
By Lemma~\ref{lem:local_pauli_Bl_equals_BlLP}, this is exactly 
the upper bound \(\widetilde A_l(\rho)\) of the radius
\(A_l(\rho)\)
defined in Definition \ref{def:selected_one_shot_radius_general}
for the present biased local Pauli protocol.

\begin{corollary}[Sample-centered selected spectral approximation for biased local Pauli shadows]
\label{cor:local_pauli_sample_centered_consequence}
Let \(J_l=\{P_1,\ldots,P_l\}\subset\mathcal P_k^\times\) be fixed
independently of the measurement data, and let \(R_l\) be the corresponding
selection matrix.  Define \(B_{l,\mathrm{LP}}\) by
\eqref{eq:local_pauli_Bl_LP_def} and \(\widetilde A_{l,\mathrm{LP}}(\rho)\) by
\eqref{eq:local_pauli_Al_LP_def}.  Then the conclusions of
Theorem~\ref{thm:selected_shadow_sample_centered_spectral_recovery} hold for
the biased local Pauli protocol with
\(
\widetilde A_l(\rho)=\widetilde A_{l,\mathrm{LP}}(\rho).
\)
In particular, with probability at least \(1-\delta\),
\begin{align}
&  \left\|
    \widehat\Sigma_{l,N}^{\mathrm{sc}}-\Sigma_l(\rho)
  \right\|_{\mathrm{op}} \notag\\
  \le &
 \widetilde A_{l,\mathrm{LP}}(\rho)^2
  \left(
    \sqrt{\frac{8\log(4l/\delta)}{N}}
    +
    \left(\frac43+2l\right)
    \frac{\log(4l/\delta)}{N}
  \right).
  \label{eq:local_pauli_sample_centered_operator_bound}
\end{align}
The corresponding Loewner, eigenvalue, and spectral-projector estimates are
the ones stated in
Theorem~\ref{thm:selected_shadow_sample_centered_spectral_recovery}.
\end{corollary}

\begin{proof}
By Lemma~\ref{lem:local_pauli_Bl_equals_BlLP}, the selected one-shot radius
in Definition~\ref{def:selected_one_shot_radius_general} is
\(B_l=B_{l,\mathrm{LP}}\).  Hence the centered radius in the general theorem is
\[A_l(\rho)\le 
\widetilde A_l(\rho)
  =
  B_l+\|R_lm\|_2
  =
  B_{l,\mathrm{LP}}+\|R_lm\|_2
  =
\widetilde A_{l,\mathrm{LP}}(\rho).
\]
Thus the assumptions of
Theorem~\ref{thm:selected_shadow_sample_centered_spectral_recovery} are
satisfied.  Substituting \(\widetilde A_{l,\mathrm{LP}}(\rho)\) into the explicit
sample-centered error formula \eqref{eq:epsilon_sc_general} gives
\eqref{eq:local_pauli_sample_centered_operator_bound}.  The remaining claims
are exactly the corresponding conclusions of
Theorem~\ref{thm:selected_shadow_sample_centered_spectral_recovery}.
\end{proof}

\begin{corollary}[Constant selected error for bounded-weight local Pauli coordinates]
\label{cor:local_pauli_constant_selected_error}
Assume that there exists a constant \(p_0>0\), independent of \(k\), such
that
\begin{equation}
  p_{j,r}\ge p_0,
  \qquad
  j=1,\ldots,k,\quad r\in\{X,Y,Z\}.
  \label{eq:local_pauli_probability_lower_bound_constant_error}
\end{equation}
Assume also that the selected Pauli coordinates satisfy
\begin{equation}
  l\le l_0,
  \qquad
  \max_{1\le a\le l}\wt(P_a)\le w_0,
  \label{eq:local_pauli_selected_size_weight_bound_constant_error}
\end{equation}
where \(l_0\) and \(w_0\) are constants independent of \(k\).  Set
\begin{equation}
  A_0
  :=
  \sqrt{l_0}\,\bigl(p_0^{-w_0}+1\bigr).
  \label{eq:local_pauli_A0_constant_error_def}
\end{equation}
Then
\begin{equation}
\widetilde A_{l,\mathrm{LP}}(\rho)\le A_0.
  \label{eq:local_pauli_Al_LP_constant_bound}
\end{equation}
Consequently, the constant-radius conclusion of
Theorem~\ref{thm:constant_selected_error_bounded_radius} applies.  In
particular, with probability at least \(1-\delta\),
\begin{align}
&  \left\|
    \widehat\Sigma_{l,N}^{\mathrm{sc}}-\Sigma_l(\rho)
  \right\|_{\mathrm{op}}\notag\\
  \le &
  A_0^2
  \left(
    \sqrt{\frac{8\log(4l_0/\delta)}{N}}
    +
    \left(\frac43+2l_0\right)
    \frac{\log(4l_0/\delta)}{N}
  \right).
  \label{eq:local_pauli_constant_error_bound}
\end{align}
Moreover, for any target accuracy \(\epsilon>0\), it is enough to take
\begin{equation}
  N
  \ge
  \max\left\{
    \frac{32A_0^4\log(4l_0/\delta)}{\epsilon^2},
    \frac{2A_0^2(\frac43+2l_0)\log(4l_0/\delta)}{\epsilon}
  \right\}
  \label{eq:local_pauli_N_sufficient_constant_error}
\end{equation}
to ensure that, with probability at least \(1-\delta\),
\begin{equation}
  \left\|
    \widehat\Sigma_{l,N}^{\mathrm{sc}}-\Sigma_l(\rho)
  \right\|_{\mathrm{op}}
  \le
  \epsilon .
  \label{eq:local_pauli_operator_constant_error}
\end{equation}
Thus, under bounded selected set size, bounded selected Pauli weight, and
basis probabilities bounded below, constant operator-norm accuracy for the
selected sample-centered covariance requires a number of samples independent
of the total number of qubits.
\end{corollary}

\begin{proof}
By the definition of \(B_{l,\mathrm{LP}}\),
\begin{align}
  B_{l,\mathrm{LP}}^2
  &=
  \max_{r\in\{X,Y,Z\}^k}
  \sum_{a=1}^l
  \left(
    \prod_{j\in\supp(P_a)}
    p_{j,P_{a,j}}^{-2}
  \right)
  \mathbf 1_{\{r\succeq P_a\}}
  \notag\\
  &\le
  \sum_{a=1}^l
  \prod_{j\in\supp(P_a)}
  p_{j,P_{a,j}}^{-2}.
  \label{eq:local_pauli_Bl_LP_first_bound}
\end{align}
Using the lower bound
\eqref{eq:local_pauli_probability_lower_bound_constant_error}, we have
\begin{equation}
  \prod_{j\in\supp(P_a)}
  p_{j,P_{a,j}}^{-2}
  \le
  p_0^{-2\wt(P_a)}
  \le
  p_0^{-2w_0}.
  \label{eq:local_pauli_inverse_probability_weight_bound}
\end{equation}
Together with \(l\le l_0\), this gives
\begin{equation}
  B_{l,\mathrm{LP}}^2
  \le
  l_0p_0^{-2w_0},
  \qquad
  B_{l,\mathrm{LP}}
  \le
  \sqrt{l_0}\,p_0^{-w_0}.
  \label{eq:local_pauli_Bl_LP_constant_bound}
\end{equation}
On the other hand, 
each Pauli string \(P_a\) satisfies
$  |m_{P_a}|=|\operatorname{Tr}(\rho P_a)|\le 1$,
because \(P_a\) has eigenvalues in \(\{\pm1\}\).  Hence
\begin{equation}
  \|R_lm\|_2^2
  =
  \sum_{a=1}^l m_{P_a}^2
  \le
  l
  \le
  l_0,
  \qquad
  \|R_lm\|_2\le\sqrt{l_0}.
  \label{eq:local_pauli_Rlm_constant_bound}
\end{equation}
Combining
\eqref{eq:local_pauli_Al_LP_def},
\eqref{eq:local_pauli_Bl_LP_constant_bound}, and
\eqref{eq:local_pauli_Rlm_constant_bound}, we obtain
\[
  \widetilde  A_{l,\mathrm{LP}}(\rho)
  =
  B_{l,\mathrm{LP}}+\|R_lm\|_2
  \le
  \sqrt{l_0}\,\bigl(p_0^{-w_0}+1\bigr)
  =
  A_0,
\]
which proves \eqref{eq:local_pauli_Al_LP_constant_bound}.  The bound
\eqref{eq:local_pauli_constant_error_bound} and the sufficient sample size
condition \eqref{eq:local_pauli_N_sufficient_constant_error} then follow
directly from
Theorem~\ref{thm:constant_selected_error_bounded_radius}, with
\(\widetilde A_l(\rho)=\widetilde A_{l,\mathrm{LP}}(\rho)\le A_0\).
\end{proof}

This selected-radius calculation also identifies the structural features of
the local Pauli protocol that make finite-sample selected covariance
estimation effective.  First, locality ensures that a bounded-weight Pauli
coordinate depends only on a bounded number of local measurement choices.
Second, the bounded-weight assumption prevents the inverse-probability
factors in the reconstructed coefficients from growing with the total number
of qubits.  Third, local basis compatibility determines which selected Pauli
coordinates can be simultaneously active under a single basis pattern, and
therefore controls the size of the selected one-shot vector.  Finally, the
inverse-probability structure shows explicitly how small local basis
probabilities amplify both the covariance entries and the selected radius.
Thus the dimension-independent selected covariance guarantee is not a
black-box consequence of matrix concentration alone.  It follows because the
local Pauli measurement structure turns the abstract selected-radius
condition into a concrete bounded quantity for bounded-size, bounded-weight
selected coordinate sets.  This point will be useful in the comparison below:
the abstract finite-sample theorem applies to any shadow protocol once the
selected radius is controlled, but the availability of such control is a
protocol-dependent structural feature.

\subsection{Comparison with global Clifford shadows}
\label{subsec:comparison_with_global_clifford_shadows}

The preceding local Pauli result illustrates how locality and bounded
observable weight can keep the selected one-shot radius independent of the
number of qubits.  We now contrast this behavior with global Clifford
shadows.  The purpose of this comparison is not to derive a full covariance
formula for global Clifford shadows, 
but only to show how
the selected-radius quantity
\(
  B_l
  =
  \operatorname*{ess\,sup}\|R_l\hat x\|_2
\)
behaves when the measurement mechanism is changed from locally chosen Pauli
bases to a global Clifford measurement.

The comparison is made within the same abstract framework as before.  The
target observables are unchanged: we still take the output coordinates to be
the non-identity Pauli strings
\[
  \mathcal O=\mathcal P_k^\times,
  \qquad
  p=|\mathcal P_k^\times|=d^2-1,
  \qquad
  d=2^k.
\]
Only the measurement mechanism is changed.  In the global Clifford protocol,
a Clifford unitary \(C\) is drawn, the state is measured in the computational
basis after applying \(C\), and a measurement outcome \(b\in\{0,1\}^k\) gives
the rank-one stabilizer snapshot
\begin{equation}
  \hat\sigma_{C,b}
  :=
  C^\dagger |b\rangle\!\langle b| C .
  \label{eq:global_clifford_raw_snapshot}
\end{equation}
The reconstructed shadow snapshot is
\begin{equation}
  \hat\rho_{C,b}
  :=
  \mathcal M_{\mathrm{Cl}}^{-1}(\hat\sigma_{C,b}),
  \label{eq:global_clifford_reconstructed_snapshot}
\end{equation}
and the Pauli output coordinate is
\begin{equation}
  \hat x_P
  :=
  \operatorname{Tr}(\hat\rho_{C,b}P),
  \qquad
  P\in\mathcal P_k^\times .
  \label{eq:global_clifford_pauli_output_coordinate}
\end{equation}

The key difference from the local Pauli protocol is the action of the shadow
channel on the non-identity Pauli sector.  For the global Clifford ensemble,
the shadow channel acts as a scalar on the
traceless Pauli sector.  Equivalently, using the standard Clifford-shadow
inverse map, or the Pauli-invariant reconstruction formula of
Bu--Koh--Garcia--Jaffe~\cite[Theorem~1]{BuKohGarciaJaffe2024}, one has
$  \mathcal M_{\mathrm{Cl}}[P]
  =
  \frac{1}{d+1}P$
for
$ P\in\mathcal P_k^\times$,
and hence
\begin{equation}
  \mathcal M_{\mathrm{Cl}}^{-1}[P]
  =
  (d+1)P.
  \label{eq:global_clifford_inverse_scalar_action}
\end{equation}
Therefore, for a single global Clifford snapshot \(\hat\sigma=\hat\sigma_{C,b}\),
\begin{equation}
  \hat x_P
  =
  \operatorname{Tr}\!\left(
    \mathcal M_{\mathrm{Cl}}^{-1}(\hat\sigma)P
  \right)
  =
  (d+1)\operatorname{Tr}(\hat\sigma P).
  \label{eq:global_clifford_coefficient_rule}
\end{equation}

For a stabilizer snapshot \(\hat\sigma\), define its non-identity stabilizer
support by
\begin{equation}
  S(\hat\sigma)
  :=
  \{P\in\mathcal P_k^\times:
    \operatorname{Tr}(\hat\sigma P)\neq 0
  \}.
  \label{eq:global_clifford_stabilizer_support_def}
\end{equation}
For a rank-one stabilizer state, the Pauli expectation
\(\operatorname{Tr}(\hat\sigma P)\) is equal to \(\pm1\) for
\(P\in S(\hat\sigma)\), and is zero otherwise.  Thus, we have
\begin{equation}
  \hat x_P^2
  =
  (d+1)^2
  \mathbf 1_{\{P\in S(\hat\sigma)\}}.
  \label{eq:global_clifford_coefficient_square}
\end{equation}

Let
$  J_l=\{P_1,\ldots,P_l\}\subset\mathcal P_k^\times$
be a fixed selected set of distinct non-identity Pauli strings, and let
\(R_l\) be the corresponding selection matrix.  Define
\begin{equation}
  \kappa_{\mathrm{Cl}}(J_l)
  :=
  \max_{\hat\sigma}
  |J_l\cap S(\hat\sigma)|,
  \label{eq:global_clifford_kappa_selected_def}
\end{equation}
where the maximum is over rank-one stabilizer snapshots that can occur in the
global Clifford protocol.

\begin{proposition}[Selected one-shot radius for global Clifford shadows]
\label{prop:global_clifford_selected_radius}
For the global Clifford shadow protocol,
\begin{equation}
  \|R_l\hat x\|_2^2
  =
  (d+1)^2
  |J_l\cap S(\hat\sigma)|.
  \label{eq:global_clifford_selected_norm_identity}
\end{equation}
Consequently, the following state-uniform selected one-shot radius is valid:
$  B_{l,\mathrm{Cl}}^{\mathrm{unif}}
  :=
  (d+1)\sqrt{\kappa_{\mathrm{Cl}}(J_l)}$.
That is, for every state \(\rho\), the selected radius \(B_l(\rho)\) in
Definition~\ref{def:selected_one_shot_radius_general} satisfies
\(
  B_l(\rho)\le B_{l,\mathrm{Cl}}^{\mathrm{unif}}.
\)
\end{proposition}

\begin{proof}
By \eqref{eq:global_clifford_coefficient_square}, 
each selected Pauli
string \(P_a\in J_l\) satisfies
$  \hat x_{P_a}^2
  =
  (d+1)^2\mathbf 1_{\{P_a\in S(\hat\sigma)\}}$.
Summing this identity over \(a=1,\ldots,l\) gives
\begin{align*}
  \|R_l\hat x\|_2^2
  =&
  \sum_{a=1}^l \hat x_{P_a}^2
  =
  (d+1)^2
  \sum_{a=1}^l
  \mathbf 1_{\{P_a\in S(\hat\sigma)\}}\notag\\
  =&
  (d+1)^2|J_l\cap S(\hat\sigma)|,
\end{align*}
which proves \eqref{eq:global_clifford_selected_norm_identity}.  
Taking the supremum over all stabilizer snapshots that can occur in the
global Clifford measurement model gives the state-uniform bound
\begin{align*}
  B_l(\rho)\le &
(d+1)\sqrt{
  \max_{\hat\sigma}|J_l\cap S(\hat\sigma)|}
=
(d+1)\sqrt{\kappa_{\mathrm{Cl}}(J_l)} \\
=  &
B_{l,\mathrm{Cl}}^{\mathrm{unif}}.
\end{align*}
\end{proof}

\begin{remark}[Selected-radius obstruction for global Clifford shadows]
\label{rem:global_clifford_selected_radius_obstruction}
The contrast with the biased local Pauli protocol is sharp at the
level of the selected one-shot radius.  In the local Pauli case, if
\(l\le l_0\), the selected Pauli strings have weight at most \(w_0\), and
the local basis probabilities satisfy \(p_{j,r}\ge p_0>0\), then
\(
  B_{l,\mathrm{LP}}
  \le
  \sqrt{l_0}\,p_0^{-w_0},
\)
which is independent of the total number of qubits.  This is the radius
input behind the constant selected-error corollary above.  

By contrast, Proposition~\ref{prop:global_clifford_selected_radius} shows
that the state-uniform selected radius for global Clifford shadows scales as
\((d+1)\sqrt{\kappa_{\mathrm{Cl}}(J_l)}\).  For any nonempty selected Pauli set,
\(\kappa_{\mathrm{Cl}}(J_l)\ge1\), and hence this radius is at least of order
\(d\).
Thus the bounded-radius mechanism used in the selected finite-sample theorem
does not yield a dimension-independent constant selected-error regime for
global Clifford shadows, even for a singleton selected coordinate.  This does
not assert an information-theoretic impossibility result for all possible
estimators; rather, it shows that the selected-radius route used in this
paper is intrinsically favorable to local Pauli measurements with bounded
selected Pauli weight, and not to the global Clifford protocol.
\end{remark}

\subsection{Exact covariance formula for qubit local Pauli shadows}
\label{subsec:exact_covariance_formula_qubit_local_pauli}

We now record an exact covariance formula for the same biased local Pauli protocol.  This calculation is not needed for the
abstract finite-sample theorem or for the selected-radius bound above.  Its role is instead structural: it identifies how local Pauli compatibility
and local basis-selection probabilities determine the raw second moments of
the reconstructed Pauli coefficients and, after subtracting the products of
the corresponding means, the entries of the covariance
matrix.

For the uniform local Pauli protocol, the underlying second-moment mechanism
already appears in the observable-wise variance analysis of
Huang--Kueng--Preskill~\cite[Lemma~4]{HuangKuengPreskill2020}.  We first
recall the corresponding full covariance-matrix form in the uniform case,
and then extend it to biased basis probabilities.

\begin{proposition}[Full covariance-matrix form of the local Pauli second moments]
\label{prop:local_pauli_exact_covariance}
Assume that $p_{j,X}=p_{j,Y}=p_{j,Z}=1/3$.
Let $P,Q\in\mathcal P_k^\times$.
Then
\[
\mathbb E[\hat x_P\hat x_Q]
=
\begin{cases}
0, & \hbox{when }P
\not\sim Q,\\[1mm]
3^{|\supp(P)\cap\supp(Q)|}\,m_{P\ominus Q}, & \hbox{when } P\sim Q.
\end{cases}
\]
Consequently,
\begin{align*}
&\Sigma_{PQ}(\rho)=
\Cov(\hat x_P,\hat x_Q)\notag\\
=&
\begin{cases}
-m_Pm_Q, & \hbox{when } P
\not\sim Q,\\[1mm]
3^{|\supp(P)\cap\supp(Q)|}\,m_{P\ominus Q}-m_Pm_Q, 
& \hbox{when }P\sim Q.
\end{cases}
\end{align*}
In particular, we have
\(
\Sigma_{PP}(\rho)=3^{\wt(P)}-m_P^2
\)
\end{proposition}

The purpose of
this section is to extend this result to 
the biased local Pauli shadows.  
Under our setting of the biased local Pauli shadows
given in \eqref{eq:local_pauli_basis_probabilities},
the above proposition is generalized as follows.
\begin{theorem}[Exact covariance-matrix formula for biased local Pauli shadows]
\label{thm:biased_local_pauli_exact_covariance}
Let $P,Q\in\mathcal P_k^\times$. Then
\begin{align}
\mathbb E[\hat x_P\hat x_Q]
=
\begin{cases}
0, &\hbox{when } P \not\sim Q,\\[1mm]
\beta(P,Q)\,m_{P\ominus Q}, &\hbox{when } P\sim Q.
\end{cases}\label{FKS1}
\end{align}
Consequently,
\begin{align}
&\Sigma_{PQ}(\rho)
=\Cov(\hat x_P,\hat x_Q)\notag\\
=&
\begin{cases}
-m_Pm_Q, & \hbox{when } P \not\sim Q,\\[1mm]
\beta(P,Q)\,m_{P\ominus Q}-m_Pm_Q, & \hbox{when }P\sim Q.
\end{cases}\label{FKS2}
\end{align}
In particular,
\[
\Sigma_{PP}(\rho)=\Bigl(\prod_{j\in\supp(P)} p_{j,P_j}^{-1}\Bigr)-m_P^2.
\]
\end{theorem}

The proof of Theorem~\ref{thm:biased_local_pauli_exact_covariance} proceeds by inserting the coefficient rule for $\hat x_P$ and $\hat x_Q$, separating the incompatible case $P\not\sim Q$ from the compatible case $P\sim Q$, identifying the unique common basis event in the latter case, and then subtracting $m_Pm_Q$ to pass from the raw second moment to the covariance.

\begin{proof}[Proof Sketch of Theorem~\ref{thm:biased_local_pauli_exact_covariance}]
Although the full proof is deferred to Appendix~\ref{app:proof_biased_local_pauli_exact_covariance},
we present its sketch, here.
Lemma~\ref{lem:biased_local_pauli_coefficient_rule} gives
\begin{align}
 &\hat x_P\hat x_Q\notag\\
 =&\Bigl(\prod_{j\in\supp(P)} p_{j,P_j}^{-1}\Bigr)
  \Bigl(\prod_{j\in\supp(Q)} p_{j,Q_j}^{-1}\Bigr)
 \mathbf 1_{\{\hat r_j=P_j\ \forall j\in\supp(P)\}} \notag\\
&\cdot \mathbf 1_{\{\hat r_j=Q_j\ \forall j\in\supp(Q)\}}
 \prod_{j\in\supp(P)} \hat s_j
 \prod_{j\in\supp(Q)} \hat s_j.
\end{align}
If $P\not\sim Q$, the two indicator events are incompatible at a site where both strings are non-identity but different, so $\hat x_P\hat x_Q=0$ for every realization and hence $\mathbb E[\hat x_P\hat x_Q]=0$.

Assume now that $P\sim Q$. There is again a unique common basis pattern on $\supp(P)\cup\supp(Q)$ for which both coefficients are nonzero. The overlap contributes exactly the factor
\[
 \beta(P,Q)=\prod_{j\in\supp(P)\cap\supp(Q)} p_{j,P_j}^{-1},
\]
while the nonoverlap basis choices are absorbed by the probability of the compatible event. On that event the outcome product reduces to the expectation of the cancellation string $P\ominus Q$, and therefore
\(
 \mathbb E[\hat x_P\hat x_Q]=\beta(P,Q)m_{P\ominus Q}.
\)
Subtracting $m_Pm_Q$ gives the covariance formula, and the diagonal case $P=Q$ yields
\[
 \Sigma_{PP}(\rho)=\Bigl(\prod_{j\in\supp(P)} p_{j,P_j}^{-1}\Bigr)-m_P^2.
\]
The detailed compatible-basis and sign-product calculations are recorded in Appendix~\ref{app:proof_biased_local_pauli_exact_covariance}.
\end{proof}

\begin{remark}[Uniform local Pauli as a special case]
\label{rem:biased_local_pauli_uniform_special_case}
If $p_{j,X}=p_{j,Y}=p_{j,Z}=1/3$ for every $j$, then
\[
\beta(P,Q)=3^{|\supp(P)\cap\supp(Q)|},
\quad
\prod_{j\in\supp(P)} p_{j,P_j}^{-1}=3^{\wt(P)},
\]
and Theorem~\ref{thm:biased_local_pauli_exact_covariance} reduces exactly to Proposition~\ref{prop:local_pauli_exact_covariance}.
\end{remark}

\section{General Local Product Shadows}
\label{sec:general_local_product_shadows}

The local Pauli analysis above uses special algebraic features of qubit Pauli
measurements.  However, the bounded-radius mechanism itself is not restricted
to Pauli measurements.  In this section we record a more general local
product-shadow setting in which the selected observables have finite support.
The main point is that, when both the measurement and the reconstruction are
local product constructions, a finite-weight observable only sees the local
snapshots on its support.  Consequently, the selected one-shot radius is
controlled by the selected support sizes and local reconstruction
coefficients, not by the total number \(k\) of tensor factors.

\subsection{Local product measurement model and finite-weight observables}
\label{subsec:general_local_product_model}

We set the physical system as
\begin{equation}
  \mathcal H
  =
  \bigotimes_{j=1}^k \mathcal H_j,
  \qquad
  \mathcal H_j\simeq \mathbb C^t ,
  \label{eq:general_local_product_hilbert_space}
\end{equation}
where the local dimension \(t\) is fixed.  At each site \(j\), let
\(\mathcal U_j\) be a finite set of local measurement settings.  A random
local setting
\(
  \hat u_j\in\mathcal U_j
\)
is drawn with probability
\(
  \mathbb P(\hat u_j=u_j)=q_{j,u_j}.
\)
We assume that the local settings are chosen independently, so that 
we have 
$  q(u)=  \prod_{j=1}^k q_{j,u_j}$ for \(u=(u_1,\ldots,u_k)\).

For each \(u_j\in\mathcal U_j\), let
$  M^{(j)}_{u_j}
  =
  \{M^{(j)}_{b_j|u_j}\}_{b_j\in\mathcal B_{j,u_j}}
$ be a POVM on \(\mathcal H_j\).  Conditional on
\(u=(u_1,\ldots,u_k)\), the product POVM is defined as
\begin{equation}
  M_{b|u}
  :=
  \bigotimes_{j=1}^k M^{(j)}_{b_j|u_j},
  \qquad
  b=(b_1,\ldots,b_k).
  \label{eq:general_local_product_povm_element}
\end{equation}
Thus one shot produces random data
\(
  (\hat u,\hat b)
  =
  ((\hat u_1,\ldots,\hat u_k),(\hat b_1,\ldots,\hat b_k)).
\)

For each site \(j\), define the local shadow channel by
\begin{equation}
  \mathcal M_j[X_j]
  :=
  \sum_{u_j\in\mathcal U_j}
  q_{j,u_j}
  \sum_{b_j\in\mathcal B_{j,u_j}}
  M^{(j)}_{b_j|u_j}
  \operatorname{Tr}
  \left(
    X_j M^{(j)}_{b_j|u_j}
  \right).
  \label{eq:general_local_product_local_shadow_channel}
\end{equation}
We assume that each \(\mathcal M_j\) is invertible on a prescribed local
operator subspace \(\mathcal V_j\).  We also assume that all local POVM
elements \(M^{(j)}_{b_j|u_j}\) belong to \(\mathcal V_j\), so that
\(\mathcal M_j^{-1}(M^{(j)}_{b_j|u_j})\) is well defined.

The global shadow channel is defined, as in
Section~\ref{subsec:general_finite_dimensional_shadow_channel} has the form
\begin{equation}
  \mathcal M[X]
  =
  \sum_{u\in\mathcal U}
  q(u)
  \sum_{b\in\mathcal B_u}
  M_{b|u}
  \operatorname{Tr}
  \left(
    X M_{b|u}
  \right).
  \label{eq:general_local_product_global_shadow_channel}
\end{equation}
In the present local product setting, \(q(u)=\prod_j q_{j,u_j}\) and
\(M_{b|u}=\bigotimes_j M^{(j)}_{b_j|u_j}\).  
Throughout this section, we assume also that each local shadow channel \(\mathcal M_j\) is invertible on \(\mathcal L(\mathcal H_j)\). 
The next lemma shows that these
two product structures imply the product form of the global reconstruction.

\begin{lemma}[Product form of the reconstructed snapshot]
\label{lem:general_local_product_snapshot_factorization}
Assume the local product measurement setting described above.  
Then, the global shadow channel factorizes as
\begin{equation}
  \mathcal M
  =
  \bigotimes_{j=1}^k \mathcal M_j .
  \label{eq:general_local_product_channel_factorization}
\end{equation}
Consequently, on this reconstruction subspace,
\begin{equation}
  \mathcal M^{-1}
  =
  \bigotimes_{j=1}^k \mathcal M_j^{-1}.
  \label{eq:general_local_product_inverse_factorization}
\end{equation}
For the realized outcome \((\hat u,\hat b)\), define
\begin{equation}
  \hat\rho_j
  :=
  \mathcal M_j^{-1}
  \left(
    M^{(j)}_{\hat b_j|\hat u_j}
  \right).
  \label{eq:general_local_product_local_snapshot}
\end{equation}
Then the global reconstructed snapshot
\(
  \hat\rho
  :=
  \mathcal M^{-1}(M_{\hat b|\hat u})
\)
factorizes as
\begin{equation}
  \hat\rho
  =
  \bigotimes_{j=1}^k \hat\rho_j .
  \label{eq:general_local_product_global_snapshot}
\end{equation}
\end{lemma}

\begin{proof}
It is enough to verify the factorization on product operators
\(X=X_1\otimes\cdots\otimes X_k\), since such operators span
\(\mathcal L(\mathcal H)\).  By the definition of the global shadow channel,
the product form of \(q(u)\), and the product form of the POVM elements,
\begin{align}
  \mathcal M[X]
  &=
  \sum_{u_1,\ldots,u_k}
  \left(
    \prod_{j=1}^k q_{j,u_j}
  \right)
  \sum_{b_1,\ldots,b_k}
  \left(
    \bigotimes_{j=1}^k M^{(j)}_{b_j|u_j}
  \right)
  \notag\\
  &\quad\cdot
  \operatorname{Tr}
  \left[
    \left(
      \bigotimes_{j=1}^k X_j
    \right)
    \left(
      \bigotimes_{j=1}^k M^{(j)}_{b_j|u_j}
    \right)
  \right].
\end{align}
The trace factorizes:
\[
  \operatorname{Tr}
  \left[
    \left(
      \bigotimes_{j=1}^k X_j
    \right)
    \left(
      \bigotimes_{j=1}^k M^{(j)}_{b_j|u_j}
    \right)
  \right]
  =
  \prod_{j=1}^k
  \operatorname{Tr}
  \left(
    X_jM^{(j)}_{b_j|u_j}
  \right).
\]
Therefore,
\begin{align}
  \mathcal M[X]
  &=
  \bigotimes_{j=1}^k
  \left[
    \sum_{u_j\in\mathcal U_j}
    q_{j,u_j}
    \sum_{b_j\in\mathcal B_{j,u_j}}
    M^{(j)}_{b_j|u_j}
    \operatorname{Tr}
    \left(
      X_jM^{(j)}_{b_j|u_j}
    \right)
  \right]
  \notag\\
  &=
  \bigotimes_{j=1}^k
  \mathcal M_j[X_j],
\end{align}
which implies \eqref{eq:general_local_product_channel_factorization}.

Since each \(\mathcal M_j\) is invertible on \(\mathcal L(\mathcal H_j)\), the
inverse of the tensor-product channel is
\(
  \mathcal M^{-1}
  =
  \bigotimes_{j=1}^k \mathcal M_j^{-1}.
\)
For the realized outcome, we have
\(
  M_{\hat b|\hat u}
  =
  \bigotimes_{j=1}^k
  M^{(j)}_{\hat b_j|\hat u_j}.
\)
Hence
\begin{align}
  \hat\rho
  &=
  \mathcal M^{-1}(M_{\hat b|\hat u})
  =
  \left(
    \bigotimes_{j=1}^k \mathcal M_j^{-1}
  \right)
  \left(
    \bigotimes_{j=1}^k
    M^{(j)}_{\hat b_j|\hat u_j}
  \right)\notag\\
  &=
  \bigotimes_{j=1}^k
  \mathcal M_j^{-1}
  \left(
    M^{(j)}_{\hat b_j|\hat u_j}
  \right)
  =
  \bigotimes_{j=1}^k\hat\rho_j,
\end{align}
which yields \eqref{eq:general_local_product_global_snapshot}.
\end{proof}

The factorization formula 
\eqref{eq:general_local_product_global_snapshot}
in Lemma~\ref{lem:general_local_product_snapshot_factorization}
is the general local product analogue of the local Pauli snapshot
formula.  We also assume throughout this section that the local snapshots are
normalized,
\begin{equation}
  \operatorname{Tr}(\hat\rho_j)=1
  \qquad
  \text{almost surely for every }j.
  \label{eq:general_local_product_trace_one_snapshot}
\end{equation}
This assumption ensures that identity tensor factors do not contribute to
the output coordinate.

We now specify the physical quantities whose expectation values are studied.
Let
\(
  O_1,\ldots,O_l
\)
be selected product observables of the form
\begin{equation}
  O_a
  =
  \bigotimes_{j=1}^k O_{a,j},
  \qquad
  a=1,\ldots,l,
  \label{eq:general_local_product_selected_observable}
\end{equation}
where \(O_{a,j}\) is a Hermitian operator on \(\mathcal H_j\), and
\(O_{a,j}=I_j\) for all but finitely many sites.  
Remember the definitions of 
the support and the weight as
\begin{equation}
  \supp(O_a)
  =
  \{j:\,O_{a,j}\neq I_j\},
  \quad
  \wt(O_a)=|\supp(O_a)|.
  \label{eq:general_local_product_support_weight}
\end{equation}
We assume that the selected observables have finite weight, and later we
will impose a uniform bound
\(
  \wt(O_a)\le w_0.
\)

For each selected observable, define the shadow-output coordinate
\begin{equation}
  \hat x_a
  :=
  \operatorname{Tr}(\hat\rho O_a),
  \qquad
  a=1,\ldots,l.
  \label{eq:general_local_product_shadow_output_coordinate}
\end{equation}
The selected output vector is
\begin{equation}
  R_l\hat x
  =
  (\hat x_1,\ldots,\hat x_l)^\top
  \in\mathbb R^l.
  \label{eq:general_local_product_selected_output_vector}
\end{equation}
Equivalently, in this section we may regard the selected observables
\(O_1,\ldots,O_l\) as the full observable family and take \(R_l=I_l\).

For a local operator \(A\) on \(\mathcal H_j\), define the local
reconstruction coefficient
\begin{equation}
  \eta_j(A;\hat u_j,\hat b_j)
  :=
  \operatorname{Tr}
  \left[
    \mathcal M_j^{-1}
    \left(
      M^{(j)}_{\hat b_j|\hat u_j}
    \right)
    A
  \right]
  =
  \operatorname{Tr}(\hat\rho_j A).
  \label{eq:general_local_product_eta_def}
\end{equation}
The following lemma is the general local product analogue of the local Pauli
coefficient rule.

\begin{lemma}[Local product shadow-output coefficient rule]
\label{lem:general_local_product_shadow_output_rule}
Under the local product-shadow assumptions above, for every selected product
observable \(O_a\) in
\eqref{eq:general_local_product_selected_observable}, the shadow-output
coordinate satisfies
\begin{equation}
  \hat x_a
  =
  \prod_{j\in\supp(O_a)}
  \eta_j(O_{a,j};\hat u_j,\hat b_j).
  \label{eq:general_local_product_shadow_output_rule}
\end{equation}
\end{lemma}

\begin{proof}
Using the standing product reconstruction assumption
\eqref{eq:general_local_product_global_snapshot} and the product form of
\(O_a\),
we obtain
\begin{align}
  \hat x_a
  &=
  \operatorname{Tr}(\hat\rho O_a)
=
  \operatorname{Tr}
  \left[
    \left(\bigotimes_{j=1}^k\hat\rho_j\right)
    \left(\bigotimes_{j=1}^k O_{a,j}\right)
  \right]
  \notag\\
  &=
  \prod_{j=1}^k
  \operatorname{Tr}(\hat\rho_j O_{a,j}).
  \label{eq:general_local_product_output_factorization}
\end{align}
If \(j\notin\supp(O_a)\), then \(O_{a,j}=I_j\), and the condition \eqref{eq:general_local_product_trace_one_snapshot} implies 
$\operatorname{Tr}(\hat\rho_j I_j)=1$.
Therefore only sites in \(\supp(O_a)\) contribute.  Using the definition
\eqref{eq:general_local_product_eta_def} of \(\eta_j\), we obtain
\[
  \hat x_a
  =
  \prod_{j\in\supp(O_a)}
  \operatorname{Tr}(\hat\rho_j O_{a,j})
  =
  \prod_{j\in\supp(O_a)}
  \eta_j(O_{a,j};\hat u_j,\hat b_j),
\]
which proves \eqref{eq:general_local_product_shadow_output_rule}.
\end{proof}

In the biased local Pauli protocol, the local coefficient
\(\eta_j\) reduces to
\(
  p_{j,P_j}^{-1}
  \mathbf 1_{\{\hat r_j=P_j\}}\hat s_j
\)
for a non-identity local Pauli \(P_j\).  Thus the Pauli compatibility
indicator is a special feature of the local Pauli measurement basis.  In the
general local product setting, the measurement and the local observable need
not belong to the same basis, and the corresponding dependence is absorbed
into the coefficient \(\eta_j\).

\subsection{Selected one-shot radius and centered radius bounds}
\label{subsec:general_local_product_radius_bounds}

We now bound the selected one-shot radius
\(  B_l =  \operatorname*{ess\,sup}\|R_l\hat x\|_2\)
and the upper bound 
\(\widetilde A_l(\rho)=B_l+\|R_lm\|_2\)
of the centered radius
\(A_l(\rho)\).
The key point is that each coordinate \(\hat x_a\) depends only on the local
snapshots on \(\supp(O_a)\).

For a local operator \(A\) on \(\mathcal H_j\), define the one-site
coefficient radius
\begin{equation}
  \Gamma_j(A)
  :=
  \operatorname*{ess\,sup}
  \left|
  \eta_j(A;\hat u_j,\hat b_j)
  \right|.
  \label{eq:general_local_product_gamma_def}
\end{equation}
For a selected product observable \(O_a\), set
\begin{equation}
  L_a
  :=
  \prod_{j\in\supp(O_a)}
  \Gamma_j(O_{a,j}).
  \label{eq:general_local_product_La_def}
\end{equation}
Then Lemma~\ref{lem:general_local_product_shadow_output_rule} immediately
gives
\begin{equation}
  |\hat x_a|\le L_a
  \qquad
  \text{almost surely}.
  \label{eq:general_local_product_coordinate_bound}
\end{equation}

\begin{proposition}[Selected radius bound for finite-weight product observables]
\label{prop:general_local_product_selected_radius_bound}
In the general local product-shadow setting above, the selected one-shot
radius satisfies
\begin{equation}
  B_l
  \le
  \left(
    \sum_{a=1}^l
    \prod_{j\in\supp(O_a)}
    \Gamma_j(O_{a,j})^2
  \right)^{1/2}.
  \label{eq:general_local_product_Bl_bound}
\end{equation}
Moreover, for every state \(\rho\),
\begin{equation}
\widetilde A_l(\rho)
  \le
  \left(
    \sum_{a=1}^l
    \prod_{j\in\supp(O_a)}
    \Gamma_j(O_{a,j})^2
  \right)^{1/2}
  +
  \left(
    \sum_{a=1}^l
    \|O_a\|_{\mathrm{op}}^2
  \right)^{1/2}.
  \label{eq:general_local_product_Al_bound}
\end{equation}
\end{proposition}

\begin{proof}
The relation \eqref{eq:general_local_product_coordinate_bound}
yields $  |\hat x_a|^2\le L_a^2
  =
  \prod_{j\in\supp(O_a)}
  \Gamma_j(O_{a,j})^2$, which implies
\begin{align}
  \|R_l\hat x\|_2^2
  =
  \sum_{a=1}^l |\hat x_a|^2
  \le
  \sum_{a=1}^l
  \prod_{j\in\supp(O_a)}
  \Gamma_j(O_{a,j})^2
  \label{eq:general_local_product_selected_norm_bound}
\end{align}
almost surely.  Taking the essential supremum proves
\eqref{eq:general_local_product_Bl_bound}.

For the centered radius, recall that
\( \widetilde A_l(\rho)=B_l+\|R_lm\|_2.
\)
For each coordinate, we have
\(
  m_a
  =
  \operatorname{Tr}(\rho O_a)
\), which implies
\(
  |m_a|
  \le
  \|O_a\|_{\mathrm{op}}.
\)
Thus, we have
\[
  \|R_lm\|_2
  =
  \left(
    \sum_{a=1}^l |m_a|^2
  \right)^{1/2}
  \le
  \left(
    \sum_{a=1}^l \|O_a\|_{\mathrm{op}}^2
  \right)^{1/2}.
\]
Combining this with \eqref{eq:general_local_product_Bl_bound} proves
\eqref{eq:general_local_product_Al_bound}.
\end{proof}

The preceding proposition gives a radius bound in terms of the actual local
reconstruction coefficients appearing in the selected observables.  The next
corollary records a simple uniform version.

\begin{corollary}[Dimension-independent radius under bounded local reconstruction]
\label{cor:general_local_product_dimension_independent_radius}
Assume that there are constants \(l_0,w_0,\Gamma_0,M_0\), independent of the
number \(k\) of tensor factors, such that
\begin{equation}
  l\le l_0,
  \qquad
  \wt(O_a)\le w_0,
  \qquad
  \Gamma_j(O_{a,j})\le \Gamma_0
  \label{eq:general_local_product_uniform_gamma_assumption}
\end{equation}
for all $a,j\in\supp(O_a)$,
and
\begin{equation}
  \|O_a\|_{\mathrm{op}}\le M_0,
  \qquad
  a=1,\ldots,l.
  \label{eq:general_local_product_uniform_operator_assumption}
\end{equation}
Then
\begin{equation}
  B_l
  \le
  \sqrt{l_0}\,\Gamma_0^{w_0},
  \label{eq:general_local_product_uniform_Bl_bound}
\end{equation}
and
\begin{equation}
\widetilde A_l(\rho)
  \le
  \sqrt{l_0}\,
  \bigl(
    \Gamma_0^{w_0}+M_0
  \bigr).
  \label{eq:general_local_product_uniform_Al_bound}
\end{equation}
In particular, \(B_l\) and \(\widetilde A_l(\rho)\) are bounded independently of \(k\).
\end{corollary}

\begin{proof}
For each \(a\), the assumptions
\(\wt(O_a)\le w_0\) and
\(\Gamma_j(O_{a,j})\le\Gamma_0\) imply
\(
  \prod_{j\in\supp(O_a)}
  \Gamma_j(O_{a,j})^2
  \le
  \Gamma_0^{2w_0},
\)
which yields
\[
  \sum_{a=1}^l
  \prod_{j\in\supp(O_a)}
  \Gamma_j(O_{a,j})^2
  \le
  l_0\Gamma_0^{2w_0}.
\]
Substituting this into
\eqref{eq:general_local_product_Bl_bound} proves
\eqref{eq:general_local_product_uniform_Bl_bound}.

Similarly, we have
\[
  \left(
    \sum_{a=1}^l
    \|O_a\|_{\mathrm{op}}^2
  \right)^{1/2}
  \le
  \sqrt{l_0}\,M_0.
\]
Combining this with
\eqref{eq:general_local_product_uniform_Bl_bound} and
\eqref{eq:general_local_product_Al_bound} proves
\eqref{eq:general_local_product_uniform_Al_bound}.
\end{proof}

A more concrete sufficient condition can be stated in terms of the
operator-norm size of the local reconstructed snapshots.  Suppose that
\begin{equation}
  \|\hat\rho_j\|_{\mathrm{op}}\le G_0
  \qquad
  \text{almost surely for all }j,
  \label{eq:general_local_product_local_snapshot_op_bound}
\end{equation}
and that the non-identity local factors of the selected observables satisfy
\begin{equation}
  \|O_{a,j}\|_{\mathrm{op}}\le 1
  \qquad
  \text{for all }a,\ j\in\supp(O_a).
  \label{eq:general_local_product_local_observable_op_bound}
\end{equation}
Since \(\dim\mathcal H_j=t\), we have
\[
  \|O_{a,j}\|_1\le t\|O_{a,j}\|_{\mathrm{op}}\le t.
\]
Therefore,
\begin{align}
  \Gamma_j(O_{a,j})
  &=
  \operatorname*{ess\,sup}
  |\operatorname{Tr}(\hat\rho_jO_{a,j})|
  \notag\\
  &\le
  \operatorname*{ess\,sup}
  \|\hat\rho_j\|_{\mathrm{op}}\,
  \|O_{a,j}\|_1
 \le
  G_0t .
  \label{eq:general_local_product_G0t_gamma_bound}
\end{align}
Thus, under the finite-weight and bounded-size assumptions of
Corollary~\ref{cor:general_local_product_dimension_independent_radius}, 
the choice 
\(
  \Gamma_0=G_0t
\)
implies 
\begin{align}
  B_l
&  \le
  \sqrt{l_0}\,(G_0t)^{w_0},
  \label{eq:general_local_product_G0t_Bl_bound} \\
\widetilde A_l(\rho)
&  \le
  \sqrt{l_0}\,
  \bigl(
    (G_0t)^{w_0}+M_0
  \bigr).
  \label{eq:general_local_product_G0t_Al_bound}
\end{align}

\begin{remark}[General local observables on a finite support]
\label{rem:general_local_product_nonproduct_observables}
The product-observable assumption is convenient because it gives the
coefficient factorization
\eqref{eq:general_local_product_shadow_output_rule}.  The same bounded-radius
principle also applies to a general observable supported on a finite set
\(S\subset\{1,\ldots,k\}\).  If
\(
  O=O_S\otimes I_{S^c},
\)
then the relation
\(
  \hat x_O
  =
  \operatorname{Tr}
  \left[
    \left(
      \bigotimes_{j\in S}\hat\rho_j
    \right)
    O_S
  \right]
\) holds.
If \(\|\hat\rho_j\|_{\mathrm{op}}\le G_0\) for \(j\in S\), then
the relation $
  |\hat x_O|
  \le
  G_0^{|S|}\|O_S\|_1$ holds.
Thus the same conclusion holds: the radius depends on the size of the
support and on local reconstruction bounds, but not on the total number
\(k\) of tensor factors.
\end{remark}

Consequently, for finite-weight selected observables under a local product
shadow protocol with uniformly bounded local reconstruction coefficients,
the bounded-radius hypothesis in
Theorem~\ref{thm:constant_selected_error_bounded_radius} holds with constants
independent of \(k\).  The selected sample-centered covariance can therefore
be estimated to constant operator-norm accuracy with a sample size independent
of the total number of tensor factors.

\section{Discussion}
\label{sec:discussion}
This paper has developed a covariance-matrix viewpoint for classical-shadow outputs.  Instead of organizing the analysis only around individual observables, marginal variances, or many-observable prediction bounds, we studied selected covariance matrices of reconstructed shadow-output vectors.  The main finite-sample theorem applies to fixed selected coordinates for arbitrary shadow protocols and gives an operator-norm error bound for the selected sample-centered empirical covariance.  The bound contains protocol-dependent constants that measure how large the selected reconstructed vector can be in a single measurement round.  When these constants and the number of selected coordinates are independent of the ambient system size, the required sample size is also independent of the ambient dimension.

The main structural message is that such protocol-dependent constants can be controlled in local measurement settings.  For general local product shadows, finite-weight product observables are controlled by their support sizes and local reconstruction coefficients, rather than by the total number of tensor factors.  Biased local Pauli shadows provide a fully explicit instance: the relevant constants are written directly in terms of selected Pauli supports and local basis-selection probabilities.  A comparison with global Clifford shadows shows that the same local behavior should not be expected without additional protocol structure.  The exact covariance formula for biased local Pauli shadows gives further structural information, but it is separate from the finite-sample estimation argument.

\subsection{What the covariance viewpoint adds}
\label{subsec:discussion_covariance_viewpoint}

The main conceptual point is that the covariance matrix
\(\Sigma(\rho)=\operatorname{Cov}_\rho(\hat x)\) contains information that is
not visible from marginal coordinate variances alone.  Its diagonal entries
describe the variances of individual reconstructed output coordinates, while
its off-diagonal entries describe statistical couplings between different
coordinates produced by the same shadow protocol.  These off-diagonal
entries are covariances of random post-processed shadow outputs, and
therefore describe the joint fluctuation structure of the
measurement-and-reconstruction procedure.  In the Pauli specialization, this
statistical covariance should not be confused with physical Pauli product
correlations.

This distinction matters whenever one studies more than one shadow-output
coordinate at a time.  A multi-term observable, or a finite family of
observables, does not depend only on the list of marginal variances.  Its
estimation variance also depends on how the selected output coordinates
fluctuate together.  The covariance matrix is the object that records this
collective statistical structure.

The selected covariance matrix
\(\Sigma_l(\rho)=R_l\Sigma(\rho)R_l^\top\) gives a natural finite-dimensional
object associated with a fixed selected coordinate set.  Its eigenvalues
describe the largest and smallest fluctuation scales among normalized linear
combinations of the selected coordinates.  Its eigenvectors identify the
corresponding principal fluctuation directions inside the selected coordinate
set.  Thus selected covariance spectra provide a diagnostic that is
different from, but complementary to, observable-wise shadow-norm guarantees.

The covariance viewpoint also connects with reconstruction error whenever the chosen output coordinates form a coordinate representation of the reconstructed object. In such settings, traces or selected traces of \(\Sigma(\rho)\) quantify total fluctuation within the chosen coordinate system. More generally, the same covariance matrix organizes single-coordinate variances, off-diagonal statistical couplings, selected spectra, and fluctuations of selected reconstructed coordinates.

\subsection{Why the finite-sample theory is selected}
\label{subsec:discussion_selected_finite_sample}

The finite-sample theory in this paper is deliberately formulated for fixed
selected coordinate sets.  This is not only a technical convenience, but also
the natural scale at which one can expect useful operator-norm guarantees.
The full shadow-output covariance matrix may have a very large ambient
dimension \(p\), and estimating the entire matrix in operator norm would be a
much stronger and substantially different task.  Instead, we fix a selected
coordinate set \(J_l=\{\alpha_1,\ldots,\alpha_l\}\) and study the compressed
covariance \(\Sigma_l(\rho)=R_l\Sigma(\rho)R_l^\top\).

This selected formulation has two advantages.  First, it matches the
statistical question of interest when only a prescribed family of output
coordinates is relevant.  In that case, the eigenvalues of
\(\Sigma_l(\rho)\) describe the fluctuation scales of linear combinations
inside that selected family, and the corresponding spectral projectors
identify stable or unstable selected directions.  Second, the finite-sample
concentration is performed after the compression by \(R_l\).  Hence the
dimension entering the probability bound is the selected dimension \(l\),
rather than the ambient output dimension \(p\).

This is why the fixed-selection assumption is important.  The selection
matrix \(R_l\) is chosen independently of the data.  Under this assumption,
one may treat the selected covariance as an ordinary \(l\times l\) covariance
matrix and apply matrix concentration directly in the selected space.  If the
coordinate set were chosen after looking at the data, additional arguments
would be needed, such as uniform concentration over a larger class of
possible selections, sample splitting, or other post-selection controls.
Those questions are outside the scope of the present paper.

The selected viewpoint should therefore be read as a finite-sample localization of covariance estimation. It does not claim to reconstruct the entire covariance structure of the full shadow-output space. Rather, it shows that once a coordinate family has been fixed, the covariance spectrum within that family can be estimated with explicit nonasymptotic guarantees. This is the covariance analogue of focusing on a fixed family of observables in ordinary classical-shadow prediction, but with the empirical target changed from expectation values to the selected covariance matrix itself.

\subsection{Interpretation of the constant selected-error regime}
\label{subsec:discussion_constant_selected_error}

The dimension-independent selected regime is the main finite-sample consequence of the selected formulation. The general theorem says that, for any chosen shadow protocol, selected covariance estimation has a sample complexity controlled by the selected dimension and by protocol-dependent constants in the error bound. If these constants remain bounded independently of the ambient system size, then constant operator-norm accuracy for the selected covariance requires a number of samples independent of the ambient dimension.

General local product shadows provide one mechanism for this behavior. When the selected observables have uniformly bounded weight and the relevant local reconstruction coefficients are uniformly bounded, the constants in the finite-sample bound are controlled independently of the number of tensor factors. This explains why dimension-independent selected covariance estimation is not tied to Pauli measurements: it follows from local product structure, finite support, and bounded local reconstruction.

The biased local Pauli protocol gives a fully explicit instance of this mechanism. If the selected set size is bounded by \(l_0\), the selected Pauli weights are bounded by \(w_0\), and the local basis probabilities satisfy \(p_{j,r}\ge p_0>0\), then the finite-sample bound depends only on \(l_0\), \(w_0\), and \(p_0\), not on the total number of qubits. Consequently, the sample size needed to estimate the selected sample-centered covariance to constant operator-norm accuracy can be chosen independently of the total number of qubits.

The comparison with global Clifford shadows separates two issues. The finite-sample theorem is a general statement about selected covariance estimation, but whether its constants remain independent of system size is a protocol-specific question. Local product protocols, including biased local Pauli shadows, give positive examples under finite-support and bounded-coefficient assumptions. Global Clifford shadows show that such behavior should not be expected without additional local or protocol-specific structure.

\subsection{Local Pauli covariance formulas as structural information}
\label{subsec:discussion_local_pauli_covariance_formula}

The exact local Pauli covariance formula should be read as structural
information that complements the finite-sample results.  The finite-sample
selected covariance theorem does not require an entrywise formula for
\(\Sigma(\rho)\).  It only requires the protocol-dependent constants in the
concentration bound to be controlled.  
Nevertheless, in the biased local Pauli protocol, the covariance entries can be computed explicitly in
terms of Pauli expectations of the state and the local basis-selection
probabilities.

The compatibility relation \(P\sim Q\), the cancellation string
\(P\ominus Q\), and the overlap factor
\(\beta(P,Q)=\prod_{j\in\supp(P)\cap\supp(Q)}p_{j,P_j}^{-1}\) describe how
the local Pauli measurement design determines the raw second moments of the
reconstructed Pauli coefficients.  In particular, Pauli compatibility
determines which off-diagonal raw second moments can be nonzero, while
inverse-probability overlap factors determine the size of the corresponding
raw second-moment contributions.  The covariance entries are obtained from
these raw second moments by subtracting the products of the corresponding
means.

This formula is useful because it shows how local basis bias changes the
covariance matrix of the reconstructed Pauli coefficient vector.  The bias
does not merely rescale marginal variances.  It also changes off-diagonal
statistical couplings through the same inverse-probability mechanism.  Thus
the biased local Pauli protocol provides an explicit model in which
the covariance matrix can be analyzed entry by entry, while the finite-sample
theory explains how selected parts of that covariance matrix can be estimated
from data.

This analytic formula is also a reminder that selected covariance estimation
and closed-form covariance calculation are different tasks.  The selected
finite-sample theorem applies even when no explicit formula for
\(\Sigma(\rho)\) is available, provided the protocol-dependent constants in
the error bound can be controlled.  Conversely, an explicit covariance
formula can reveal detailed structure, but it does not by itself give
finite-sample operator-norm recovery.  In the local Pauli setting both
ingredients are available: the finite-sample bound gives the estimation
guarantee, and the exact covariance formula explains the covariance structure
being estimated.

\subsection{Scope and outlook}
\label{subsec:discussion_scope_outlook}

The results of this paper should be read as selected covariance results.
They do not assert operator-norm recovery of the full ambient covariance
matrix \(\Sigma(\rho)\), and the selected coordinate set is assumed to be
fixed independently of the data.  Adaptive or post-selected coordinate
choices would require additional tools, such as sample splitting or uniform
concentration over candidate selections.

The dimension-independent regime also depends on protocol-specific constants.
The finite-sample theorem applies to arbitrary shadow protocols, but it does
not guarantee that these constants are small.  General local product shadows
and biased local Pauli shadows provide positive examples under finite-support
and bounded-coefficient assumptions, whereas the global Clifford comparison
shows that the corresponding constants can grow with the Hilbert-space
dimension.

Finally, the exact covariance formula derived for biased local Pauli shadows
is a protocol-specific structural result.  It is separate from the
finite-sample theorem, which only requires the relevant finite-sample
constants to be controlled.  Natural future directions include sharper bounds
for structured selected sets, adaptive selected covariance estimation,
extensions to other shadow protocols, and robust covariance estimation under
noise, drift, or corrupted samples.

\appendix

\section{Proof of \eqref{eq:exact_rank_one_relation_true_sample_centered_general}}
\label{AA1}
Indeed, using
\(
  \hat x_i-\bar{\hat x}_N
  =
  (\hat x_i-m)-(\bar{\hat x}_N-m),
\)
we expand the sample-centered covariance as
\begin{align}
&  \widehat\Sigma_N^{\mathrm{sc}}
  \notag\\
  =&
  \frac1N
  \sum_{i=1}^N
  (\hat x_i-\bar{\hat x}_N)(\hat x_i-\bar{\hat x}_N)^\top
  \notag\\
  =&
  \frac1N
  \sum_{i=1}^N
  \bigl((\hat x_i-m)-(\bar{\hat x}_N-m)\bigr)
  \bigl((\hat x_i-m)-(\bar{\hat x}_N-m)\bigr)^\top .
\end{align}
Expanding the product gives
\begin{align}
  \widehat\Sigma_N^{\mathrm{sc}}
  &=
  \frac1N
  \sum_{i=1}^N
  (\hat x_i-m)(\hat x_i-m)^\top
  \notag\\
  &\quad
  -
  \frac1N
  \sum_{i=1}^N
  (\hat x_i-m)(\bar{\hat x}_N-m)^\top
  \notag\\
  &\quad
  -
  \frac1N
  \sum_{i=1}^N
  (\bar{\hat x}_N-m)(\hat x_i-m)^\top
  \notag\\
  &\quad
  +
  \frac1N
  \sum_{i=1}^N
  (\bar{\hat x}_N-m)(\bar{\hat x}_N-m)^\top .
\end{align}
Since
\(
  \frac1N\sum_{i=1}^N(\hat x_i-m)
  =
  \bar{\hat x}_N-m,
\)
the two cross terms become
\(
  -
  (\bar{\hat x}_N-m)(\bar{\hat x}_N-m)^\top.
\)
The last term is
\[
  \frac1N
  \sum_{i=1}^N
  (\bar{\hat x}_N-m)(\bar{\hat x}_N-m)^\top
  =
  (\bar{\hat x}_N-m)(\bar{\hat x}_N-m)^\top .
\]
Therefore, we obtain
\begin{align}
  \widehat\Sigma_N^{\mathrm{sc}}
  &=
  \frac1N
  \sum_{i=1}^N
  (\hat x_i-m)(\hat x_i-m)^\top
  -
  (\bar{\hat x}_N-m)(\bar{\hat x}_N-m)^\top
  \notag\\
  &=
  \widehat\Sigma_N^{\mathrm{tc}}
  -
  (\bar{\hat x}_N-m)(\bar{\hat x}_N-m)^\top ,
\end{align}
which implies
\(
  \widehat\Sigma_N^{\mathrm{sc}}
  -
  \widehat\Sigma_N^{\mathrm{tc}}
  =
  -
  (\bar{\hat x}_N-m)(\bar{\hat x}_N-m)^\top .
\)

\section{Proof of \eqref{KJ13}}\label{KJ13P}
We briefly derive the reconstruction formula.  For a single qubit at site
\(j\), the biased Pauli measurement channel acts on an operator \(A\)
as
\begin{align*}
  \mathcal M_j[A]
&  =
  \sum_{r\in\{X,Y,Z\}} p_{j,r}
  \sum_{s=\pm1}
  M^{(j)}_{s|r}\,
  \operatorname{Tr}\!\left(A M^{(j)}_{s|r}\right),\\
  M^{(j)}_{s|r}
  &=
  \frac12(I+s r).
\end{align*}
Since
$\operatorname{Tr}(I M^{(j)}_{s|r})=1$, and 
$\operatorname{Tr}(t M^{(j)}_{s|r})=s\,\mathbf 1_{\{t=r\}}$
for $t\in\{X,Y,Z\}$,
we have
\begin{equation}
  \mathcal M_j[I]=I,
  \qquad
  \mathcal M_j[t]=p_{j,t}\,t,
  \quad t\in\{X,Y,Z\}.
  \label{eq:single_qubit_biased_pauli_channel_diagonal_action}
\end{equation}
Thus \(\mathcal M_j\) is diagonal in the single-qubit Pauli basis, and its
inverse is
\begin{equation}
  \mathcal M_j^{-1}[I]=I,
  \qquad
  \mathcal M_j^{-1}[t]=p_{j,t}^{-1}t,
  \quad t\in\{X,Y,Z\}.
  \label{eq:single_qubit_biased_pauli_inverse_action}
\end{equation}
Applying this inverse map to the actually observed rank-one projector
\(M^{(j)}_{\hat s_j|\hat r_j}=\frac12(I+\hat s_j\hat r_j)\) gives the
single-qubit reconstructed snapshot
\begin{equation}
  \hat\rho_j^{(1)}
  =
  \mathcal M_j^{-1}\!\left[
    M^{(j)}_{\hat s_j|\hat r_j}
  \right]
  =
  \frac12
  \left(
    I+\frac{\hat s_j}{p_{j,\hat r_j}}\hat r_j
  \right).
  \label{eq:single_qubit_biased_pauli_reconstructed_snapshot}
\end{equation}
Because both the basis choice and the POVM are product over sites, the
\(k\)-qubit shadow channel is the tensor product
\(\mathcal M=\mathcal M_1\otimes\cdots\otimes\mathcal M_k\), and hence its
inverse is
\(\mathcal M^{-1}=\mathcal M_1^{-1}\otimes\cdots\otimes\mathcal M_k^{-1}\).
Therefore the full reconstructed shadow snapshot is
\begin{equation}
  \hat\rho
  =
  \bigotimes_{j=1}^k
  \hat\rho_j^{(1)}
  =
  \bigotimes_{j=1}^k
  \frac12
  \left(
    I+\frac{\hat s_j}{p_{j,\hat r_j}}\hat r_j
  \right),
  \label{eq:biased_local_pauli_reconstructed_snapshot}
\end{equation}
which yields \eqref{KJ13}.

\section{Proof of Theorem~\ref{thm:biased_local_pauli_exact_covariance}}
\label{sec:proof_of_theorem_ref_thm_biased_local_pauli_exact_covariance}

This appendix extends the exact local-Pauli covariance calculation 
(Proposition \ref{prop:local_pauli_exact_covariance}) by Huang--Kueng--Preskill~\cite[Lemma~4]{HuangKuengPreskill2020}
to the biased measurement design of Section~\ref{sec:local_pauli_exact_covariance}. The proof follows the same overall logic as in the uniform case, but with the deterministic coefficient rule now carrying explicit inverse-probability weights. We first derive the weighted single-shot coefficient formula, then compute the corresponding raw second moments for compatible and incompatible pairs, and finally obtain the exact covariance expression in terms of the overlap bias factor. The result is the closed-form covariance formula for the biased local design stated in Theorem~\ref{thm:biased_local_pauli_exact_covariance}.

\label{app:proof_biased_local_pauli_exact_covariance}
In this appendix we prove Theorem~\ref{thm:biased_local_pauli_exact_covariance}.
Throughout, we work with the $k$-qubit biased local Pauli classical-shadow protocol described in Section~\ref{sec:local_pauli_exact_covariance}.

\begin{lemma}[Incompatible strings have zero raw second moment]
\label{lem:biased_local_pauli_incompatible_zero_second_moment}
Let $P,Q\in\mathcal P_k^\times$. If $P
\not\sim Q$, then $\mathbb E[\hat x_P\hat x_Q]=0$.
\end{lemma}
\begin{proof}
If $P
\not\sim Q$, then there exists a site $j$ such that $P_j
\neq I$, $Q_j
\neq I$, and $P_j
\neq Q_j$. By Lemma~\ref{lem:biased_local_pauli_coefficient_rule}, the event $\hat x_P
\neq 0$ requires 
$\mathbf 1_{\{\hat r_{j'}=P_{j'}\ \forall j'\in\supp(P)\}}
\neq 0$, i.e., 
$\hat r_j=P_j$ 
for any $j \in \supp(P)$.
Also, the event $\hat x_Q
\neq 0$ requires $\hat r_{j'}=Q_{j'}$
for any $j' \in \supp(Q)$.
These conditions are incompatible. Hence,
we have $\hat x_P=0$ or $\hat x_Q=0$, i.e.,
$\hat x_P\hat x_Q=0$ with probability one.
\end{proof}

\begin{lemma}[Unique compatible basis assignment]
\label{lem:biased_local_pauli_unique_compatible_basis_assignment}
Let $P,Q\in\mathcal P_k^\times$ satisfy $P\sim Q$. 
When both $\hat x_P$ and $\hat x_Q$ are nonzero,
there is a unique local measurement basis 
$\hat r_j$ for any $j \in \supp(P)\cup\supp(Q)$, 
namely
\begin{align}
\hat r_j=\begin{cases}
P_j & \hbox{when } j\in\supp(P)\setminus\supp(Q),\\
Q_j &\hbox{when } j\in\supp(Q)\setminus\supp(P),\\
P_j(=Q_j) &\hbox{when } j\in\supp(P)\cap\supp(Q).
\end{cases}
\label{ANM}
\end{align}
\end{lemma}
\begin{proof}
By Lemma~\ref{lem:biased_local_pauli_coefficient_rule}, $\hat x_P
\neq 0$ requires $\hat r_j=P_j$ for every $j\in\supp(P)$, while $\hat x_Q
\neq 0$ requires $\hat r_j=Q_j$ for every $j\in\supp(Q)$. Compatibility guarantees that these prescriptions agree on the overlap and therefore determine a unique common basis pattern on $\supp(P)\cup\supp(Q)$.
\end{proof}

\begin{lemma}[Support and sitewise form of the cancellation string]
\label{lem:biased_local_pauli_cancellation_string_support}
Let $P,Q\in\mathcal P_k^\times$ satisfy $P\sim Q$, and define $R:=P\ominus Q$. Then we have the following sitewise relations
\[
R_j=\begin{cases}
I & \hbox{when }P_j=Q_j,\\
P_j & \hbox{when }Q_j=I,\\
Q_j & \hbox{when }P_j=I.
\end{cases}
\]
\end{lemma}
\begin{proof}
The definition of $P\ominus Q$ cancels the common non-identity factors and retains the unique non-identity factor on sites where only one of $P$ or $Q$ is $I$.
\end{proof}

\begin{lemma}[Conditional sign-product identity]
\label{lem:biased_local_pauli_conditional_sign_product_identity}
Let $P,Q\in\mathcal P_k^\times$ satisfy $P\sim Q$, and let $R:=P\ominus Q$. 
We assume that $\hat r_j$ is the unique compatible basis assignment of Lemma~\ref{lem:biased_local_pauli_unique_compatible_basis_assignment}
for any $j \in \supp(P)\cup\supp(Q)$. 
Then, we have
\begin{align}
\mathbb E\!\left[\prod_{j\in\supp(P\ominus Q)} \hat s_j\ \middle|\ \hat x_P\hat x_Q
\neq 0\right]=\Tr(\rho R)=m_R.
\label{BNE}
\end{align}
\end{lemma}
\begin{proof}
Due to Lemma~\ref{lem:biased_local_pauli_unique_compatible_basis_assignment},
the condition $\hat x_P\hat x_Q\neq 0$ implies 
that the measurement basis $\hat r_j$
for $j \in \supp(P)\cup\supp(Q)$ is given as \eqref{ANM}.
By Lemma~\ref{lem:biased_local_pauli_cancellation_string_support}, the operator $R=P\ominus Q$ is supported on $\supp(P\ominus Q)$.
Hence the product of outcomes on this support is exactly the same as the commuting Pauli observable $R$, which implies 
\eqref{BNE}.
\end{proof}

\begin{lemma}[Probability of the compatible basis event]
\label{lem:biased_local_pauli_compatible_basis_event_probability}
Let $P,Q\in\mathcal P_k^\times$ satisfy $P\sim Q$. 
Then the probability of the event of 
$\hat x_P\hat x_Q
\neq 0$ is
\begin{align*}
&\Big(\prod_{j\in\supp(P)\setminus\supp(Q)} p_{j,P_j}\Big)\cdot\Big(\prod_{j\in\supp(Q)\setminus\supp(P)} p_{j,Q_j}\Big)
\notag\\
&\cdot\Big(\prod_{j\in\supp(P)\cap\supp(Q)} p_{j,P_j}
\Big).
\end{align*}
\end{lemma}
\begin{proof}
The event of 
$\hat x_P\hat x_Q
\neq 0$ does not depend on $r_j$ for $ j \in 
(\supp(P)\cup\supp(Q))^c$.
The set $\supp(P)\cup\supp(Q)$ is divided into 
$\supp(P)\setminus\supp(Q)$,
$\supp(Q)\setminus\supp(P)$,
and $\supp(P)\cap\supp(Q)$.
The local basis selections are independent across sites.
Hence, we obtain the desired relation.
\end{proof}

\begin{proof}[Proof of Theorem~\ref{thm:biased_local_pauli_exact_covariance}]
If $P
\not\sim Q$, Lemma~\ref{lem:biased_local_pauli_incompatible_zero_second_moment} gives $\mathbb E[\hat x_P\hat x_Q]=0$. 

Assume now that $P\sim Q$. 
We assume that $\hat r_j$ is the unique compatible basis assignment of Lemma~\ref{lem:biased_local_pauli_unique_compatible_basis_assignment}
for any $j \in \supp(P)\cup\supp(Q)$. 
On that event, Lemma~\ref{lem:biased_local_pauli_coefficient_rule} yields
\begin{align}
\hat x_P\hat x_Q=
\Bigl(\prod_{j\in\supp(P)} p_{j,P_j}^{-1}\Bigr)\Bigl(\prod_{j\in\supp(Q)} p_{j,Q_j}^{-1}\Bigr)\prod_{j\in\supp(P\ominus Q)} \hat s_j,
\label{BSHJ1}
\end{align}
because the overlap sign factors square to one. With $R:=P\ominus Q$, Lemma~\ref{lem:biased_local_pauli_conditional_sign_product_identity} gives
\begin{align}
\mathbb E\!\left[\prod_{j\in\supp(P\ominus Q)} \hat s_j\ \middle|\ \hat x_P\hat x_Q
\neq 0\right]=m_R=m_{P\ominus Q}.
\label{BSHJ2}
\end{align}
Thus,
\begin{align}
&\mathbb E[\hat x_P\hat x_Q]
=
\mathbb P\!\left(
\hat x_P\hat x_Q \neq 0 \right)
\mathbb E\!\left[
\hat x_P\hat x_Q 
 \middle|\ \hat x_P\hat x_Q
\neq 0\right] \notag\\
\stackrel{(a)}{=}&
\Big(\prod_{j\in\supp(P)\setminus\supp(Q)} p_{j,P_j}\Big)\cdot\Big(\prod_{j\in\supp(Q)\setminus\supp(P)} p_{j,Q_j}\Big)\notag\\&\cdot\Big(\prod_{j\in\supp(P)\cap\supp(Q)} p_{j,P_j}
\Big)\notag\\
&\cdot \Bigl(\prod_{j\in\supp(P)} p_{j,P_j}^{-1}\Bigr)\Bigl(\prod_{j\in\supp(Q)} p_{j,Q_j}^{-1}\Bigr)\notag\\
&\cdot\mathbb E\!\left[\prod_{j\in\supp(P\ominus Q)} \hat s_j\ \middle|\ \hat x_P\hat x_Q
\neq 0\right]\notag\\
\stackrel{(b)}{=}&
\Bigl(\prod_{j\in\supp(P)\cap\supp(Q)} p_{j,P_j}^{-1}\Bigr)m_{P\ominus Q}
\stackrel{(c)}{=}
\beta(P,Q)m_{P\ominus Q}.
\end{align}
Here,
$(a)$ follows from 
Lemma \ref{lem:biased_local_pauli_compatible_basis_event_probability}
and \eqref{BSHJ1}.
$(b)$ follows from \eqref{BSHJ2}.
$(c)$ follows from \eqref{BSHJ3}.
We obtain \eqref{FKS1}.
Subtracting $m_Pm_Q$ gives the covariance formula, and taking $Q=P$ yields the stated diagonal formula because $P\ominus P=I^{\otimes k}$ and $m_{I^{\otimes k}}=1$.
Thus, we obtain \eqref{FKS2}.
\end{proof}

\section*{Acknowledgements}
Large language model tools were used as auxiliary aids in preparing this
manuscript, including assistance with exposition, literature search, and
exploratory discussions of possible approaches to the research problem.
The manuscript was written under the author's
direction, and the author is solely responsible for all mathematical content,
proofs, references, and conclusions.

\section*{Ethics approval and consent to participate}
Not applicable.

\section*{Consent for publication}
Not applicable.

\section*{Availability of supporting data}
Data sharing is not applicable to this article as no datasets were generated
or analyzed during the current study.

\section*{Competing interests}
There are no competing interests.

\section*{Funding}
The author was supported in part by 
the General R\& D Projects of 1+1+1 CUHK-CUHK(SZ)-GDST Joint Collaboration Fund (Grant No.
GRDP2025-022) and 
the Guangdong Provincial Quantum Science Strategic Initiative (Grant No. 
GDZX2505003).

\bibliography{refs_used}
\end{document}